\documentclass[a4paper,UKenglish]{lipics}
 
\usepackage{microtype}

\title{Trust, but Verify: Two-Phase Typing for Dynamic Languages\footnote{This work was supported by NSF Grants CNS-1223850, CNS-0964702 and gifts from Microsoft Research.}}
\titlerunning{Trust, but Verify: Two-Phase Typing for Dynamic Languages}

\author{Panagiotis Vekris}
\author{Benjamin Cosman}
\author{Ranjit Jhala}

\affil{Univeristy of California, San Diego\\  La Jolla, CA, 92093, USA\\ \{\texttt{pvekris},\texttt{blcosman},\texttt{rjhala}\}\texttt{@cs.ucsd.edu}}

\authorrunning{P. Vekris, B. Cosman and R. Jhala} 

\Copyright{Panagiotis Vekris, Benjamin Cosman, Ranjit Jhala}

\usepackage{inconsolata}
\usepackage{url}                  
\usepackage{listings}          
\usepackage{enumitem}      
\hypersetup{colorlinks=true,allcolors=blue}
\usepackage{amsmath, amssymb, latexsym}
\usepackage{stmaryrd}
\usepackage{xcolor}
\usepackage{mathtools}
\usepackage{apptools}
\usepackage{chngcntr}
\usepackage{scalerel}
\usepackage{xspace}
\usepackage{mathpartir}
\usepackage{graphicx}
\usepackage{caption}
\usepackage{csvsimple}
\usepackage{pdfpages}
\usepackage{tabu}

\definecolor{haskellblue}{rgb}{0.0, 0.0, 1.0}
\definecolor{haskellblue}{rgb}{1.0, 0.0, 0.0}
\definecolor{gray_ulisses}{gray}{0.55}
\definecolor{castanho_ulisses}{rgb}{0.71,0.33,0.14}
\definecolor{preto_ulisses}{rgb}{0.41,0.20,0.04}
\definecolor{green_ulisses}{rgb}{0.0,0.4,0.0}
\definecolor{mygreen}{rgb}{0,0.6,0}
\definecolor{mygray}{rgb}{0.5,0.5,0.5}
\definecolor{mymauve}{rgb}{0.58,0,0.82}
\definecolor{darkgray}{gray}{0.30}
\definecolor{darkblue}{rgb}{0,0,0.75}
\definecolor{darkgreen}{rgb}{0,0.4,0}

\lstdefinelanguage{RefScript}{
  basicstyle=\ttfamily\footnotesize,
  backgroundcolor=\color{white},  
  sensitive=true,
  keywords={typeof, instanceof, new, true, false, catch, function, return, 
            null, catch, switch, var, if, in, for, while, do, 
            else, case, break, forall, type, true, false,  
            constructor, super, public, private, class, 
            interface, extends, implements, module, 
            declare, this,let, then, undefined, fn,
            },
  keywordstyle=\color{mymauve},
  morecomment=[s]{/*}{*/},
  morecomment=[l]//,
  morestring=[b]",
  morestring=[b]',
  commentstyle=\color{mygreen},
  stringstyle=\color{darkgreen},
  showstringspaces=false,
  numberblanklines=true,
  showspaces=false,
  breaklines=true,
  showtabs=false,
  emph={[1] reduce, $reduce,minIndexFO,minIndexHO, minIndex, neg,
            neg\#1, neg\#2, step},
  emphstyle={[1]\color{darkblue}},
  emph={[2]     boolean, bool,tt,ff,any,void,number,string,nat,index,idx,Array},
  emphstyle={[2]\color{darkgreen}},
  literate= {=>}{$\Rightarrow$}{1}
            {>=}{{$\geq$}}1
            {<=}{{$\leq$}}1
            {\/\\}{{$\wedge$}}2
}

\lstnewenvironment{code}{\lstset{language=RefScript}}{}

\lstMakeShortInline[language=RefScript, basicstyle=\normalsize\ttfamily]@

\usepackage{commands}
\newtheorem{assumption}{Assumption}
\usepackage[english]{babel}

\subjclass{
D.3.3 [Programming Languages]: 
  Language Constructs and Features -- 
    Constraints,
    Polymorhpism;
F.3.1 [Logics and Meanings of Programs]:
  Specifying and Verifying and Reasoning about Programs -- 
    Assertions,
    Pre- and post-conditions;    
F.3.3 [Logics and Meanings of Programs]:
  Studies of Program Constructs --
    Type structure
}
\keywords{Dynamic Languages, Type Systems, Refinement Types, Intersection Types, Overloading}

\volumeinfo{John Tang Boyland}
{1}
{29th European Conference on Object-Oriented Programming (ECOOP'15)}
{37}
{1}
{999}
\EventShortName{ECOOP'15}
\DOI{10.4230/LIPIcs.ECOOP.2015.999}
\serieslogo{}

\begin{document}

\maketitle

\begin{abstract}
A key challenge when statically typing so-called dynamic languages 
is the ubiquity of \emph{value-based overloading}, where a given 
function can dynamically reflect upon and behave according to the 
types of its arguments. 
Thus, to establish basic types, the analysis must reason precisely
about values, but in the presence of higher-order functions and
polymorphism, this reasoning itself can require basic types.
In this paper we address this chicken-and-egg problem by introducing 
the framework of two-phased typing.
The first ``trust'' phase performs classical, i.e. flow-, path- and
value-insensitive type checking to assign basic types to various
program expressions.
When the check inevitably runs into ``errors'' due to value-insensitivity, 
it wraps problematic expressions with $\tdead$-casts, which explicate 
the trust obligations that must be discharged by the second phase.
The second phase uses refinement typing,  a flow- and path-sensitive
analysis, that decorates the first phase's types with logical  
predicates to track value relationships and thereby verify the 
casts and establish other correctness properties for dynamically typed languages.
\end{abstract}

\section{Introduction}\label{sec:intro}

Higher-order constructs are increasingly 
adopted in \emph{dynamic scripting} languages, as they 
facilitate the production of clean, correct and maintainable code. 
Consider, for example, the following (first-order) \jsc function
\begin{code}
  function minIndexFO(a) {
    if (a.length <= 0) 
      return -1;
    var min = 0;
    for (var i = 0; i < a.length; i++) {
      if (a[i] < a[min]) 
        min = i;
    }
    return min;
  }
\end{code}
which computes the index of the minimum value 
in the array @a@ by looping over the array, 
updating the @min@ value with each index @i@ 
whose value @a[i]@ is smaller than the ``current'' @a[min]@.
Modern dynamic languages let programmers factor 
the looping pattern into a higher-order @$reduce@ 
function (Figure~\ref{fig:reduce}), which frees 
them from manipulating indices and thereby 
prevents the attendant ``off-by-one'' mistakes. 
Instead, the programmer can compute the minimum index 
by supplying an appropriate @f@ to @reduce@ as in @minIndex@ 
shown at the right of Figure~\ref{fig:reduce}.

\begin{figure}[t!]
\begin{minipage}[t]{0.49\textwidth}
\begin{code}
function $reduce(a, f, x) {
  var res = x, i = 0;
  for (i = 0; i < a.length; i++)
    res = f(res, a[i], i++);
  return res;
}
  
function reduce(a, f, x) {
  if (arguments.length === 3) 
    return $reduce(a, f, x);
  return $reduce(a, f, a[0]);
}
\end{code}
\end{minipage}
\hfill
\begin{minipage}[t]{0.49\textwidth}
\begin{code}
function minIndex(a) {
  if (a.length <= 0) 
    return -1;
  function step(min, cur, i) { 
    return cur < a[min] ? i:min; 
  } 
  return reduce(a, step, 0);
}
\end{code}
\end{minipage}
\caption{Computing the minimum-valued index with Higher-Order Functions}
\label{fig:reduce}
\end{figure}

This trend towards abstraction and reuse poses 
a challenge to static program analyses: 
\emph{how to precisely trace value relationships 
across higher-order functions and containers?}
A variety of dataflow- or abstract interpretation- based 
analyses could be used to verify the safety of array 
accesses in @minIndexFO@ by inferring the loop 
invariant that @i@ and @min@ are between @0@ 
and @a.length@.
Alas, these analyses would fail on @minIndex@. 
The usual methods of procedure summarization apply
to first-order functions, and it is not clear how to 
extend higher-order analyses like CFA to track the 
\emph{relationships} between the values and closures 
that flow to @$reduce@.

\subparagraph*{An Approach: Refinement Types.}
Refinement types \cite{XiPfenning99}
hold the promise of a precise and compositional analysis for higher-order 
functions.
Here, \emph{basic} types are decorated with \emph{refinement} 
predicates that constrain the values inhabiting the type.
For example, we can define
\begin{code}
  type idx<x> = {v:number | 0 <= v && v < len(x) } 
\end{code}
to denote the set of valid indices for an array @x@
and can be used to type @$reduce@ as
\begin{code}
  $reduce :: <A,B>(a: A[], f: (B,A,idx<a>) => B, x: B) => B
\end{code}
The above type is a precise \emph{relational summary} 
of the behavior of @$reduce@: 
the higher-order @f@ is 
only invoked with valid indices for @a@. 
Consequently, @step@ is only called with valid indices 
for @a@, which ensures array safety.

\subparagraph*{Problem: Value-based Overloading.}
A main attraction of dynamic languages is
\emph{value-based overloading}, 
where
syntactic entities (\eg variables) may be bound to 
{multiple types} at run-time, and furthermore,
computations may be customized to particular types,
by {reflecting} on the values bound to variables.
For example, it is common to simplify APIs by overloading 
the @reduce@ function to make the initial value @x@ 
{optional};
when omitted, the first array element @a[0]@ 
is used instead.
Here, @reduce@ really has \emph{two} different function types:
one with 3 parameters and another one with 2.
Furthermore, @reduce@
\emph{reflects} on the size of @arguments@ to select the 
behavior appropriate to the calling context.

Value-based overloading conflicts with a crucial 
prerequisite for refinements, namely that the 
language possesses an \emph{unrefined} static type 
system that provides basic invariants about 
values which can then be refined using logical 
predicates.
Unfortunately, as shown by @reduce@, to soundly 
establish basic typing we must reason about the 
logical relationships between values, which, 
ironically, is exactly the problem we wished to 
solve via refinement typing. 
In other words, value-based overloading creates a 
chicken-and-egg problem: refinements require us to 
first establish basic typing, but the latter itself 
requires reasoning about values (and hence, refinements!).

\subparagraph*{Solution: Trust but Verify.}
We introduce 
\emph{two-phased typing}, a new strategy for 
statically analyzing dynamic languages. 
The key insight is that we can completely decouple
reasoning about \emph{basic} types and \emph{refinements} into
distinct phases by converting ``type errors'' from the first phase
into ``assertion failures'' for the second.
Two-phase typing starts with a source language where value-based 
overloading is specified using \emph{intersections} and (untagged) 
\emph{unions} of the different possible (run-time) types. 

The first phase performs classical, \ie flow-, path- and 
value-insensitive type checking to assign basic types to 
various program expressions.  
When the check inevitably runs into ``errors'' due to 
value-insensitivity, it wraps problematic expressions 
with $\tdead$-casts which allow the first phase to 
proceed, \emph{trusting} that the expressions have 
the casted types.
In other words, the first phase \emph{elaborates}~\cite{Dunfield2012} 
the source language with intersection and (untagged) union types, 
into a target ML-like language with classical products, (tagged) 
sums and \tdead-casts, which explicate the trust obligations that 
must be discharged by the second phase.
The second phase carries out \emph{refinement}, \ie flow- and path-sensitive 
inference, to decorate the basic types (from the first phase) with 
predicates that precisely track relationships about values, and uses
the refinements to \emph{verify} the casts and other properties, 
discharging the assumptions of the first phase.

For example, @reduce@ is described as the intersection of two contexts, \ie
function types which take two and three parameters respectively.
The trust-phase checks the body under both contexts (separately). 
In each context, one of the calls to @$reduce@ is ``ill-typed''.
In the context where the function takes two inputs, 
the call using @x@ is undefined; when the function takes
three inputs, there is a mismatch in the types of @f@ and 
@a[0]@.
Consequently, each ill-typed expression is wrapped with a
\emph{cast} which obliges the verify phase to prove that 
the call is dead code in that context, thereby verifying 
overloading in a cooperative manner.

\subparagraph*{Benefits.}
While it is possible to account for value-based overloading in a 
single phase, the currently known methods that do so are limited 
to the extremes of types and program logics.
At one end, systems like Typed Racket~\cite{typedracket} and Flow
Typing~\cite{lambdajs} extend classical
type systems to account for a fixed set of @typeof@-style tests, but 
cannot reason about general value tests (\eg the size of 
@arguments@) that often appear in idiomatic code.
At the other end, systems like System D~\cite{NestedPOPL12} embed the typing 
relation in an expressive program logic, allowing
general value tests, but give up on basic type structure, thereby 
sacrificing inference, causing a significant annotation overhead.
In contrast, our approach separates the 
concerns of basic typing and reasoning about values, thereby 
yielding several concrete benefits by \emph{modularizing} 
specification, verification and soundness.
\begin{itemize}

\item{\emphbf{Specification:}} 
Instead of a fixed set of type-tests, two-phase typing handles complex
value relationships which can be captured inside refinements in an
expressive logic.  Furthermore, the \emph{expressiveness} of the 
basic type system and logics can be extended independently, \eg 
to account for polymorphism, classes or new logical theories, 
directly yielding a more expressive specification mechanism.

\item{\emphbf{Verification:}} 
Two-phase typing enables the straightforward composition 
of simple type checkers (uncomplicated by reasoning about 
values) with program logics (relying upon the basic 
invariants provided by typing -- \eg the parametric 
polymorphism needed to verify @minIndex@).
Furthermore, two-phase typing allows us to compose basic typing with
abstract interpretation~\cite{LiquidPLDI08}, which drastically lowers
the annotation burden for using refinement types.

\item{\emphbf{Soundness:}}
Finally, our elaboration-based approach makes it straightforward 
to establish soundness for two-phased typing. 
The first phase ignores values and refinements, so we can 
use classical methods to prove the elaborated target is 
``equivalent to'' the source.
The second phase uses standard refinement typing techniques on the 
well-typed elaborated target, and hence lets us directly reuse the 
soundness theorems for such systems~\cite{Knowles10}
to obtain end-to-end soundness for two-phased typing.
\end{itemize}

\subparagraph*{Contributions.}
Concretely, in this paper we make the following contributions.
First, we informally illustrate (\S~\ref{sec:overview}) 
how two-phase typing lets us statically analyze dynamic, 
value-based overloading patterns drawn from real-world code, 
where, we empirically demonstrate, value-based overloading 
is ubiquitous.
Second, we formalize two-phase typing using a core calculus, 
\toolname, whose syntax and semantics are detailed in \S~\ref{sec:language}.
Third, we formalize the first phase (\S~\ref{sec:elaboration}),
which \emph{elaborates} \cite{Dunfield2012} a source language 
with value-based overloading into a target language with
$\tdead$-casts in lieu of overloading. We prove that the 
elaborated target preserves the semantics of the source, 
\ie the $\tdead$-casts fail iff the source would hit a
type error at run time.
Finally, we demonstrate how standard refinement typing 
machinery can be applied to the elaborated well-typed target
(\S~\ref{sec:ref-checking})
to statically verify the $\tdead$-casts, yielding 
end-to-end soundness for our system.

\section{Overview}\label{sec:overview}

We begin with an overview illustrating how we soundly verify
value-based overloading using our novel two-phased approach.

\subsection{Value-based Overloading}\label{sec:overview:problem}

\begin{figure}[t]
\begin{minipage}[t]{0.49\textwidth}
\begin{code}
neg :: (number, number) => number
    /\ (number, boolean) => boolean
function neg(flag, x) {
  if (flag) return 0-x;
  return !x;
} 
\end{code}
\end{minipage}
\hfill
\begin{minipage}[t]{0.49\textwidth}
\begin{code}
var a = neg(1,1);    // OK
var b = neg(0,true); // OK
var c = neg(0,1);    // ERR
var d = neg(1,true); // ERR
\end{code}
\end{minipage}
\caption{An example program with value-based overloading}
\label{fig:negate:ts} 
\end{figure}

Consider the code in Figure~\ref{fig:negate:ts}.
The function $\tnegate$ behaves as follows.
When a @number@ is passed as input, indicated by passing in a
\emph{non-zero}, \ie ``truthy'' @flag@, the function flips its 
sign by subtracting the input from @0@. 
Instead, when a @boolean@ is passed in,
indicated by a \emph{zero}, \ie ``falsy'' @flag@, the function 
returns the boolean negation.  Hence, the calls made to assign 
@a@ and @b@ are legitimate and should be statically accepted. 
However, the calls made to assign @c@ and @d@ lead to run-time 
errors (assuming we eschew implicit coercions), and hence, 
should be rejected.

The function $\tnegate$ distils value-based overloading to its 
essence: a run-time test on one parameter's value is used to
determine the type of, and hence the operation to be applied to,
another value. Of course in JavaScript, one could use a single 
parameter and the @typeof@ operator for this particular simple 
case, and design analyses targeted towards a fixed set of type 
tests, \eg using variants of the @typeof@ 
operator~\cite{typedracket,lambdajs}.
However, arbitrary value tests -- such as tests of the size of 
@arguments@ shown in @reduce@ in Figure~\ref{fig:reduce} -- 
can be and are used in practice. 
Thus, we illustrate the generality of the problem and our 
solution \emph{without} using the @typeof@ operator (which
is a special case of our solution).

\subparagraph*{Prevalence of Value-based Overloading.} 
The code from Figure~\ref{fig:reduce} is not a 
pathological toy example. It is adapted from the 
widely used \textsc{D3} visualization library.
The advent of \tsc makes it possible to establish 
the prevalence of value-based overloading in real-world 
libraries, as it allows developers to specify overloaded 
signatures for functions.
(Even though \tsc does not {verify} those signatures, 
it uses them as trusted interfaces for 
external \jsc libraries and code completion.)
The Definitely Typed repository~\footnote{\url{http://definitelytyped.org}}
contains \tsc interfaces for a large number of popular 
\jsc libraries.
We analyzed the \tsc interfaces to determine the prevalence of
value-based overloading. Intuitively, every function or method
with multiple (overloaded) signatures or
optional arguments has an implementation that uses value-based 
overloading.

Figure~\ref{fig:overload-graph} summarizes the results of 
our study. 
On the left, we show the fraction of overloaded 
functions in the 10 benchmarks analyzed by Feldthaus \etal~\cite{MollerOOPSLA14}. 
The data shows that over 25\% of the functions in 4 of 10 
libraries use value-based overloading, and an even larger
fraction is overloaded in libraries like @jquery@ and @d3@.
On the right we summarize the occurrence of overloading across all
the libraries in Definitely Typed. The data shows, for 
example, that in more than
25\% of the libraries, \emph{more than} 25\% of the functions are 
overloaded with multiple types. The figure jumps to nearly 55\% of 
functions if we also include optional arguments.

\begin{figure}[t!]
\begin{minipage}[b]{0.50\textwidth}
\begin{tabular}{|l|r @{\hskip 3pt} r @{\hskip 3pt} r @{\hskip 3pt} r|}
\hline
\textbf{File} & \textbf{\#Funs} & \textbf{\%Ovl} & \textbf{\%Opt} & \textbf{\%Any} \\
\hline\hline
\begin{filecontents*}{overload-table-full.csv}
Filename,Total Functions,OverloadANDOptArg Functions,Overload Functions,OptArg Functions,Percent Overload,Percent OptArg,Percent OverloadUNIONOptArg,display name
angular-bootstrap-lightbox/angular-bootstrap-lightbox.d.ts,1,0,0,0,0,0,0,angular-bootstrap-lightbox
angularjs/legacy/angular-sanitize-1.0.d.ts,1,0,0,0,0,0,0,angularjs
atmosphere/atmosphere.d.ts,1,0,0,0,0,0,0,atmosphere
buffer-equal/buffer-equal.d.ts,1,0,0,0,0,0,0,buffer-equal
bunyan-logentries/bunyan-logentries.d.ts,1,0,0,0,0,0,0,bunyan-logentries
consolidate/consolidate.d.ts,1,0,0,0,0,0,0,consolidate
deep-freeze/deep-freeze.d.ts,1,0,0,0,0,0,0,deep-freeze
dojo/dojox.jsonPath.d.ts,1,0,0,0,0,0,0,dojo
dojo/dojox.rails.d.ts,1,0,0,0,0,0,0,dojo
domready/domready.d.ts,1,0,0,0,0,0,0,domready
errorhandler/errorhandler.d.ts,1,0,0,0,0,0,0,errorhandler
exit/exit.d.ts,1,0,0,0,0,0,0,exit
express-jwt/express-jwt.d.ts,1,0,0,0,0,0,0,express-jwt
express-myconnection/express-myconnection.d.ts,1,0,0,0,0,0,0,express-myconnection
express-unless/express-unless.d.ts,1,0,0,0,0,0,0,express-unless
jquery.customSelect/jquery.customSelect.d.ts,1,0,0,0,0,0,0,jquery.customSelect
jquery.fileupload/jquery.fileupload.d.ts,1,0,0,0,0,0,0,jquery.fileupload
jquery.placeholder/jquery.placeholder.d.ts,1,0,0,0,0,0,0,jquery.placeholder
pegjs/pegjs.d.ts,1,0,0,0,0,0,0,pegjs
readdir-stream/readdir-stream.d.ts,1,0,0,0,0,0,0,readdir-stream
slickgrid/slick.headerbuttons.d.ts,1,0,0,0,0,0,0,slickgrid
stream-to-array/stream-to-array.d.ts,1,0,0,0,0,0,0,stream-to-array
timelinejs/timelinejs.d.ts,1,0,0,0,0,0,0,timelinejs
traceback/traceback.d.ts,1,0,0,0,0,0,0,traceback
webcrypto/WebCrypto.d.ts,1,0,0,0,0,0,0,webcrypto
angular-spinner/angular-spinner.d.ts,2,0,0,0,0,0,0,angular-spinner
business-rules-engine/Utils.d.ts,2,0,0,0,0,0,0,business-rules-engine
cordova/plugins/Splashscreen.d.ts,2,0,0,0,0,0,0,cordova
gamepad/gamepad.d.ts,2,0,0,0,0,0,0,gamepad
jasmine-data_driven_tests/jasmine-data_driven_tests.d.ts,2,0,0,0,0,0,0,jasmine-data_driven_tests
jasmine-fixture/jasmine-fixture.d.ts,2,0,0,0,0,0,0,jasmine-fixture
jquery.pnotify/jquery.pnotify.d.ts,2,0,0,0,0,0,0,jquery.pnotify
karma-jasmine/karma-jasmine.d.ts,2,0,0,0,0,0,0,karma-jasmine
knockout.projections/knockout.projections.d.ts,2,0,0,0,0,0,0,knockout.projections
logg/logg.d.ts,2,0,0,0,0,0,0,logg
mixto/mixto.d.ts,2,0,0,0,0,0,0,mixto
rimraf/rimraf.d.ts,2,0,0,0,0,0,0,rimraf
riotjs/riotjs.d.ts,2,0,0,0,0,0,0,riotjs
rx/rx.coincidence-lite.d.ts,2,0,0,0,0,0,0,rx
sprintf/sprintf.d.ts,2,0,0,0,0,0,0,sprintf
angular-file-upload/angular-file-upload.d.ts,3,0,0,0,0,0,0,angular-file-upload
angularjs/angular-cookies.d.ts,3,0,0,0,0,0,0,angularjs
angularjs/angular-route.d.ts,3,0,0,0,0,0,0,angularjs
angularjs/legacy/angular-cookies-1.0.d.ts,3,0,0,0,0,0,0,angularjs
angularjs/legacy/angular-cookies-1.2.d.ts,3,0,0,0,0,0,0,angularjs
angularjs/legacy/angular-route-1.2.d.ts,3,0,0,0,0,0,0,angularjs
bootstrap-notify/bootstrap-notify.d.ts,3,0,0,0,0,0,0,bootstrap-notify
chalk/chalk.d.ts,3,0,0,0,0,0,0,chalk
configstore/configstore.d.ts,3,0,0,0,0,0,0,configstore
cordova-ionic/plugins/keyboard.d.ts,3,0,0,0,0,0,0,cordova-ionic
cordova/plugins/Push.d.ts,3,0,0,0,0,0,0,cordova
cordova/plugins/Vibration.d.ts,3,0,0,0,0,0,0,cordova
dojo/dojox.analytics.d.ts,3,0,0,0,0,0,0,dojo
dojo/dojox.secure.d.ts,3,0,0,0,0,0,0,dojo
form-data/form-data.d.ts,3,0,0,0,0,0,0,form-data
linq/linq.jquery.d.ts,3,0,0,0,0,0,0,linq
moviedb/moviedb.d.ts,3,0,0,0,0,0,0,moviedb
platform/platform.d.ts,3,0,0,0,0,0,0,platform
angular-agility/angular-agility.d.ts,4,0,0,0,0,0,0,angular-agility
bgiframe/typescript.bgiframe.d.ts,4,0,0,0,0,0,0,bgiframe
cometd/cometd.d.ts,4,0,0,0,0,0,0,cometd
googlemaps.infobubble/google.maps.infobubble.d.ts,4,0,0,0,0,0,0,googlemaps.infobubble
heatmap.js/heatmap.d.ts,4,0,0,0,0,0,0,heatmap.js
mocha-phantomjs/mocha-phantomjs.d.ts,4,0,0,0,0,0,0,mocha-phantomjs
scroller/easyscroller.d.ts,4,0,0,0,0,0,0,scroller
stats/stats.d.ts,4,0,0,0,0,0,0,stats
threejs/three-maskpass.d.ts,4,0,0,0,0,0,0,threejs
uri-templates/uri-templates.d.ts,4,0,0,0,0,0,0,uri-templates
viewporter/viewporter.d.ts,4,0,0,0,0,0,0,viewporter
debug/debug.d.ts,5,0,0,0,0,0,0,debug
dojo/dojox.socket.d.ts,5,0,0,0,0,0,0,dojo
dojo/dojox.testing.d.ts,5,0,0,0,0,0,0,dojo
http-string-parser/http-string-parser.d.ts,5,0,0,0,0,0,0,http-string-parser
mime/mime.d.ts,5,0,0,0,0,0,0,mime
precise/precise.d.ts,5,0,0,0,0,0,0,precise
threejs/three-css3drenderer.d.ts,5,0,0,0,0,0,0,threejs
cordovarduino/cordovarduino.d.ts,6,0,0,0,0,0,0,cordovarduino
flux/flux.d.ts,6,0,0,0,0,0,0,flux
jquery.menuaim/jquery.menuaim.d.ts,6,0,0,0,0,0,0,jquery.menuaim
knockout.kogrid/ko-grid.d.ts,6,0,0,0,0,0,0,knockout.kogrid
lz-string/lz-string.d.ts,6,0,0,0,0,0,0,lz-string
d3/plugins/d3.superformula.d.ts,7,0,0,0,0,0,0,d3
messenger/messenger.d.ts,7,0,0,0,0,0,0,messenger
radius/radius.d.ts,7,0,0,0,0,0,0,radius
stripe/stripe.d.ts,7,0,0,0,0,0,0,stripe
filewriter/filewriter.d.ts,8,0,0,0,0,0,0,filewriter
uuid/UUID.d.ts,8,0,0,0,0,0,0,uuid
hashmap/hashmap.d.ts,9,0,0,0,0,0,0,hashmap
recaptcha/recaptcha.d.ts,9,0,0,0,0,0,0,recaptcha
sinon-chai/sinon-chai.d.ts,9,0,0,0,0,0,0,sinon-chai
twig/twig.d.ts,9,0,0,0,0,0,0,twig
js-fixtures/fixtures.d.ts,10,0,0,0,0,0,0,js-fixtures
swipeview/swipeview.d.ts,10,0,0,0,0,0,0,swipeview
colors/colors.d.ts,11,0,0,0,0,0,0,colors
dojo/dojox.gesture.d.ts,11,0,0,0,0,0,0,dojo
easystarjs/easystarjs.d.ts,12,0,0,0,0,0,0,easystarjs
jquery.noty/jquery.noty.d.ts,13,0,0,0,0,0,0,jquery.noty
fingerprintjs/fingerprint.d.ts,14,0,0,0,0,0,0,fingerprintjs
history/history.d.ts,14,0,0,0,0,0,0,history
rx/rx.virtualtime.d.ts,14,0,0,0,0,0,0,rx
status-bar/status-bar.d.ts,15,0,0,0,0,0,0,status-bar
json-pointer/json-pointer.d.ts,17,0,0,0,0,0,0,json-pointer
knockstrap/knockstrap.d.ts,17,0,0,0,0,0,0,knockstrap
rx/rx.joinpatterns.d.ts,20,0,0,0,0,0,0,rx
change-case/change-case.d.ts,30,0,0,0,0,0,0,change-case
webrtc/MediaStream.d.ts,30,0,0,0,0,0,0,webrtc
dojo/dojox.NodeList.d.ts,119,0,0,0,0,0,0,dojo
ng-grid/ng-grid.d.ts,120,0,0,0,0,0,0,ng-grid
dojo/dojox.gantt.d.ts,210,0,0,0,0,0,0,dojo
space-pen/space-pen.d.ts,298,0,4,0,1,0,1,space-pen
dojo/dojox.mdnd.d.ts,124,0,1,1,1,1,2,dojo
dojo/dojox.date.d.ts,174,0,3,0,2,0,2,dojo
noVNC/noVNC.d.ts,50,0,0,1,0,2,2,noVNC
jest/jest.d.ts,47,0,0,1,0,2,2,jest
dojo/dojox.gfx.d.ts,1369,0,34,1,2,0,3,dojo
sharepoint/SharePoint.d.ts,4016,10,65,48,2,1,3,sharepoint
typescript-services/typescriptServices.d.ts,5502,0,3,150,0,3,3,typescript-services
box2d/box2dweb.d.ts,529,0,0,15,0,3,3,box2d
flot/jquery.flot.d.ts,32,1,1,1,3,3,3,flot
unity-webapi/unity-webapi.d.ts,32,0,0,1,0,3,3,unity-webapi
jwplayer/jwplayer.d.ts,61,0,1,1,2,2,3,jwplayer
dojo/dojox.dtl.d.ts,332,0,10,1,3,0,3,dojo
jsfl/xJSFL.d.ts,30,0,0,1,0,3,3,jsfl
signature_pad/signature_pad.d.ts,26,0,0,1,0,4,4,signature_pad
js-git/js-git.d.ts,24,0,0,1,0,4,4,js-git
business-rules-engine/business-rules-engine.d.ts,48,0,0,2,0,4,4,business-rules-engine
business-rules-engine/Validation.d.ts,48,0,0,2,0,4,4,business-rules-engine
playerframework/playerFramework.d.ts,91,0,0,4,0,4,4,playerframework
dojo/dojox.html.d.ts,43,0,2,0,5,0,5,dojo
emissary/emissary.d.ts,21,0,0,1,0,5,5,emissary
jsbn/jsbn.d.ts,81,0,0,4,0,5,5,jsbn
bigint/bigint.d.ts,59,1,2,2,3,3,5,bigint
kineticjs/kineticjs.d.ts,266,0,6,8,2,3,5,kineticjs
dojo/dojox.data.d.ts,704,0,34,5,5,1,6,dojo
dojo/dojox.gfx3d.d.ts,404,0,23,0,6,0,6,dojo
ace/ace.d.ts,484,0,3,25,1,5,6,ace
leapmotionTS/LeapMotionTS.d.ts,134,0,0,8,0,6,6,leapmotionTS
diff/diff.d.ts,16,0,0,1,0,6,6,diff
hashset/hashset.d.ts,16,0,0,1,0,6,6,hashset
jquery.bootstrap.wizard/jquery.bootstrap.wizard.d.ts,16,0,0,1,0,6,6,jquery.bootstrap.wizard
socket.io-client/legacy/socket.io-client-0.9.d.ts,16,0,0,1,0,6,6,socket.io-client
crossfilter/crossfilter.d.ts,47,0,2,1,4,2,6,crossfilter
jasmine-jquery/jasmine-jquery.d.ts,94,0,2,4,2,4,6,jasmine-jquery
socket.io/legacy/socket.io-0.9.d.ts,31,0,2,0,6,0,6,socket.io
dojo/dojox.drawing.d.ts,1420,1,58,35,4,2,6,dojo
log4javascript/log4javascript.d.ts,166,0,0,11,0,7,7,log4javascript
slickgrid/SlickGrid.d.ts,191,1,4,10,2,5,7,slickgrid
jsoneditoronline/jsoneditoronline.d.ts,88,0,0,6,0,7,7,jsoneditoronline
urlrouter/urlrouter.d.ts,29,0,0,2,0,7,7,urlrouter
zeroclipboard/zeroclipboard.d.ts,14,0,0,1,0,7,7,zeroclipboard
jquery.colorpicker/jquery.colorpicker.d.ts,28,1,1,2,4,7,7,jquery.colorpicker
dock-spawn/dock-spawn.d.ts,42,0,0,3,0,7,7,dock-spawn
datejs/sugarpak.d.ts,181,0,1,12,1,7,7,datejs
mathjax/mathjax.d.ts,96,0,7,0,7,0,7,mathjax
selectize/selectize.d.ts,53,0,1,3,2,6,8,selectize
convert-source-map/convert-source-map.d.ts,13,0,0,1,0,8,8,convert-source-map
dojo/dojox.math.d.ts,13,0,0,1,0,8,8,dojo
pdf/pdf.d.ts,52,1,1,4,2,8,8,pdf
jasmine-matchers/jasmine-matchers.d.ts,25,0,2,0,8,0,8,jasmine-matchers
jquery.cycle2/jquery.cycle2.d.ts,37,1,2,2,5,5,8,jquery.cycle2
pathwatcher/pathwatcher.d.ts,37,0,0,3,0,8,8,pathwatcher
atom/atom.d.ts,955,2,9,72,1,8,8,atom
fpsmeter/FPSMeter.d.ts,12,0,0,1,0,8,8,fpsmeter
ladda/ladda.d.ts,12,1,1,1,8,8,8,ladda
jquery.dynatree/jquery.dynatree.d.ts,71,2,2,6,3,8,8,jquery.dynatree
amcharts/AmCharts.d.ts,82,1,7,1,9,1,9,amcharts
backgrid/backgrid.d.ts,23,0,0,2,0,9,9,backgrid
express-validator/express-validator.d.ts,56,0,3,2,5,4,9,express-validator
dhtmlxscheduler/dhtmlxscheduler.d.ts,166,0,0,15,0,9,9,dhtmlxscheduler
glidejs/glidejs.d.ts,11,0,0,1,0,9,9,glidejs
dojo/dojox.color.d.ts,33,0,0,3,0,9,9,dojo
svg-pan-zoom/svg-pan-zoom.d.ts,33,0,0,3,0,9,9,svg-pan-zoom
datejs/datejs.d.ts,55,0,1,4,2,7,9,datejs
mapsjs/mapsjs.d.ts,335,0,1,30,0,9,9,mapsjs
dojo/dojox.rpc.d.ts,21,0,1,1,5,5,10,dojo
titanium/titanium.d.ts,2845,14,98,187,3,7,10,titanium
react/legacy/react-addons-0.11.d.ts,31,1,1,3,3,10,10,react
phantomjs/phantomjs.d.ts,93,2,6,5,6,5,10,phantomjs
dhtmlxgantt/dhtmlxgantt.d.ts,111,0,0,11,0,10,10,dhtmlxgantt
knockout.editables/ko.editables.d.ts,10,0,0,1,0,10,10,knockout.editables
threejs/three-projector.d.ts,10,0,0,1,0,10,10,threejs
dojo/dojox.css3.d.ts,20,0,0,2,0,10,10,dojo
dojo/dojox.string.d.ts,19,0,0,2,0,11,11,dojo
hashtable/hashtable.d.ts,19,0,0,2,0,11,11,hashtable
youtube/youtube.d.ts,38,4,4,4,11,11,11,youtube
dojo/dojox.encoding.d.ts,56,0,0,6,0,11,11,dojo
tween.js/tween.js.d.ts,56,0,0,6,0,11,11,tween.js
jasmine/legacy/jasmine-1.3.d.ts,172,0,3,16,2,9,11,jasmine
atpl/atpl.d.ts,9,0,0,1,0,11,11,atpl
chai-jquery/chai-jquery.d.ts,9,0,0,1,0,11,11,chai-jquery
memory-cache/memory-cache.d.ts,9,0,0,1,0,11,11,memory-cache
vimeo/froogaloop.d.ts,9,0,1,0,11,0,11,vimeo
keyboardjs/keyboardjs.d.ts,18,0,0,2,0,11,11,keyboardjs
dojo/dojox.validate.d.ts,107,0,12,0,11,0,11,dojo
dojo/doh.d.ts,165,0,19,0,12,0,12,dojo
timezonecomplete/timezonecomplete-1.6.0.d.ts,147,1,4,14,3,10,12,timezonecomplete
timezonecomplete/timezonecomplete-1.4.6.d.ts,129,1,4,12,3,9,12,timezonecomplete
timezonecomplete/timezonecomplete-1.5.1.d.ts,146,1,4,14,3,10,12,timezonecomplete
dojo/dojox.uuid.d.ts,17,0,0,2,0,12,12,dojo
minilog/minilog.d.ts,17,0,0,2,0,12,12,minilog
timezonecomplete/timezonecomplete-1.3.0.d.ts,102,0,3,9,3,9,12,timezonecomplete
sinon/sinon.d.ts,142,1,13,5,9,4,12,sinon
jshamcrest/jshamcrest.d.ts,124,0,1,14,1,11,12,jshamcrest
timezonecomplete/timezonecomplete-1.2.0.d.ts,99,0,3,9,3,9,12,timezonecomplete
pixi/pixi.d.ts,123,0,0,15,0,12,12,pixi
timezonecomplete/timezonecomplete-1.9.0.d.ts,155,1,6,14,4,9,12,timezonecomplete
timezonecomplete/timezonecomplete-1.8.0.d.ts,153,1,6,14,4,9,12,timezonecomplete
highlightjs/highlightjs-7.5.0.d.ts,8,0,0,1,0,13,13,highlightjs
storejs/storejs.d.ts,8,0,0,1,0,13,13,storejs
dojo/dojox.charting.d.ts,1483,4,68,124,5,8,13,dojo
jasmine/jasmine.d.ts,180,2,6,19,3,11,13,jasmine
fabricjs/fabricjs.d.ts,371,1,17,32,5,9,13,fabricjs
node-form/node-form.d.ts,100,0,0,13,0,13,13,node-form
cli-color/cli-color.d.ts,15,0,0,2,0,13,13,cli-color
jquery.ui.layout/jquery.ui.layout.d.ts,15,0,0,2,0,13,13,jquery.ui.layout
numeraljs/numeraljs.d.ts,15,0,0,2,0,13,13,numeraljs
aws-sdk/aws-sdk.d.ts,52,0,0,7,0,13,13,aws-sdk
dojo/dojox.dnd.d.ts,37,0,4,1,11,3,14,dojo
soundjs/soundjs.d.ts,74,7,8,9,11,12,14,soundjs
text-buffer/text-buffer.d.ts,131,1,12,7,9,5,14,text-buffer
timezonecomplete/timezonecomplete.d.ts,167,1,7,17,4,10,14,timezonecomplete
dojo/dojo.d.ts,1680,0,179,54,11,3,14,dojo
jasmine/jasmine.d.ts,180,2,8,19,4,11,14,jasmine
openlayers/openlayers.d.ts,280,3,17,25,6,9,14,openlayers
goJS/goJS.d.ts,713,6,13,94,2,13,14,goJS
htmlparser2/htmlparser2.d.ts,7,0,0,1,0,14,14,htmlparser2
oclazyload/oclazyload.d.ts,7,0,1,0,14,0,14,oclazyload
routie/routie.d.ts,7,0,1,0,14,0,14,routie
keypress/keypress.d.ts,14,0,2,0,14,0,14,keypress
typeahead/typeahead.d.ts,14,0,1,1,7,7,14,typeahead
camljs/camljs.d.ts,104,0,9,6,9,6,14,camljs
backbone-relational/backbone-relational.d.ts,48,1,5,3,10,6,15,backbone-relational
sigmajs/sigmajs.d.ts,41,0,6,0,15,0,15,sigmajs
node-webkit/node-webkit.d.ts,82,2,3,11,4,13,15,node-webkit
async/async.d.ts,68,7,7,10,10,15,15,async
knockout.validation/knockout.validation.d.ts,33,0,2,3,6,9,15,knockout.validation
crossroads/crossroads.d.ts,13,0,0,2,0,15,15,crossroads
notifyjs/notifyjs.d.ts,13,0,0,2,0,15,15,notifyjs
dojo/dojox.sketch.d.ts,234,1,8,29,3,12,15,dojo
winrt/winrt.d.ts,3321,12,494,31,15,1,15,winrt
angular-protractor/angular-protractor.d.ts,96,0,3,12,3,13,16,angular-protractor
googlemaps/google.maps.d.ts,318,5,14,41,4,13,16,googlemaps
dojo/dojox.io.d.ts,19,0,0,3,0,16,16,dojo
mustache/mustache-0.7.3.d.ts,19,0,1,2,5,11,16,mustache
stylus/stylus.d.ts,220,0,31,4,14,2,16,stylus
npm/npm.d.ts,25,0,1,3,4,12,16,npm
less/less.d.ts,193,2,8,25,4,13,16,less
dojo/dojox.app.d.ts,297,1,7,42,2,14,16,dojo
firefox/firefox.d.ts,6,0,0,1,0,17,17,firefox
impress/impress.d.ts,6,0,0,1,0,17,17,impress
insight/insight.d.ts,6,0,0,1,0,17,17,insight
jquery.payment/jquery.payment.d.ts,6,0,1,0,17,0,17,jquery.payment
msgpack/msgpack.d.ts,6,0,0,1,0,17,17,msgpack
showdown/showdown.d.ts,6,0,0,1,0,17,17,showdown
tv4/tv4.d.ts,12,0,0,2,0,17,17,tv4
webaudioapi/waa-20120802.d.ts,42,1,2,6,5,14,17,webaudioapi
jquery-handsontable/jquery-handsontable.d.ts,60,3,5,8,8,13,17,jquery-handsontable
libxmljs/libxmljs.d.ts,66,0,8,3,12,5,17,libxmljs
swiz/swiz.d.ts,90,0,3,12,3,13,17,swiz
long/long.d.ts,41,1,1,7,2,17,17,long
signalr/signalr.d.ts,29,0,1,4,3,14,17,signalr
lunr/lunr.d.ts,63,0,3,8,5,13,17,lunr
react/legacy/react-0.11.d.ts,40,0,0,7,0,18,18,react
dat-gui/dat-gui.d.ts,17,0,2,1,12,6,18,dat-gui
nomnom/nomnom.d.ts,17,0,0,3,0,18,18,nomnom
yeoman-generator/yeoman-generator.d.ts,39,1,4,4,10,10,18,yeoman-generator
dojo/dojox.grid.d.ts,3762,13,106,586,3,16,18,dojo
microsoft-ajax/microsoft.ajax.d.ts,365,6,26,46,7,13,18,microsoft-ajax
cordova/plugins/Media.d.ts,11,0,0,2,0,18,18,cordova
pleasejs/please.d.ts,11,2,2,2,18,18,18,pleasejs
expectations/expectations.d.ts,22,0,0,4,0,18,18,expectations
js-signals/js-signals.d.ts,16,0,0,3,0,19,19,js-signals
node-fibers/node-fibers.d.ts,16,0,1,2,6,13,19,node-fibers
socket.io-client/socket.io-client.d.ts,16,2,2,3,13,19,19,socket.io-client
mapbox/mapbox.d.ts,58,2,4,9,7,16,19,mapbox
angular-protractor/legacy/angular-protractor-0.17.0.d.ts,79,0,7,8,9,10,19,angular-protractor
tedious/tedious.d.ts,21,0,1,3,5,14,19,tedious
ref-array/ref-array.d.ts,36,2,2,7,6,19,19,ref-array
react/react.d.ts,77,1,7,9,9,12,19,react
rethinkdb/rethinkdb.d.ts,87,5,7,15,8,17,20,rethinkdb
jquery.total-storage/jquery.total-storage.d.ts,5,0,0,1,0,20,20,jquery.total-storage
jquery.watermark/jquery.watermark.d.ts,5,0,0,1,0,20,20,jquery.watermark
ngprogress-lite/ngprogress-lite.d.ts,5,0,0,1,0,20,20,ngprogress-lite
slickgrid/slick.rowselectionmodel.d.ts,5,0,0,1,0,20,20,slickgrid
source-map-support/source-map-support.d.ts,5,0,0,1,0,20,20,source-map-support
keymaster/keymaster.d.ts,10,0,2,0,20,0,20,keymaster
lru-cache/lru-cache.d.ts,10,0,1,1,10,10,20,lru-cache
tar/tar.d.ts,10,0,1,1,10,10,20,tar
x2js/xml2json.d.ts,10,0,1,1,10,10,20,x2js
jquery.cleditor/jquery.cleditor.d.ts,15,0,0,3,0,20,20,jquery.cleditor
emscripten/emscripten.d.ts,80,2,5,13,6,16,20,emscripten
jquery.validation/jquery.validation.d.ts,19,0,2,2,11,11,21,jquery.validation
knockout/knockout.d.ts,166,3,21,17,13,10,21,knockout
knockout/knockout.d.ts,166,4,22,17,13,10,21,knockout
webaudioapi/waa.d.ts,47,1,2,9,4,19,21,webaudioapi
greasemonkey/greasemonkey.d.ts,14,1,2,2,14,14,21,greasemonkey
eventemitter2/eventemitter2.d.ts,28,2,4,4,14,14,21,eventemitter2
lazy.js/lazy.js.d.ts,116,0,8,17,7,15,22,lazy.js
jointjs/jointjs.d.ts,46,0,0,10,0,22,22,jointjs
highland/highland.d.ts,69,0,13,2,19,3,22,highland
node-git/node-git.d.ts,170,0,35,2,21,1,22,node-git
sammyjs/sammyjs.d.ts,155,0,14,20,9,13,22,sammyjs
_infrastructure/tests/typescript/1.1.0-1/lib.d.ts,1822,50,92,361,5,20,22,_infrastructure
_infrastructure/tests/typescript/1.3.0/lib.d.ts,1822,50,92,361,5,20,22,_infrastructure
deep-diff/deep-diff.d.ts,9,0,1,1,11,11,22,deep-diff
kolite/kolite.d.ts,9,0,0,2,0,22,22,kolite
webmidi/webmidi.d.ts,9,0,0,2,0,22,22,webmidi
royalslider/royalslider.d.ts,18,0,0,4,0,22,22,royalslider
node_redis/node_redis.d.ts,27,1,2,5,7,19,22,node_redis
reveal/reveal.d.ts,27,0,0,6,0,22,22,reveal
sendgrid/sendgrid.d.ts,27,0,4,2,15,7,22,sendgrid
teechart/teechart.d.ts,81,0,0,18,0,22,22,teechart
dojo/dojox.geo.d.ts,543,3,33,91,6,17,22,dojo
flight/flight.d.ts,31,3,5,5,16,16,23,flight
snapsvg/snapsvg.d.ts,168,3,21,20,13,12,23,snapsvg
ember/ember.d.ts,989,8,53,180,5,18,23,ember
_infrastructure/tests/typescript/0.9.7/lib.d.ts,1712,51,74,372,4,22,23,_infrastructure
wrench/wrench.d.ts,13,0,0,3,0,23,23,wrench
_infrastructure/tests/typescript/0.9.1.1/lib.d.ts,1683,48,71,367,4,22,23,_infrastructure
ckeditor/ckeditor.d.ts,509,10,22,106,4,21,23,ckeditor
webaudioapi/waa-nightly.d.ts,43,1,2,9,5,21,23,webaudioapi
_infrastructure/tests/typescript/0.9.5/lib.d.ts,1684,49,72,370,4,22,23,_infrastructure
selenium-webdriver/selenium-webdriver.d.ts,230,6,20,40,9,17,23,selenium-webdriver
dojo/dojox.atom.d.ts,892,6,46,170,5,19,24,dojo
jsfl/jsfl.d.ts,378,1,2,88,1,23,24,jsfl
_infrastructure/tests/typescript/1.0.0/lib.d.ts,1712,51,94,372,5,22,24,_infrastructure
_infrastructure/tests/typescript/1.0.1/lib.d.ts,1712,51,94,372,5,22,24,_infrastructure
_infrastructure/tests/typescript/1.0.2.1-225bad/lib.d.ts,1712,51,94,372,5,22,24,_infrastructure
knockback/knockback.d.ts,74,1,12,7,16,9,24,knockback
threejs/three.d.ts,1022,4,8,245,1,24,24,threejs
jquery.elang/jquery.elang.d.ts,53,0,0,13,0,25,25,jquery.elang
dojo/dojox.json.d.ts,4,0,0,1,0,25,25,dojo
ion.rangeSlider/ion.rangeSlider.d.ts,4,0,1,0,25,0,25,ion.rangeSlider
jquery.tinycarousel/jquery.tinycarousel.d.ts,4,0,0,1,0,25,25,jquery.tinycarousel
power-assert-formatter/power-assert-formatter.d.ts,4,0,0,1,0,25,25,power-assert-formatter
pubsubjs/pubsub.d.ts,4,0,1,0,25,0,25,pubsubjs
type-detect/type-detect.d.ts,4,0,1,0,25,0,25,type-detect
bigscreen/bigscreen.d.ts,8,0,0,2,0,25,25,bigscreen
google.feeds/google.feed.api.d.ts,8,0,0,2,0,25,25,google.feeds
jquery.colorbox/jquery.colorbox.d.ts,8,0,2,0,25,0,25,jquery.colorbox
hammerjs/hammerjs-1.1.3.d.ts,12,0,0,3,0,25,25,hammerjs
handlebars/handlebars.d.ts,12,0,0,3,0,25,25,handlebars
highlightjs/highlightjs.d.ts,12,0,0,3,0,25,25,highlightjs
threejs/three-canvasrenderer.d.ts,20,1,1,5,5,25,25,threejs
source-map/source-map.d.ts,24,0,0,6,0,25,25,source-map
vega/vega.d.ts,56,0,11,3,20,5,25,vega
winjs/winjs-2.1.d.ts,599,5,10,145,2,24,25,winjs
q-io/Q-io.d.ts,79,5,17,8,22,10,25,q-io
giraffe/giraffe.d.ts,63,0,0,16,0,25,25,giraffe
codemirror/codemirror.d.ts,115,3,11,22,10,19,26,codemirror
winjs/winjs.d.ts,747,5,10,192,1,26,26,winjs
stampit/stampit.d.ts,15,0,1,3,7,20,27,stampit
winjs/winjs-2.0.d.ts,718,5,10,188,1,26,27,winjs
dojo/dojox.fx.d.ts,156,1,17,26,11,17,27,dojo
each/each.d.ts,11,0,3,0,27,0,27,each
progressjs/progress.d.ts,11,0,1,2,9,18,27,progressjs
tweenjs/tweenjs.d.ts,44,1,5,8,11,18,27,tweenjs
jsplumb/jquery.jsPlumb.d.ts,25,3,4,6,16,24,28,jsplumb
rabbit.js/rabbit.js.d.ts,46,2,2,13,4,28,28,rabbit.js
google.visualization/google.visualization.d.ts,145,6,25,22,17,15,28,google.visualization
winjs/winjs-1.0.d.ts,555,5,10,152,2,27,28,winjs
valerie/valerie.d.ts,120,13,17,30,14,25,28,valerie
jquery.prettyphoto/jquery.prettyphoto.d.ts,7,0,2,0,29,0,29,jquery.prettyphoto
nprogress/NProgress.d.ts,7,0,1,1,14,14,29,nprogress
state-machine/state-machine.d.ts,7,0,0,2,0,29,29,state-machine
threejs/three-effectcomposer.d.ts,7,0,0,2,0,29,29,threejs
mustache/mustache.d.ts,14,0,0,4,0,29,29,mustache
optimist/optimist.d.ts,14,0,4,0,29,0,29,optimist
rest/rest.d.ts,14,0,2,2,14,14,29,rest
business-rules-engine/BasicValidators.d.ts,28,0,0,8,0,29,29,business-rules-engine
jsdeferred/jsdeferred.d.ts,35,0,4,6,11,17,29,jsdeferred
raphael/raphael.d.ts,184,7,27,33,15,18,29,raphael
dojo/dojox.av.d.ts,634,6,42,147,7,23,29,dojo
dojo/dojox.editor.d.ts,1563,13,92,375,6,24,29,dojo
easeljs/easeljs.d.ts,286,13,23,74,8,26,29,easeljs
rx/rx.testing.d.ts,17,0,2,3,12,18,29,rx
hammerjs/hammerjs.d.ts,61,0,7,11,11,18,30,hammerjs
durandal/durandal.d.ts,149,7,14,37,9,25,30,durandal
dropboxjs/dropboxjs.d.ts,132,1,29,11,22,8,30,dropboxjs
sjcl/sjcl.d.ts,115,4,14,24,12,21,30,sjcl
dojo/dojox.treemap.d.ts,137,1,6,36,4,26,30,dojo
fastclick/fastclick.d.ts,10,0,0,3,0,30,30,fastclick
moment-timezone/moment-timezone.d.ts,10,0,3,0,30,0,30,moment-timezone
promise-pool/promise-pool.d.ts,10,0,1,2,10,20,30,promise-pool
requirejs/require.d.ts,10,0,2,1,20,10,30,requirejs
commander/commander.d.ts,30,1,4,6,13,20,30,commander
breeze/breeze.d.ts,221,9,29,47,13,21,30,breeze
dropzone/dropzone.d.ts,13,1,1,4,8,31,31,dropzone
fancybox/fancybox.d.ts,13,0,1,3,8,23,31,fancybox
tspromise/tspromise.d.ts,13,2,3,3,23,23,31,tspromise
dustjs-linkedin/dustjs-linkedin.d.ts,26,0,4,4,15,15,31,dustjs-linkedin
q/Q.d.ts,65,3,15,8,23,12,31,q
easy-table/easy-table.d.ts,16,0,1,4,6,25,31,easy-table
jstorage/jstorage.d.ts,16,0,2,3,13,19,31,jstorage
scroller/scroller.d.ts,16,0,0,5,0,31,31,scroller
videojs/videojs.d.ts,19,0,5,1,26,5,32,videojs
dojo/dojox.image.d.ts,294,3,21,75,7,26,32,dojo
angular-scenario/angular-scenario.d.ts,66,0,18,3,27,5,32,angular-scenario
angularjs/legacy/angular-scenario-1.2.d.ts,66,0,18,3,27,5,32,angularjs
joi/joi.d.ts,75,0,19,5,25,7,32,joi
gm/gm.d.ts,205,4,42,28,20,14,32,gm
dojo/dojox.embed.d.ts,102,1,7,27,7,26,32,dojo
timezone-js/timezone-js.d.ts,52,0,0,17,0,33,33,timezone-js
validator/validator.d.ts,52,0,0,17,0,33,33,validator
marionette/marionette.d.ts,167,0,6,49,4,29,33,marionette
asciify/asciify.d.ts,3,0,1,0,33,0,33,asciify
cordova/plugins/Camera.d.ts,3,0,0,1,0,33,33,cordova
cordova/plugins/DeviceMotion.d.ts,3,0,0,1,0,33,33,cordova
fibers/fibers.d.ts,3,0,0,1,0,33,33,fibers
gae.channel.api/gae.channel.api.d.ts,3,0,0,1,0,33,33,gae.channel.api
google.geolocation/google.geolocation.d.ts,3,0,0,1,0,33,33,google.geolocation
knockout.viewmodel/knockout.viewmodel.d.ts,3,0,0,1,0,33,33,knockout.viewmodel
lscache/lscache.d.ts,3,0,0,1,0,33,33,lscache
passport-local/passport-local.d.ts,3,0,0,1,0,33,33,passport-local
podcast/podcast.d.ts,3,0,0,1,0,33,33,podcast
promises-a-plus/promises-a-plus.d.ts,3,0,1,0,33,0,33,promises-a-plus
angular-hotkeys/angular-hotkeys.d.ts,6,0,2,0,33,0,33,angular-hotkeys
angularLocalStorage/angularLocalStorage.d.ts,6,0,0,2,0,33,33,angularLocalStorage
arbiter/Arbiter.d.ts,6,0,1,1,17,17,33,arbiter
dojo/dojox.timing.d.ts,6,0,1,1,17,17,33,dojo
flexSlider/flexSlider.d.ts,6,0,0,2,0,33,33,flexSlider
imagemagick/imagemagick.d.ts,6,0,2,0,33,0,33,imagemagick
lockfile/lockfile.d.ts,6,0,2,0,33,0,33,lockfile
qajax/qajax.d.ts,6,0,2,0,33,0,33,qajax
dojo/dojox.collections.d.ts,15,0,0,5,0,33,33,dojo
ip/ip.d.ts,15,0,0,5,0,33,33,ip
epiceditor/epiceditor.d.ts,21,0,0,7,0,33,33,epiceditor
socket.io/socket.io.d.ts,21,2,5,4,24,19,33,socket.io
google.analytics/ga.d.ts,24,2,4,6,17,25,33,google.analytics
mithril/mithril.d.ts,24,2,3,7,13,29,33,mithril
durandal/durandal-1.x.d.ts,27,0,3,6,11,22,33,durandal
dojo/dojox.gauges.d.ts,2781,21,147,809,5,29,34,dojo
dojo/dojox.layout.d.ts,1420,11,94,396,7,28,34,dojo
angularjs/legacy/angular-scenario-1.0.d.ts,62,0,18,3,29,5,34,angularjs
popcorn/popcorn.d.ts,59,3,6,17,10,29,34,popcorn
physijs/physijs.d.ts,79,0,3,24,4,30,34,physijs
angularfire/angularfire.d.ts,35,3,5,10,14,29,34,angularfire
swiper/swiper.d.ts,35,3,3,12,9,34,34,swiper
dojo/dijit.d.ts,12686,97,748,3722,6,29,34,dojo
dojo/dojox.form.d.ts,3865,24,204,1164,5,30,35,dojo
phonegap/phonegap.d.ts,92,1,1,32,1,35,35,phonegap
node-azure/azure.d.ts,293,13,95,20,32,7,35,node-azure
gruntjs/gruntjs.d.ts,86,5,24,11,28,13,35,gruntjs
createjs-lib/createjs-lib.d.ts,20,5,6,6,30,30,35,createjs-lib
parsimmon/parsimmon.d.ts,20,0,7,0,35,0,35,parsimmon
dojo/dojox.calendar.d.ts,2950,18,105,954,4,32,35,dojo
jake/jake.d.ts,51,4,8,14,16,27,35,jake
couchbase/couchbase.d.ts,45,0,16,0,36,0,36,couchbase
dcjs/dc.d.ts,14,0,1,4,7,29,36,dcjs
chrome/chrome-app.d.ts,50,0,6,12,12,24,36,chrome
sipml/sipml.d.ts,50,0,0,18,0,36,36,sipml
dojo/dojox.highlight.d.ts,100,1,8,29,8,29,36,dojo
dojo/dojox.widget.d.ts,3576,36,254,1072,7,30,36,dojo
dojo/dojox.mobile.d.ts,6830,68,477,2055,7,30,36,dojo
bootstrap.v3.datetimepicker/bootstrap.v3.datetimepicker.d.ts,11,2,4,2,36,18,36,bootstrap.v3.datetimepicker
xpath/xpath.d.ts,11,0,0,4,0,36,36,xpath
dojo/dojox.calc.d.ts,1314,14,84,409,6,31,36,dojo
express/express.d.ts,63,4,17,10,27,16,37,express
svgjs/svgjs.d.ts,99,6,18,25,18,25,37,svgjs
jquery.form/jquery.form.d.ts,8,1,1,3,13,38,38,jquery.form
vinyl/vinyl.d.ts,8,0,0,3,0,38,38,vinyl
jquery.tagsmanager/jquery.tagsmanager.d.ts,16,1,1,6,6,38,38,jquery.tagsmanager
meteor/meteor.d.ts,183,4,10,63,5,34,38,meteor
highcharts/highcharts.d.ts,58,0,20,2,34,3,38,highcharts
bunyan/bunyan.d.ts,21,7,7,8,33,38,38,bunyan
dojo/dojox.dgauges.d.ts,2784,37,203,898,7,32,38,dojo
jszip/jszip.d.ts,13,0,3,2,23,15,38,jszip
node-imap/imap.d.ts,18,0,7,0,39,0,39,node-imap
request/request.d.ts,36,9,11,12,31,33,39,request
docCookies/docCookies.d.ts,5,1,1,2,20,40,40,docCookies
jquery.livestampjs/jquery.livestampjs.d.ts,5,0,2,0,40,0,40,jquery.livestampjs
spin/spin.d.ts,5,0,0,2,0,40,40,spin
angularjs/angular-animate.d.ts,10,0,0,4,0,40,40,angularjs
Finch/Finch.d.ts,10,1,3,2,30,20,40,Finch
iscroll/iscroll-lite.d.ts,10,1,1,4,10,40,40,iscroll
knockout.deferred.updates/knockout.deferred.updates.d.ts,10,0,0,4,0,40,40,knockout.deferred.updates
knockout.es5/knockout.es5.d.ts,10,0,3,1,30,10,40,knockout.es5
node_zeromq/zmq.d.ts,15,2,4,4,27,27,40,node_zeromq
pg/pg.d.ts,15,2,5,3,33,20,40,pg
smoothie/smoothie.d.ts,15,0,0,6,0,40,40,smoothie
dojo/dojox.mvc.d.ts,2399,51,244,779,10,32,41,dojo
casperjs/casperjs.d.ts,174,5,10,66,6,38,41,casperjs
preloadjs/preloadjs.d.ts,22,2,3,8,14,36,41,preloadjs
hapi/hapi.d.ts,44,2,6,14,14,32,41,hapi
ix.js/ix.d.ts,44,5,12,11,27,25,41,ix.js
leaflet/leaflet.d.ts,414,40,51,159,12,38,41,leaflet
devextreme/dx.chartjs.d.ts,278,11,42,84,15,30,41,devextreme
bluebird/bluebird.d.ts,87,11,34,13,39,15,41,bluebird
minimatch/minimatch.d.ts,12,0,0,5,0,42,42,minimatch
underscore.string/underscore.string.d.ts,62,0,1,25,2,40,42,underscore.string
parse/parse.d.ts,207,5,18,74,9,36,42,parse
rx/rx-lite.d.ts,178,6,38,43,21,24,42,rx
restangular/restangular.d.ts,66,6,8,26,12,39,42,restangular
chrome/chrome.d.ts,347,7,20,135,6,39,43,chrome
auth0.widget/auth0.widget.d.ts,7,0,0,3,0,43,43,auth0.widget
content-type/content-type.d.ts,7,0,0,3,0,43,43,content-type
jquery.timer/jquery.timer.d.ts,7,0,0,3,0,43,43,jquery.timer
knockout.rx/knockout.rx.d.ts,7,1,1,3,14,43,43,knockout.rx
localForage/localForage.d.ts,7,0,3,0,43,0,43,localForage
bucks/bucks.d.ts,14,0,0,6,0,43,43,bucks
angularjs/angular.d.ts,213,8,56,45,26,21,44,angularjs
angularjs/legacy/angular-1.2.d.ts,188,9,50,42,27,22,44,angularjs
glob/glob.d.ts,9,0,1,3,11,33,44,glob
iscroll/iscroll-5-lite.d.ts,9,1,1,4,11,44,44,iscroll
mongodb/mongodb.d.ts,103,15,28,33,27,32,45,mongodb
winston/winston.d.ts,47,12,12,21,26,45,45,winston
filesystem/filesystem.d.ts,40,0,0,18,0,45,45,filesystem
dojo/dojox.main.d.ts,60,0,27,0,45,0,45,dojo
jquery.dataTables/jquery.dataTables.d.ts,42,6,8,17,19,40,45,jquery.dataTables
underscore/underscore.d.ts,344,46,85,117,25,34,45,underscore
fbsdk/fbsdk.d.ts,22,1,5,6,23,27,45,fbsdk
bluebird/bluebird-1.0.d.ts,77,11,33,13,43,17,45,bluebird
azure-mobile-services-client/AzureMobileServicesClient.d.ts,24,1,8,4,33,17,46,azure-mobile-services-client
firebase/firebase.d.ts,61,8,9,27,15,44,46,firebase
supertest/supertest.d.ts,63,1,3,27,5,43,46,supertest
gapi/gapi.d.ts,13,0,0,6,0,46,46,gapi
iscroll/iscroll.d.ts,13,1,1,6,8,46,46,iscroll
swfobject/swfobject.d.ts,13,0,0,6,0,46,46,swfobject
node/node-0.8.8.d.ts,405,17,26,178,6,44,46,node
angular-translate/angular-translate.d.ts,45,3,15,9,33,20,47,angular-translate
websocket/websocket.d.ts,49,2,17,8,35,16,47,websocket
angular-ui-bootstrap/angular-ui-bootstrap.d.ts,17,0,0,8,0,47,47,angular-ui-bootstrap
linq/linq.3.0.3-Beta4.d.ts,121,5,20,42,17,35,47,linq
angularjs/legacy/angular-1.0.d.ts,135,7,37,34,27,25,47,angularjs
sugar/sugar.d.ts,446,85,131,166,29,37,48,sugar
i18next/i18next.d.ts,21,2,4,8,19,38,48,i18next
expect.js/expect.js.d.ts,23,2,8,5,35,22,48,expect.js
nock/nock.d.ts,23,1,4,8,17,35,48,nock
passport/passport.d.ts,27,4,8,9,30,33,48,passport
when/when.d.ts,29,4,11,7,38,24,48,when
devextreme/dx.webappjs.d.ts,468,2,33,196,7,42,49,devextreme
webrtc/RTCPeerConnection.d.ts,35,0,1,16,3,46,49,webrtc
devextreme/dx.phonejs.d.ts,485,2,33,206,7,42,49,devextreme
angularjs/angular-sanitize.d.ts,2,0,0,1,0,50,50,angularjs
angularjs/legacy/angular-sanitize-1.2.d.ts,2,0,0,1,0,50,50,angularjs
clone/clone.d.ts,2,0,0,1,0,50,50,clone
csurf/csurf.d.ts,2,0,0,1,0,50,50,csurf
domo/domo.d.ts,2,1,1,1,50,50,50,domo
dotdotdot/dotdotdot.d.ts,2,0,0,1,0,50,50,dotdotdot
from/from.d.ts,2,0,1,0,50,0,50,from
fuse/fuse.d.ts,2,0,0,1,0,50,50,fuse
morgan/morgan.d.ts,2,1,1,1,50,50,50,morgan
purl/purl-jquery.d.ts,2,0,1,0,50,0,50,purl
svgjs.draggable/svgjs.draggable.d.ts,2,0,1,0,50,0,50,svgjs.draggable
threejs/detector.d.ts,2,0,0,1,0,50,50,threejs
threejs/three-orbitcontrols.d.ts,2,0,0,1,0,50,50,threejs
threejs/three-shaderpass.d.ts,2,0,0,1,0,50,50,threejs
threejs/three-trackballcontrols.d.ts,2,0,0,1,0,50,50,threejs
angular-notify/angular-notify.d.ts,4,0,1,1,25,25,50,angular-notify
canvasjs/canvasjs.d.ts,4,0,0,2,0,50,50,canvasjs
cordova/plugins/DeviceOrientation.d.ts,4,0,0,2,0,50,50,cordova
gldatepicker/gldatepicker.d.ts,4,0,1,1,25,25,50,gldatepicker
humane/humane.d.ts,4,0,1,1,25,25,50,humane
jsonwebtoken/jsonwebtoken.d.ts,4,0,2,0,50,0,50,jsonwebtoken
notify/notify.d.ts,4,2,2,2,50,50,50,notify
persona/persona.d.ts,4,0,2,0,50,0,50,persona
rx/rx.coincidence.d.ts,4,0,2,0,50,0,50,rx
sockjs/sockjs.d.ts,4,0,0,2,0,50,50,sockjs
velocity-animate/velocity-animate.d.ts,4,0,2,0,50,0,50,velocity-animate
elm/elm.d.ts,6,0,0,3,0,50,50,elm
gently/gently.d.ts,6,1,1,3,17,50,50,gently
gridfs-stream/gridfs-stream.d.ts,6,2,2,3,33,50,50,gridfs-stream
intercomjs/intercom.d.ts,6,0,0,3,0,50,50,intercomjs
nexpect/nexpect.d.ts,6,1,3,1,50,17,50,nexpect
parallel/parallel.d.ts,6,0,1,2,17,33,50,parallel
passport-strategy/passport-strategy.d.ts,6,0,1,2,17,33,50,passport-strategy
purl/purl.d.ts,6,0,3,0,50,0,50,purl
alertify/alertify.d.ts,10,0,0,5,0,50,50,alertify
bufferstream/bufferstream.d.ts,10,0,3,2,30,20,50,bufferstream
amqp-rpc/amqp-rpc.d.ts,12,0,0,6,0,50,50,amqp-rpc
iscroll/iscroll-5.d.ts,12,1,1,6,8,50,50,iscroll
peerjs/Peer.d.ts,14,0,3,4,21,29,50,peerjs
inflection/inflection.d.ts,16,0,0,8,0,50,50,inflection
hellojs/hellojs.d.ts,18,1,1,9,6,50,50,hellojs
chroma-js/chroma-js.d.ts,48,2,9,17,19,35,50,chroma-js
node/node-0.11.d.ts,468,35,69,203,15,43,51,node
urijs/URI.d.ts,75,0,32,6,43,8,51,urijs
d3/d3.d.ts,475,43,206,83,43,17,52,d3
node/node.d.ts,457,36,74,199,16,44,52,node
restify/restify.d.ts,52,0,7,20,13,38,52,restify
pouchDB/pouch.d.ts,23,0,12,0,52,0,52,pouchDB
mysql/mysql.d.ts,38,0,18,2,47,5,53,mysql
greensock/greensock.d.ts,123,0,0,65,0,53,53,greensock
adm-zip/adm-zip.d.ts,28,3,10,8,36,29,54,adm-zip
ix.js/l2o.d.ts,84,3,15,33,18,39,54,ix.js
rx/rx.binding-lite.d.ts,13,1,5,3,38,23,54,rx
jscrollpane/jscrollpane.d.ts,26,1,1,14,4,54,54,jscrollpane
siesta/siesta.d.ts,213,18,31,102,15,48,54,siesta
appframework/appframework.d.ts,103,6,24,38,23,37,54,appframework
angular-ui/angular-ui-router.d.ts,20,3,8,6,40,30,55,angular-ui
knockout.mapping/knockout.mapping.d.ts,9,0,2,3,22,33,56,knockout.mapping
mocha/mocha.d.ts,25,3,12,5,48,20,56,mocha
shelljs/shelljs.d.ts,25,0,14,0,56,0,56,shelljs
moment/moment.d.ts,93,2,49,6,53,6,57,moment
angular-http-auth/angular-http-auth.d.ts,7,0,0,4,0,57,57,angular-http-auth
bcrypt/bcrypt.d.ts,7,0,3,1,43,14,57,bcrypt
cookiejs/cookiejs.d.ts,7,3,4,3,57,43,57,cookiejs
cordova/plugins/WebSQL.d.ts,7,0,0,4,0,57,57,cordova
jqueryui/jqueryui.d.ts,59,5,30,9,51,15,58,jqueryui
mongoose/mongoose.d.ts,177,28,60,70,34,40,58,mongoose
node-ffi/node-ffi.d.ts,26,3,5,13,19,50,58,node-ffi
cordova/plugins/Globalization.d.ts,12,0,0,7,0,58,58,cordova
jquery.address/jquery.address.d.ts,24,0,14,0,58,0,58,jquery.address
jquery.pickadate/jquery.pickadate.d.ts,24,2,6,10,25,42,58,jquery.pickadate
rx/rx.d.ts,17,2,7,5,41,29,59,rx
mixpanel/mixpanel.d.ts,22,2,4,11,18,50,59,mixpanel
qunit/qunit.d.ts,49,4,6,27,12,55,59,qunit
cordova/plugins/FileTransfer.d.ts,5,0,0,3,0,60,60,cordova
ftdomdelegate/ftdomdelegate.d.ts,5,2,2,3,40,60,60,ftdomdelegate
jstree/jstree.d.ts,5,0,2,1,40,20,60,jstree
nodemailer/nodemailer.d.ts,10,0,1,5,10,50,60,nodemailer
ravenjs/ravenjs.d.ts,10,0,2,4,20,40,60,ravenjs
cryptojs/cryptojs.d.ts,75,16,21,40,28,53,60,cryptojs
zepto/zepto.d.ts,135,7,66,22,49,16,60,zepto
cheerio/cheerio.d.ts,51,4,21,14,41,27,61,cheerio
any-db/any-db.d.ts,13,2,3,7,23,54,62,any-db
cordova/plugins/Contacts.d.ts,13,0,0,8,0,62,62,cordova
knockout.postbox/knockout-postbox.d.ts,13,0,0,8,0,62,62,knockout.postbox
cordova/plugins/FileSystem.d.ts,21,0,0,13,0,62,62,cordova
browserify/browserify.d.ts,8,1,3,3,38,38,63,browserify
cordova/cordova.d.ts,8,2,2,5,25,63,63,cordova
flipsnap/flipsnap.d.ts,8,2,2,5,25,63,63,flipsnap
jquery.gridster/gridster.d.ts,8,2,2,5,25,63,63,jquery.gridster
superagent/superagent.d.ts,46,0,3,26,7,57,63,superagent
sencha_touch/SenchaTouch-2.2.1.d.ts,7930,0,0,5004,0,63,63,sencha_touch
sencha_touch/SenchaTouch.d.ts,8478,0,0,5354,0,63,63,sencha_touch
angularjs/legacy/angular-mocks-1.0.d.ts,30,14,18,15,60,50,63,angularjs
big.js/big.js.d.ts,25,0,16,0,64,0,64,big.js
yui/yui-test.d.ts,53,2,5,31,9,58,64,yui
gulp-util/gulp-util.d.ts,14,1,3,7,21,50,64,gulp-util
phantom/phantom.d.ts,28,2,4,16,14,57,64,phantom
lodash/lodash.d.ts,230,76,115,109,50,47,64,lodash
semver/semver.d.ts,31,0,0,20,0,65,65,semver
fs-extra/fs-extra.d.ts,96,15,22,55,23,57,65,fs-extra
ref/ref.d.ts,34,3,9,16,26,47,65,ref
linq/linq.d.ts,102,32,59,40,58,39,66,linq
backbone/backbone.d.ts,138,8,20,79,14,57,66,backbone
extjs/ExtJS.d.ts,6561,0,0,4351,0,66,66,extjs
dot/dot.d.ts,3,0,0,2,0,67,67,dot
jquery.simplemodal/jquery.simplemodal.d.ts,3,1,1,2,33,67,67,jquery.simplemodal
through2/through2.d.ts,3,1,1,2,33,67,67,through2
watch/watch.d.ts,3,0,2,0,67,0,67,watch
xml2js/xml2js.d.ts,3,0,1,1,33,33,67,xml2js
bootbox/bootbox.d.ts,6,4,4,4,67,67,67,bootbox
jquery.timeago/jquery.timeago.d.ts,6,0,4,0,67,0,67,jquery.timeago
mu2/mu2.d.ts,6,0,2,2,33,33,67,mu2
rtree/rtree.d.ts,6,0,0,4,0,67,67,rtree
rx/rx.backpressure-lite.d.ts,6,0,0,4,0,67,67,rx
xsockets/XSockets.d.ts,6,0,0,4,0,67,67,xsockets
es6-promise/es6-promise.d.ts,9,4,4,6,44,67,67,es6-promise
auth0/auth0.d.ts,15,0,0,10,0,67,67,auth0
pickadate/pickadate.d.ts,15,4,6,8,40,53,67,pickadate
rx/rx.experimental.d.ts,15,0,8,2,53,13,67,rx
chai-datetime/chai-datetime.d.ts,18,0,0,12,0,67,67,chai-datetime
node-ffi/node-ffi-buffer.d.ts,21,0,4,10,19,48,67,node-ffi
big-integer/big-integer.d.ts,30,0,20,0,67,0,67,big-integer
swig/swig.d.ts,30,3,5,18,17,60,67,swig
jquery/jquery.d.ts,226,36,117,71,52,31,67,jquery
nconf/nconf.d.ts,49,2,2,33,4,67,67,nconf
underscore-ko/underscore-ko.d.ts,72,8,16,41,22,57,68,underscore-ko
jquerymobile/jquerymobile.d.ts,35,0,17,7,49,20,69,jquerymobile
jquery.color/jquery.color.d.ts,16,0,10,1,63,6,69,jquery.color
gamequery/gamequery.d.ts,29,9,11,18,38,62,69,gamequery
jdataview/jdataview.d.ts,58,0,1,39,2,67,69,jdataview
angularjs/angular-mocks.d.ts,33,14,19,18,58,55,70,angularjs
angularjs/legacy/angular-mocks-1.2.d.ts,33,14,19,18,58,55,70,angularjs
jsrender/jsrender.d.ts,10,0,4,3,40,30,70,jsrender
microsoft-live-connect/microsoft-live-connect.d.ts,17,0,0,12,0,71,71,microsoft-live-connect
nodeunit/nodeunit.d.ts,17,0,0,12,0,71,71,nodeunit
rx-jquery/rx.jquery.d.ts,41,0,0,29,0,71,71,rx-jquery
amplifyjs/amplifyjs.d.ts,7,3,4,4,57,57,71,amplifyjs
cordova/plugins/InAppBrowser.d.ts,7,1,5,1,71,14,71,cordova
firebase/firebase-simplelogin.d.ts,7,0,0,5,0,71,71,firebase
gulp/gulp.d.ts,7,2,4,3,57,43,71,gulp
jquery.pjax/jquery.pjax.d.ts,7,3,3,5,43,71,71,jquery.pjax
simple-cw-node/simple-cw-node.d.ts,7,0,5,0,71,0,71,simple-cw-node
howlerjs/howler.d.ts,21,0,4,11,19,52,71,howlerjs
xregexp/xregexp.d.ts,18,3,5,11,28,61,72,xregexp
jquery.pjax.falsandtru/jquery.pjax.d.ts,11,2,5,5,45,45,73,jquery.pjax.falsandtru
interactjs/interact.d.ts,48,0,30,5,63,10,73,interactjs
gapi.youtube/gapi.youtube.d.ts,26,0,0,19,0,73,73,gapi.youtube
dojo/dojox.lang.d.ts,231,0,163,10,71,4,75,dojo
angularjs/legacy/angular-resource-1.0.d.ts,4,0,2,1,50,25,75,angularjs
cordova/plugins/Dialogs.d.ts,4,0,0,3,0,75,75,cordova
mousetrap/mousetrap.d.ts,4,2,2,3,50,75,75,mousetrap
threejs/three-renderpass.d.ts,4,0,0,3,0,75,75,threejs
modernizr/modernizr.d.ts,8,1,4,3,50,38,75,modernizr
noble/noble.d.ts,28,0,6,15,21,54,75,noble
wolfy87-eventemitter/wolfy87-eventemitter.d.ts,21,1,16,1,76,5,76,wolfy87-eventemitter
ws/ws.d.ts,21,3,8,11,38,52,76,ws
i18n-node/i18n-node.d.ts,13,0,8,2,62,15,77,i18n-node
globalize/globalize.d.ts,9,2,4,5,44,56,78,globalize
yargs/yargs.d.ts,27,1,18,4,67,15,78,yargs
should/should.d.ts,49,3,5,37,10,76,80,should
cordova/plugins/MediaCapture.d.ts,5,0,0,4,0,80,80,cordova
promptly/promptly.d.ts,5,4,4,4,80,80,80,promptly
rx/rx.async-lite.d.ts,5,3,3,4,60,80,80,rx
universal-analytics/universal-analytics.d.ts,10,0,6,2,60,20,80,universal-analytics
marked/marked.d.ts,6,0,2,3,33,50,83,marked
sockjs-node/sockjs-node.d.ts,6,0,2,3,33,50,83,sockjs-node
chartjs/chart.d.ts,7,0,0,6,0,86,86,chartjs
power-assert/power-assert.d.ts,14,2,2,12,14,86,86,power-assert
levelup/levelup.d.ts,21,4,6,16,29,76,86,levelup
tape/tape.d.ts,57,0,0,49,0,86,86,tape
bl/bl.d.ts,22,0,0,19,0,86,86,bl
needle/needle.d.ts,8,0,6,1,75,13,88,needle
ref-struct/ref-struct.d.ts,8,2,3,6,38,75,88,ref-struct
ref-union/ref-union.d.ts,8,2,3,6,38,75,88,ref-union
date.format.js/date.format.d.ts,53,11,43,15,81,28,89,date.format.js
angularjs/legacy/angular-animate-1.2.d.ts,9,0,0,8,0,89,89,angularjs
rx/rx.aggregates.d.ts,27,0,6,18,22,67,89,rx
chai/chai.d.ts,84,7,9,74,11,88,90,chai
assert/assert.d.ts,13,2,2,12,15,92,92,assert
msnodesql/msnodesql.d.ts,13,1,5,8,38,62,92,msnodesql
jquery.bbq/jquery.bbq.d.ts,15,5,5,14,33,93,93,jquery.bbq
rx/rx.time.d.ts,18,5,6,16,33,89,94,rx
sqlite3/sqlite3.d.ts,21,11,12,19,57,90,95,sqlite3
redis/redis.d.ts,135,128,130,128,96,95,96,redis
browser-harness/browser-harness.d.ts,73,2,6,68,8,93,99,browser-harness
assertion-error/assertion-error.d.ts,1,0,0,1,0,100,100,assertion-error
bootstrap.paginator/bootstrap.paginator.d.ts,1,0,1,0,100,0,100,bootstrap.paginator
bootstrap.timepicker/bootstrap.timepicker.d.ts,1,0,1,0,100,0,100,bootstrap.timepicker
chosen/chosen.jquery.d.ts,1,0,1,0,100,0,100,chosen
compression/compression.d.ts,1,0,0,1,0,100,100,compression
cookie-parser/cookie-parser.d.ts,1,0,0,1,0,100,100,cookie-parser
cordova/plugins/NetworkInformation.d.ts,1,1,1,1,100,100,100,cordova
detect-indent/detect-indent.d.ts,1,0,0,1,0,100,100,detect-indent
empower/empower.d.ts,1,0,0,1,0,100,100,empower
express-session/express-session.d.ts,1,0,0,1,0,100,100,express-session
findup-sync/findup-sync.d.ts,1,1,1,1,100,100,100,findup-sync
gapi.pagespeedonline/gapi.pagespeedonline.d.ts,1,0,0,1,0,100,100,gapi.pagespeedonline
gapi.youtubeAnalytics/gapi.youtubeAnalytics.d.ts,1,0,0,1,0,100,100,gapi.youtubeAnalytics
glob-stream/glob-stream.d.ts,1,1,1,1,100,100,100,glob-stream
html2canvas/html2canvas.d.ts,1,0,1,0,100,0,100,html2canvas
icheck/icheck.d.ts,1,1,1,1,100,100,100,icheck
ion.rangeSlider/ion.rangeSlider-1.9.1.d.ts,1,0,1,0,100,0,100,ion.rangeSlider
jquery.are-you-sure/jquery.are-you-sure.d.ts,1,0,1,0,100,0,100,jquery.are-you-sure
jquery.autosize/jquery.autosize.d.ts,1,0,1,0,100,0,100,jquery.autosize
jquery.contextMenu/jquery.contextMenu.d.ts,1,0,0,1,0,100,100,jquery.contextMenu
jquery.cycle/jquery.cycle.d.ts,1,1,1,1,100,100,100,jquery.cycle
jquery.finger/jquery.finger.d.ts,1,0,1,0,100,0,100,jquery.finger
jquery.joyride/jquery.joyride.d.ts,1,0,1,0,100,0,100,jquery.joyride
jquery.jsignature/jquery.jsignature.d.ts,1,0,1,0,100,0,100,jquery.jsignature
jquery.notifyBar/jquery.notifyBar.d.ts,1,0,0,1,0,100,100,jquery.notifyBar
jquery.rowGrid/jquery.rowGrid.d.ts,1,1,1,1,100,100,100,jquery.rowGrid
jquery.simplePagination/jquery.simplePagination.d.ts,1,1,1,1,100,100,100,jquery.simplePagination
jquery.simulate/jquery.simulate.d.ts,1,0,0,1,0,100,100,jquery.simulate
jquery.slimScroll/jquery.slimScroll.d.ts,1,0,1,0,100,0,100,jquery.slimScroll
jquery.sortElements/jquery.sortElement.d.ts,1,0,0,1,0,100,100,jquery.sortElements
jquery.superLink/jquery.superLink.d.ts,1,0,0,1,0,100,100,jquery.superLink
jquery.tile/jquery.tile.d.ts,1,0,0,1,0,100,100,jquery.tile
jquery.timepicker/jquery.timepicker.d.ts,1,0,1,0,100,0,100,jquery.timepicker
jquery.tooltipster/jquery.tooltipster.d.ts,1,0,1,0,100,0,100,jquery.tooltipster
jquery.transit/jquery.transit.d.ts,1,0,1,0,100,0,100,jquery.transit
jquery.ui.datetimepicker/jquery.ui.datetimepicker.d.ts,1,0,1,0,100,0,100,jquery.ui.datetimepicker
js-url/js-url.d.ts,1,0,1,0,100,0,100,js-url
jsesc/jsesc.d.ts,1,0,0,1,0,100,100,jsesc
mCustomScrollbar/mCustomScrollbar.d.ts,1,1,1,1,100,100,100,mCustomScrollbar
md5/md5.d.ts,1,0,0,1,0,100,100,md5
method-override/method-override.d.ts,1,1,1,1,100,100,100,method-override
minimist/minimist.d.ts,1,0,0,1,0,100,100,minimist
mousetrap/mousetrap-global-bind.d.ts,1,1,1,1,100,100,100,mousetrap
ncp/ncp.d.ts,1,0,1,0,100,0,100,ncp
open/open.d.ts,1,0,0,1,0,100,100,open
opn/opn.d.ts,1,1,1,1,100,100,100,opn
passport-facebook/passport-facebook.d.ts,1,0,0,1,0,100,100,passport-facebook
pgwmodal/pgwmodal.d.ts,1,0,1,0,100,0,100,pgwmodal
riotjs/riotjs-render.d.ts,1,0,0,1,0,100,100,riotjs
sharedworker/SharedWorker.d.ts,1,0,0,1,0,100,100,sharedworker
through/through.d.ts,1,0,0,1,0,100,100,through
trunk8/trunk8.d.ts,1,1,1,1,100,100,100,trunk8
yosay/yosay.d.ts,1,0,0,1,0,100,100,yosay
checksum/checksum.d.ts,2,0,1,1,50,50,100,checksum
cookie/cookie.d.ts,2,0,0,2,0,100,100,cookie
cordova/plugins/BatteryStatus.d.ts,2,2,2,2,100,100,100,cordova
esprima/esprima.d.ts,2,0,0,2,0,100,100,esprima
Headroom/headroom.d.ts,2,0,0,2,0,100,100,Headroom
jquery.base64/jquery.base64.d.ts,2,0,0,2,0,100,100,jquery.base64
jquery.cookie/jquery.cookie.d.ts,2,0,2,0,100,0,100,jquery.cookie
jquery.jnotify/jquery.jnotify.d.ts,2,0,0,2,0,100,100,jquery.jnotify
jquery.leanModal/jquery.leanModal.d.ts,2,0,2,0,100,0,100,jquery.leanModal
jquery.postMessage/jquery.postMessage.d.ts,2,2,2,2,100,100,100,jquery.postMessage
jquery.scrollTo/jquery.scrollTo.d.ts,2,2,2,2,100,100,100,jquery.scrollTo
jquery.tinyscrollbar/jquery.tinyscrollbar.d.ts,2,0,0,2,0,100,100,jquery.tinyscrollbar
jwt-simple/jwt-simple.d.ts,2,0,0,2,0,100,100,jwt-simple
knockout.mapper/knockout.mapper.d.ts,2,0,0,2,0,100,100,knockout.mapper
mkdirp/mkdirp.d.ts,2,0,1,1,50,50,100,mkdirp
multiparty/multiparty.d.ts,2,0,0,2,0,100,100,multiparty
q-retry/q-retry.d.ts,2,0,2,0,100,0,100,q-retry
rx/rx.async.d.ts,2,1,1,2,50,100,100,rx
spectrum/spectrum.d.ts,2,1,2,1,100,50,100,spectrum
toastr/toastr.d.ts,2,0,2,0,100,0,100,toastr
update-notifier/update-notifier.d.ts,2,0,0,2,0,100,100,update-notifier
yui/yui.d.ts,2,1,1,2,50,100,100,yui
any-db-transaction/any-db-transaction.d.ts,3,1,1,3,33,100,100,any-db-transaction
foundation/foundation.d.ts,3,0,3,0,100,0,100,foundation
gapi.translate/gapi.translate.d.ts,3,0,0,3,0,100,100,gapi.translate
gapi.urlshortener/gapi.urlshortener.d.ts,3,0,0,3,0,100,100,gapi.urlshortener
jqrangeslider/jqrangeslider.d.ts,3,0,3,0,100,0,100,jqrangeslider
jquery.blockUI/jquery.blockUI.d.ts,3,0,1,2,33,67,100,jquery.blockUI
JSONStream/JSONStream.d.ts,3,0,3,0,100,0,100,JSONStream
knockout-secure-binding/knockout-secure-binding.d.ts,3,0,0,3,0,100,100,knockout-secure-binding
bootstrap.datepicker/bootstrap.datepicker.d.ts,4,2,3,3,75,75,100,bootstrap.datepicker
node-uuid/node-uuid.d.ts,4,4,4,4,100,100,100,node-uuid
select2/select2.d.ts,4,2,4,2,100,50,100,select2
vinyl-fs/vinyl-fs.d.ts,4,3,3,4,75,100,100,vinyl-fs
accounting/accounting.d.ts,5,3,3,5,60,100,100,accounting
body-parser/body-parser.d.ts,5,0,0,5,0,100,100,body-parser
fullCalendar/fullCalendar.d.ts,5,1,1,5,20,100,100,fullCalendar
chai-fuzzy/chai-fuzzy.d.ts,6,0,0,6,0,100,100,chai-fuzzy
js-yaml/js-yaml.d.ts,6,0,0,6,0,100,100,js-yaml
text-encoding/text-encoding.d.ts,6,0,0,6,0,100,100,text-encoding
object-path/object-path.d.ts,8,4,8,4,100,50,100,object-path
rx/rx.time-lite.d.ts,8,5,5,8,63,100,100,rx
bootstrap/bootstrap.d.ts,12,5,10,7,83,58,100,bootstrap
angularjs/angular-resource.d.ts,14,1,13,2,93,14,100,angularjs
angularjs/legacy/angular-resource-1.2.d.ts,14,1,13,2,93,14,100,angularjs
add2home/add2home.d.ts,0,0,0,0,,,,add2home
ansicolors/ansicolors.d.ts,0,0,0,0,,,,ansicolors
cordova-ionic/cordova-ionic.d.ts,0,0,0,0,,,,cordova-ionic
cordova/plugins/Device.d.ts,0,0,0,0,,,,cordova
createjs/createjs.d.ts,0,0,0,0,,,,createjs
dojo/dojox.flash.d.ts,0,0,0,0,,,,dojo
dojo/dojox.help.d.ts,0,0,0,0,,,,dojo
dojo/dojox.jq.d.ts,0,0,0,0,,,,dojo
dojo/dojox.robot.d.ts,0,0,0,0,,,,dojo
dojo/dojox.sql.d.ts,0,0,0,0,,,,dojo
dojo/dojox.storage.d.ts,0,0,0,0,,,,dojo
dojo/dojox.xml.d.ts,0,0,0,0,,,,dojo
es6-promises/es6-promises.d.ts,0,0,0,0,,,,es6-promises
geojson/geojson.d.ts,0,0,0,0,,,,geojson
graceful-fs/graceful-fs.d.ts,0,0,0,0,,,,graceful-fs
jquery.clientSideLogging/jquery.clientSideLogging.d.ts,0,0,0,0,,,,jquery.clientSideLogging
knockout.amd.helpers/knockout-amd-helpers.d.ts,0,0,0,0,,,,knockout.amd.helpers
mozilla-spidermonkey-parser-api/mozilla-spidermonkey-parser-api.d.ts,0,0,0,0,,,,mozilla-spidermonkey-parser-api
phonejs/dx.phonejs.d.ts,0,0,0,0,,,,phonejs
rickshaw/rickshaw.d.ts,0,0,0,0,,,,rickshaw
rx/rx.all.d.ts,0,0,0,0,,,,rx
rx/rx.backpressure.d.ts,0,0,0,0,,,,rx
rx/rx.binding.d.ts,0,0,0,0,,,,rx
rx/rx.lite.d.ts,0,0,0,0,,,,rx
threejs/three-copyshader.d.ts,0,0,0,0,,,,threejs
touch-events/touch-events.d.ts,0,0,0,0,,,,touch-events
\end{filecontents*}
\csvreader[filter=\equal{\fullname}{ace/ace.d.ts} \OR
                  \equal{\fullname}{fabricjs/fabricjs.d.ts} \OR
                  \equal{\fullname}{jquery/jquery.d.ts} \OR
                  \equal{\fullname}{underscore/underscore.d.ts} \OR
                  \equal{\fullname}{pixi/pixi.d.ts} \OR
                  \equal{\fullname}{box2d/box2dweb.d.ts} \OR
                  \equal{\fullname}{leaflet/leaflet.d.ts} \OR
                  \equal{\fullname}{threejs/three.d.ts} \OR
                  \equal{\fullname}{d3/d3.d.ts} \OR
                  \equal{\fullname}{sugar/sugar.d.ts},
           late after line=\\,
           late after last line=\\\hline]
          {overload-table-full.csv}
          {1=\fullname, 2=\funcs, 6=\pctOverload, 7=\pctOptArg, 8=\pctEither, 9=\name}
          {\texttt{\name}&\funcs&\pctOverload&\pctOptArg&\textbf{\pctEither}}
\end{tabular}
\end{minipage}
\hfill
\begin{minipage}{0.49\textwidth}
\includegraphics[width=\textwidth]{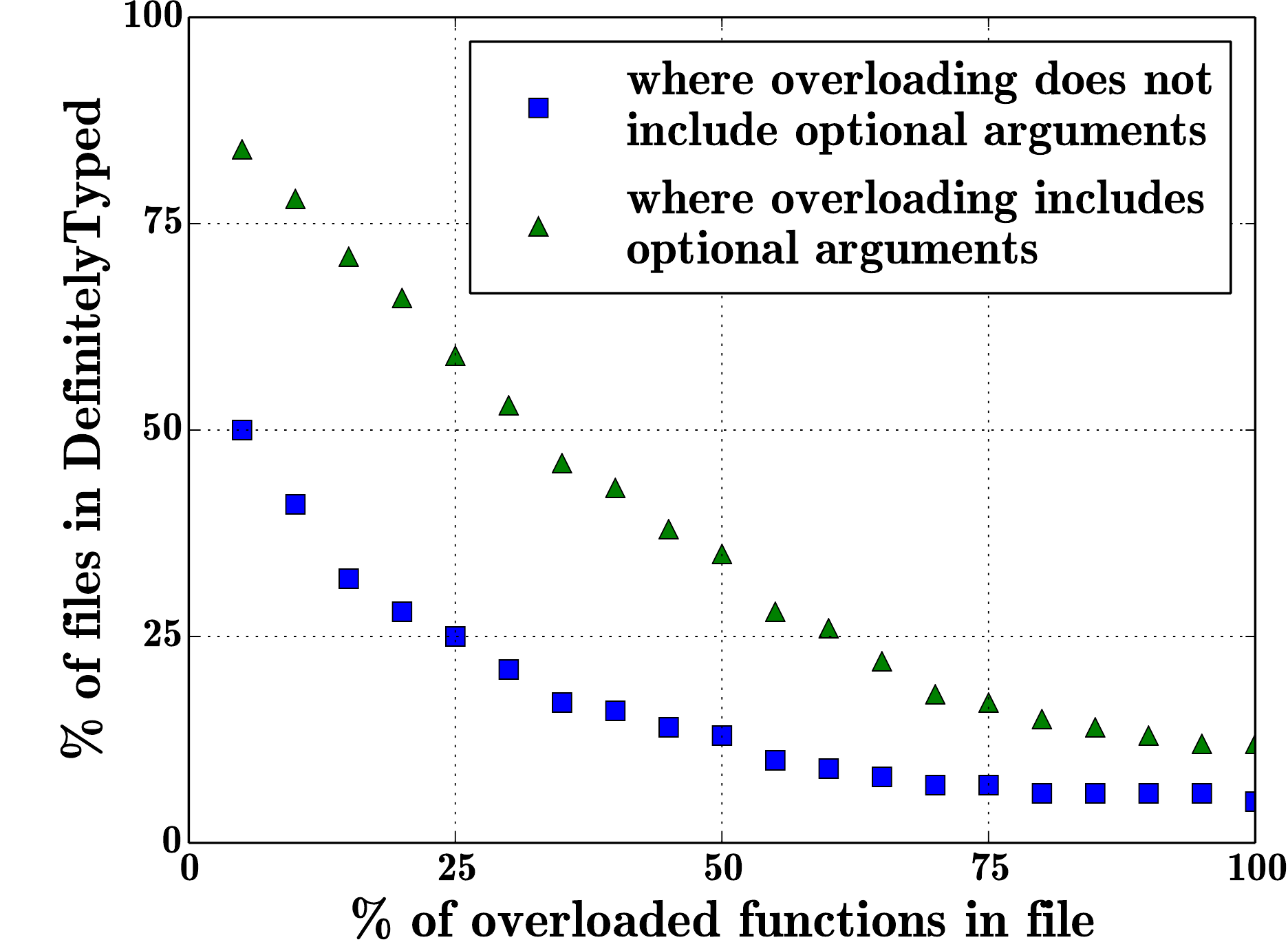}
\end{minipage}
\caption{The prevalence of value-based overloading. 
         \textbf{(L)} Libraries from \cite{MollerOOPSLA14}:
         \textbf{\#Funs} is the number of functions in the signature,
         \textbf{\%Ovl} is \%-functions with \emph{multiple} signatures, 
         \textbf{\%Opt} is \%-functions with \emph{optional} arguments, and
         \textbf{\%Any} is \%-functions with either of these features.
         \textbf{(R)} Overloading across \emph{all} files in 
         \textsc{DefinitelyTyped}. 
         A point $(x, y)$ means $y\%$ of 
         files have more than $x\%$ overloaded functions.}
\label{fig:overload-graph}
\end{figure}

The signatures in Definitely Typed have not been soundly checked 
against\footnote{Feldthaus \etal~\cite{MollerOOPSLA14} describe an effective but unsound inconsistency detector.}
their implementations. 
Hence, it is possible that they mischaracterize 
the semantics of the actual code, but modulo this caveat,
we believe the study demonstrates that value-based overloading
is ubiquitous,
and so to soundly and statically analyze dynamic languages, it 
is crucial that we develop techniques that can precisely and 
flexibly account for it.

\subsection{Refinement Types}\label{sec:overview:refinements}

\subparagraph*{Types and Refinements.} 
A basic refinement type $\rtype$
is a basic type, \eg @number@, refined with a logical 
formula from an SMT decidable logic -- for the purposes 
of this paper, the quantifier-free logic of uninterpreted 
functions and linear integer arithmetic (\decidablelogic~\cite{SMTLIB2}). 
For example, @{v:number | v != 0}@ describes the \emph{subset} 
of numbers that are non-zero.
We write $\type$ to abbreviate the trivially refined type
$\reftp{\type}{\rtrue}$, \eg $\tcnumber$ is 
an abbreviation for $\reftp{\tcnumber}{\rtrue}$.

\subparagraph*{Summaries: Function Types.}  
We can specify the behavior of functions with refined function 
types, of the form 
$$
\funtypen{x}{\rtype}{\rtype}
$$
where arguments are named $x_i$ and have types $\rtype_i$ 
and the output is a $\rtype$.
In essence, the \emph{input} types $\rtype_i$ specify the 
function's preconditions, and the \emph{output} type $\rtype$ 
describes the postcondition. Furthermore, each input type and the
output type can \emph{refer to} the arguments $x_i$ which yields
precise function contracts. 
For example,
$$\funtype{\bind{x}{0 \leq x}}{\reftp{\tcnumber}{x < \vv}}$$
is a function type that describes functions that \emph{require} 
a non-negative input, and \emph{ensure} that the output is greater
than the input.

\subparagraph*{Example.} 
Returning to $\tnegate$ in 
Figure~\ref{fig:negate:ts}, we can define two refinements of @number@:
\begin{code}
  type tt = {v:number | v != 0}    // "truthy" numbers
  type ff = {v:number | v  = 0}    // "falsy"  numbers
\end{code}
which are used to specify a refined type for $\tnegate$ 
shown on the left in Figure~\ref{fig:negate:source}.

\subparagraph*{Problem: A Circular Dependency.}
While it is easy enough to specify a type signature, 
it is another matter to verify it, and yet another matter 
to ensure soundness. 
The challenge is that value-based overloading introduces a
circular dependency between types and refinements. 
The soundness 
of basic types requires (\ie is established by) the refinements, 
while the refinements themselves require (\ie are attached to) 
basic types.
In classical refinement systems like DML~\cite{XiPfenning99}, 
basic types are established \emph{without} requiring 
refinements. A classical refinement system is thus a conservative
extension of the corresponding non-refined language, \ie 
removing the refinements from a DML program, 
yields valid, well-typed ML.
Unfortunately, value-based overloading removes this 
crucial property, posing a circular 
dependency between types and refinements.

\subparagraph*{Solution: Two-Phase Checking.}
We break the cycle by typing programs in two phases.
In the first, we \emph{trust} the basic types are correct
and use them (ignoring the refinements) to elaborate 
source programs into a target overloading-free language.
Inevitably, value-based overloading leads to ``errors''
when typing certain sub-expressions in the wrong context, 
\eg subtracting a @boolean@-valued @x@ from @0@. 
Instead of rejecting the program, the elaboration wraps 
ill-typed expressions with $\tdead$-casts, which are 
assertions stating the program is well-typed \emph{assuming} 
those expressions are dead code.
In the second phase we reuse classical refinement typing 
techniques to \emph{verify} that the $\tdead$-casts are indeed 
unreachable, thereby discharging the assumptions made in the first phase.

\begin{figure}[t!]
\begin{minipage}[t]{0.49\textwidth}
\begin{code}
neg :: (tt, number) => number
    /\ (ff, boolean) => boolean
function neg(flag, x) {
  if (flag) return 0-x;
  return !x;
} 






var a = neg(1,1);    //OK
var b = neg(0,true); //OK
var c = neg(0,1);    //ERR
var d = neg(1,true); //ERR
\end{code}
\end{minipage}
\hfill
\begin{minipage}[t]{0.49\textwidth}
\begin{code}
neg#1 :: (tt, number) => number
function neg#1(flag, x) {
  if (flag) return 0-x;
  return !DEAD(x);
} 
neg#2 :: (ff, boolean) => boolean
function neg#2(flag, x) {
  if (flag) return 0-DEAD(x);
  return !x;
} 
var neg = (neg#1, neg#2);

var a = fst(neg)(1,1);    //OK
var b = snd(neg)(0,true); //OK
var c = fst(neg)(0,1);    //ERR
var d = snd(neg)(1,true); //ERR
\end{code}
\end{minipage}
\caption{Source program (l) and target (r) resulting from first phase elaboration.}
\label{fig:negate:source} 
\label{fig:negate:target} 
\end{figure}

\subsection{Phase 1: Trust}

The first phase 
\emph{elaborates} the source program into an equivalent typed target
language with two key properties:
First, the target program is simply typed -- \ie has 
\emph{no} union or intersection types, but just classical
ML-style sums and products.
Second, source-level type errors are elaborated to 
target-level $\tdead$-casts.
The right side of Figure~\ref{fig:negate:target} shows the elaboration 
of the source from the left side. 
While we formalize the elaboration declaratively using
a single judgment form (\S~\ref{sec:elaboration}), it 
comprises two different steps. 
Critically, each step, and hence the entire first phase, is
\emph{independent} of the refinements -- they are simply 
carried along unchanged.

\subparagraph*{A. Clone.} 
In the first step, we create 
separate clones of each overloaded function, where 
each clone is assigned a single conjunct of the 
original overloaded type.
For example, we create two clones $\tnegateone$ and 
$\tnegatetwo$ respectively typed using the two conjuncts 
of the original $\tnegate$.
The binder $\tnegate$ is replaced with a \emph{tuple} 
of its clones. 
Finally, each use of $\tnegate$ extracts the appropriate 
element from the tuple before issuing the call. 

Since the trust phase must be independent of refinements, the overload
resolution in this step uses \emph{only} the basic types at the
call-site to determine which of the two clones to invoke.
For example, in the assignment to @a@, the source call @neg(1,1)@ --
which passes in two @number@ values, and hence, matches the first
overload (conjunct) -- is elaborated to the target call @fst(neg)(1,1)@.
In the assignment to @d@, the source call @neg(1,true)@ -- which
passes in a @number@ and a @boolean@, and hence matches the second
overload -- is elaborated to the target call @snd(neg)(1,true)@, even
though @1@ does \emph{not} have the refined type @ff@.

\subparagraph*{B: Cast.} 
In the second step we check -- using classical, unrefined 
type checking -- that each clone adheres to its specified
type.
Unlike under usual intersection typing \cite{forsythe,Dunfield2012},
in our context these checks almost surely ``fail''.
For example, $\tnegateone$ \emph{does not} type-check 
as the parameter @x@ has type @number@ and so we cannot compute 
@!x@. Similarly, $\tnegatetwo$ fails because @x@ has type @boolean@ 
and so @0-x@ is erroneous. 
Rather than reject the program, we wrap such failures with 
$\tdead$-casts.
For example, the above occurrences of $\tvarx$ elaborate to
$\tdead(\tvarx)$ on the right in Figure~\ref{fig:negate:target}.

Intuitively, the \emph{value relationships} established at
the call-sites and guards ensure that the failures will not 
happen at run-time. However, recall that the first phase's 
goal is to decouple reasoning about types from reasoning 
about values. Hence, we just \emph{trust} all the types 
but use $\tdead$-casts to \emph{explicate} the value-relationship 
obligations that are needed to establish typing: namely 
that the $\tdead$-casts are indeed dead code.

\subsection{Phase 2: Verify}

The second phase takes as input the elaborated program emitted by the
first phase, which is essentially a classical \emph{well-typed}
ML program with assertions and without any value-overloading.
Hence, the second phase can use any existing program
logic~\cite{Floyd67,JML}, refinement
typing~\cite{XiPfenning99,Knowles10,LiquidPLDI08,GordonTOPLAS2011},
or contracts \& abstract interpretation~\cite{vanhorn14} 
to check that the target's assertions never fail, 
which, we prove, ensures that the source is type-safe.

To analyze programs with closures, collections and polymorphism,
(\eg @minIndex@ from Figure~\ref{fig:reduce}) we perform the 
second phase using the refinement types that are carried over 
unchanged by the elaboration process of the first phase. 
Intuitively, refinement typing can be viewed as a 
generalization of classical program logics where 
\emph{assertions} are generalized to type bindings, 
and the rule of \emph{consequence} is generalized as subtyping.
While refinement typing is a previously known technique, 
to make the paper self-contained, we illustrate how the 
second phase verifies the @DEAD@-casts in Figure~\ref{fig:negate:target}.

\subparagraph*{Refinement Type Checking.}
A refinement type checker works by building
up an \emph{environment} of type bindings that describe 
the machine state at each program point, and by checking that 
at each call-site, the actual argument's type is a refined 
\emph{subtype} of the expected type for the callee, under 
the context described by the environment at that site. 
The subtyping relation for basic types is converted to a logical
\emph{verification condition} whose validity is checked by an SMT
solver.
The subtyping relation for \emph{compound} types (\eg functions, 
collections) is decomposed, via co- 
and contra-variant subtyping rules, into subtyping constraints 
over \emph{basic} types, which can be discharged as above.

\subparagraph*{Typing $\tdead$-Casts.}
To use a standard refinement type checker for the second
phase of verification, we only need to treat @DEAD@ as a 
primitive operation with the refined type:
$$\tdead :: \funsig{A, B}{\reftp{A}{\rfalse}}{B}$$
That is, we assign $\tdead$ the \emph{precondition} $\rfalse$ which
states there are \emph{no} valid inputs for it, \ie that it should
never be called (akin to @assert(false)@ in other settings).

\subparagraph*{Environments.}
To verify $\tdead$-casts, the refinement type checker builds up
an environment of type binders describing \emph{variables} and
\emph{branch conditions} that are in scope at each program point.
For example, the $\tdead$ call in $\tnegateone$, has the environment:
\begin{align}
\Gamma_1 \ \doteq & \ \bind{\tflag}{\tpos},\ \bind{\tvarx}{\tcnumber},\ \bind{\tvarg_1}{\tgrd{\tflag = 0}} \label{eq:gamma:negate1}\\
\intertext{where the first two bindings are the function parameters, 
whose types are the input types. The third binding is from the ``else'' 
branch of the $\tflag$ test, asserting the branch condition $\tflag$ 
is ``falsy'' \ie equals $0$. At the $\tdead$ call in $\tnegatetwo$ 
the environment is: }
\Gamma_2 \ \doteq & \ \bind{\tflag}{\tzero},\ \bind{\tvarx}{\tcboolean},\ \bind{\tvarg_1}{\tgrd{\tflag \not = 0}} \label{eq:gamma:negate2}\\
\intertext{At the assignments to $\tvara$, $\tvarb$ and $\taz$ the environments
are respectively:}
\Gamma_a\   \doteq & \ \bind{\tnegate}{\tnegtype} \label{eq:gamma:a} \\
\Gamma_b\   \doteq & \ \Gamma_a ,\ \bind{\tvara}{\tcnumber} \label{eq:gamma:b} \\
\Gamma_{c}\ \doteq & \ \Gamma_b ,\ \bind{\tvarb}{\tcboolean} \label{eq:gamma:a0} \\
\intertext{where $\tnegtype$ abbreviates the \emph{product} type of the (elaborated) tuple $\tnegate$.}
\tnegtype\ \doteq & \ \tprod{(\funtype{\tpos, \tcnumber}{\tcnumber})}{(\funtype{\tzero, \tcboolean}{\tcboolean})}
\end{align}

\subparagraph*{Subtyping.}
At each function call-site, the refinement type system 
checks that the \emph{actual} argument is indeed a subtype 
of the \emph{expected} one. For example, the @DEAD@ calls   
inside $\tnegateone$ and $\tnegatetwo$ yield the respective 
subtyping obligation:
\begin{alignat}{2}
\Gamma_1 & \vdash\ \reftp{\tcnumber}{\vv = \tvarx}  &&\ \lqsubt\ \reftp{\tcnumber}{\rfalse} \label{eq:sub:negate1} \\
\Gamma_2 & \vdash\ \reftp{\tcboolean}{\vv = \tvarx}  &&\ \lqsubt\ \reftp{\tcboolean}{\rfalse} \label{eq:sub:negate2} \\
\intertext{The obligation states that the type of the argument $\tvarx$ should be a subtype of the input type of \tdead.
Similarly, at the assignments to $\tvara$, $\tvarb$ and $\taz$ the first
arguments generate the respective subtyping obligations:}
\Gamma_a & \vdash\ \reftp{\tcnumber}{\vv = 1}    &&\ \lqsubt\ \reftp{\tcnumber}{\vv \not = 0}  \label{eq:sub:a} \\
\Gamma_b & \vdash\ \reftp{\tcnumber}{\vv = 0}    &&\ \lqsubt\ \reftp{\tcnumber}{\vv = 0} \label{eq:sub:b} \\
\Gamma_{c} & \vdash\ \reftp{\tcnumber}{\vv = 0}  &&\ \lqsubt\ \reftp{\tcnumber}{\vv \not = 0} \label{eq:sub:c} 
\end{alignat}

\subparagraph*{Verification Conditions.}
To verify subtyping obligations, we convert them into logical verification conditions (VCs), 
whose validity determines whether the subtyping holds. 
A subtyping obligation
${\lqsubtype{\Gamma}{\reftp{b}{p}}{\reftp{b}{q}}}$ 
translates to the VC
${\dbrkts{\Gamma} \Rightarrow (p \Rightarrow q)}$
where $\dbrkts{\Gamma}$ is the conjunction of the refinements of the binders in $\Gamma$.
For example, the subtyping obligations~\eqref{eq:sub:negate1} and
\eqref{eq:sub:negate2} yield the respective VCs:
\begin{alignat}{2}
(\tflag \not = 0 \wedge \rtrue \wedge \tflag = 0) & \Rightarrow\ {\vv = \tvarx}  &&\ \Rightarrow\ {\rfalse} \label{eq:vc:negate1} \\
(\tflag = 0 \wedge \rtrue \wedge \tflag \not = 0) & \Rightarrow\ {\vv = \tvarx}  &&\ \Rightarrow\ {\rfalse} \label{eq:vc:negate2} \\
\intertext{Here, the conjunct $\rtrue$ arises from the trivial refinements 
\eg the binding for $\tvarx$. The above VCs are deemed valid by an SMT solver 
as the hypotheses are inconsistent, which proves the call 
is indeed dead code. Similarly, \eqref{eq:sub:a}, \eqref{eq:sub:b} respectively
yield VCs:}
\rtrue & \Rightarrow \ {\vv = 1}    &&\ \Rightarrow\ {\vv \not = 0} \label{eq:vc:a} \\
\rtrue & \Rightarrow \ {\vv = 0}    &&\ \Rightarrow\ {\vv = 0} \label{eq:vc:b} \\
\intertext{which are deemed valid by SMT, verifying the assignments to $\tvara$, $\tvarb$. 
However, by \eqref{eq:sub:c}:} 
\rtrue & \Rightarrow \ {\vv = 0}  &&\ \Rightarrow \ {\vv \not = 0} \label{eq:vc:c}
\end{alignat}
which is invalid, ensuring that we \emph{reject} the call that assigns to $\taz$.
 
\subsection{Two-Phase Inference}\label{sec:overview:infer}

Our two-phased approach readily lends itself 
to abstract interpretation based \emph{refinement inference} 
which can drastically lower the programmer annotations required 
to verify various safety properties, 
\eg reducing the annotations
needed to verify array bounds safety in ML programs from 
31\% of code size to under 1\%~\cite{LiquidPLDI08}.
Here we illustrate how inference works in the presence 
of value-based overloading. Suppose we are \emph{not} given the
refinements for the signature of $\tnegate$ but only the unrefined 
signature (either given to us explicitly as in \tsc, or inferred via dataflow analysis
\cite{lambdajs,flow}). 
As inference is difficult with incorrect code, 
we omit the erroneous statements that assign to @c@ and @d@.

Refinement inference proceeds in three steps.
First, we create \emph{templates} which are the basic types decorated 
with \emph{refinement variables} $\kvar$ in place of the unknown refinements.
Second, we perform the \emph{trust} phase to elaborate the source 
program into a well-typed target free of overloading. Remember that this phase uses only the basic types
and is oblivious to the (in this case unknown) refinements.
Third, we perform the \emph{verify} phase which now generates VCs over 
the refinement variables $\kvar$. These VCs -- \emph{logical implications} 
between the refinements and $\kvar$ variables -- correspond to so-called 
Horn constraints over the $\kvar$ variables, and can be
solved via abstract interpretation~\cite{Houdini,LiquidPLDI08}.

\subparagraph*{0. Templates:} 
Let us revisit the program from 
Figure~\ref{fig:negate:ts}, with the goal of inferring 
the refinements. Recall that the (unrefined) type 
of $\tnegate$ is: 
\begin{align*}
\tnegate ::     &\ {\funtype{\tcnumber,\tcnumber}{\tcnumber}} \\
         \wedge &\ {\funtype{\tcnumber,\tcboolean}{\tcboolean}} \\
\intertext{We create a \emph{template} by refining each base 
type with a (distinct) refinement variable:}
\tnegate ::     &\ {\funtype{\reftp{\tcnumber}{\kvar_1},\reftp{\tcnumber}{\kvar_2}}{\reftp{\tcnumber}{\kvar_3}}} \\
         \wedge &\ {\funtype{\reftp{\tcnumber}{\kvar_4},\reftp{\tcboolean}{\kvar_5}}{\reftp{\tcboolean}{\kvar_6}}}
\end{align*}

\subparagraph*{1. Trust:} 
The {trust} phase proceeds as before, propagating the refinements 
to the signatures of the elaborated target, yielding the code on the right in 
Figure~\ref{fig:negate:target} {except} that $\tnegateone$ and $\tnegatetwo$ 
have the respective templates:
\begin{align*}
\tnegateone ::  &\ {\funtype{\reftp{\tcnumber}{\kvar_1},\reftp{\tcnumber}{\kvar_2}}{\reftp{\tcnumber}{\kvar_3}}} \\
\tnegatetwo ::  &\ {\funtype{\reftp{\tcnumber}{\kvar_4},\reftp{\tcboolean}{\kvar_5}}{\reftp{\tcboolean}{\kvar_6}}}
\end{align*}

\subparagraph*{2. Verify:}
The verify phase proceeds as before, but using templates instead of the types. 
Hence, at 
the $\tdead$-cast in $\tnegateone$ and $\tnegatetwo$, and
the calls to @neg@ that assign to @a@ and @b@,
instead of the VCs \eqref{eq:vc:negate1}, \eqref{eq:vc:negate2}, \eqref{eq:vc:a}
and \eqref{eq:vc:b},
we get the respective Horn constraints:
\begin{alignat}{2}
(\SUBST{\kvar_1}{\vv}{\tflag} \wedge \rtrue \wedge \tflag = 0) & \Rightarrow\ {\vv = \tvarx}  &&\ \Rightarrow\ {\rfalse} \label{eq:horn:negate1} \\
(\SUBST{\kvar_4}{\vv}{\tflag} \wedge \rtrue \wedge \tflag \not = 0) & \Rightarrow\ {\vv = \tvarx}  &&\ \Rightarrow\ {\rfalse} \label{eq:horn:negate2} \\
\rtrue & \Rightarrow \ {\vv = 1}    &&\ \Rightarrow\ \kvar_1  \label{eq:horn:a} \\
\rtrue & \Rightarrow \ {\vv = 0}    &&\ \Rightarrow\ \kvar_4  \label{eq:horn:b}
\end{alignat}
These constraints are identical to the corresponding VCs except that $\kvar$ 
variables appear in place of the unknown refinements for the corresponding binders. 
We can solve these constraints using fixpoint computations over a variety of 
abstract domains such as monomial predicate abstraction 
\cite{Houdini,LiquidPLDI08} over a set of ground predicates 
which are arithmetic (in)equalities between program variables 
and constants, to obtain a solution mapping each $\kvar$ 
to a concrete refinement:
$$\kvar_1 \ \doteq \ \vv = 0 \qquad \qquad \kvar_4 \ \doteq \ \vv \not = 0 \qquad \qquad \kvar_2, \kvar_3, \kvar_5, \kvar_6 \ \doteq \ \rtrue $$
which, when plugged back into the templates, allow us to infer types for $\tnegate$.

\subparagraph*{Higher-Order Verification.}  
Our two-phased approach generalizes directly to offer precise analysis 
for \emph{polymorphic, higher-order} functions. 
Returning to the code in Figure~\ref{fig:reduce}, our two-phased inference algorithm infers the refinement types:
\begin{align*}
  \text{\tdreduce}  ::      & \ \funsig{A,B}{\bind{a}{\tarray{A}}, \bind{f}{\funtype{B, A, \tidx{a}}{B}}, \bind{x}{B}}{B} \\
  \text{\treduce}   ::      & \ \funsig{A}{\bind{a}{\tnarray{A}}, \bind{f}{\funtype{A, A, \tidx{a}}{A}}}{A}               \\
           \wedge  & \ \funsig{A,B}{\bind{a}{\tarray{A}}, \bind{f}{\funtype{B, A, \tidx{a}}{B}}, \bind{x}{B}}{B} \\
\intertext{where $\tidx{a}$ describes \emph{valid indices} for array
$a$, and $\tnarray{A}$ describes non-empty arrays:}
\tidx{a} \doteq    & \ \reftp{\tcnumber}{0 \leq \vv < \tlen{a}} \\ 
\tnarray{A} \doteq & \  \ \reftp{\tarray{A}}{0 < \tlen{\vv}} \\
\intertext{The above type is a precise \emph{summary} for the 
higher-order behavior of \tdreduce: it describes the relationship
between the input array $a$, the step (``callback'') function $f$, 
and the initial value of the accumulator, and stipulates that 
the output satisfies the same \emph{properties} $B$ as the input $x$.
Furthermore, it captures the fact that the callback $f$ is only 
invoked on inputs that are valid indices for the array $a$ that is
being reduced. 
Consequently, Liquid Types~\cite{LiquidPLDI08}, for example, would automatically infer:}
\tstep     ::     & \ \funsig{A}{\tidx{a},A,\tidx{a}}{\tidx{a}} \\
\tminindex ::     & \ \funsig{A}{\tarray{A}}{\tcnumber}
\end{align*}
thereby verifying the safety of array accesses in the presence of 
higher order functions, collections, and value-based overloading.

\section{Syntax and Operational Semantics of \lang}\label{sec:language}

\begin{figure}[t!]
\noindent\textbf{Source Language: Syntax}

\begin{tabular}{>{$}r<{$} >{$}r<{$} >{$}r<{$} >{$}l<{$}}
\textit{Values}	&	\val & ::= & 
					\vconst
             \spmid  	\evar
             \spmid  	\elambda{\evar}{\expr}{}
\\
\textit{Expressions} & \expr & ::= &
					\val
               \spmid \eletin{\evar}{\expr_1}{\expr_2}
               \spmid \ite{\expr}{\expr_1}{\expr_2}
               \spmid \app{\expr_1}{\expr_2}
\\[0.1in]             
\textit{Primitive Types}  & \tprim & ::= & \tnumber \spmid \tbool \\
\textit{Types}            & \type, \typeb  & ::= & 
					  \tprim
                 \spmid \tfun{\type}{\typeb}
                 \spmid \tand{\type}{\typeb}
                 \spmid \tor{\type}{\typeb} 
\end{tabular}

\vspace*{0.2in}

\judgementHead{Source Language: Operational Semantics}{\steps{\expr}{\expr'}}
\begin{mathpar}
  \inferrule[\opsrcevalctx]
  {\steps{\expr}{\expr'}}
  {\steps{\evalctxarg{\expr}}{\evalctxarg{\expr'}}}
\and
\inferrule[\opsrcappa]{}
  {\steps{\app{\vconst}{\val}}{\primapp{\vconst}{\val}}}
\and
\inferrule[\opsrcappb]{}
{\steps{\app{(\elambda{\evar}{\expr}{})}{\val}}
         {\appsubst{\esubst{\val}{\evar}}{\expr}}}
\\
\inferrule[\opsrccondtrue]{}
  {\steps{\ite{\vtrue}{\expr_1}{\expr_2}}{\expr_1}}
\and
\inferrule[\opsrccondfalse]{}
  {\steps{\ite{\vfalse}{\expr_1}{\expr_2}}{\expr_2}}
\and
\inferrule[\opsrclet]{}
  {\steps{\eletin{\evar}{\val}{\expr}}{\appsubst{\esubst{\val}{\evar}}{\expr}}}
\end{mathpar}
\caption{Syntax and Operational Semantics of \srclang}
\label{fig:srclang}
\label{fig:source:opsem}
\end{figure}

Next, we formalize two-phase typing via a core calculus \lang comprising 
a \emph{source} language \srclang \emph{with} overloading via union and 
intersection types, and a simply typed \emph{target} language \tgtlang 
\emph{without} overloading, where the assumptions for safe overloading 
are explicated via $\tdead$-casts.
In \S~\ref{sec:elaboration}, we describe the first phase that elaborates
source programs into target programs, and finally, in \S~\ref{sec:ref-checking}
we describe how the second phase verifies the $\tdead$-casts on the target 
to establish the safety of the source.
Our elaboration follows the overall compilation strategy of Dunfield~\cite{Dunfield2012} 
except that we have value-based overloading instead of an explicit ``merge'' 
operator \cite{forsythe}, and consequently, our elaboration and proofs must 
account for source level ``errors'' via $\tdead$-casts.

\subsection{Source Language \texorpdfstring{(\srclang)}{}}\label{subsec:source}

\subparagraph*{Terms.} 
We define a source language \srclang, with syntax 
shown in Figure~\ref{fig:srclang}. Expressions include variables, 
functions, applications, let-bindings, a ternary conditional construct,
and primitive constants $\vconst$ which include numbers $0, 1, \ldots$, 
operators $+, -, \ldots$, \etc

\subparagraph*{Operational Semantics.} 
In figure~\ref{fig:source:opsem} we also define a standard small-step 
operational semantics for \srclang with a left-to-right
order of evaluation, based on evaluation contexts
$$
\evalctx ::=     \empevalctx
\spmid \eletin{\evar}{\evalctx}{\expr}
\spmid \ite{\evalctx}{\expr_1}{\expr_2}
\spmid \app{\evalctx}{\expr}
\spmid \app{\val}{\evalctx}
$$

\subparagraph*{Types.}
Figure~\ref{fig:srclang} shows the types $\type$ in the source language.
These include primitive types \tprim, arrow types
$\tfun{\type}{\typeb}$ and, most notably, intersections $\tand{\type}{\typeb}$ 
and (untagged) unions $\tor{\type}{\typeb}$ (hence the name \srclang).  
Note that the source level types are \emph{not} refined, as crucially,
the first phase \emph{ignores} the refinements when carrying out
the elaboration.

\subparagraph*{Tags.}
As is common in dynamically typed languages, runtime values are associated
with \emph{type tags}, which can be inspected with a type test (cf.
\jsc's @typeof@ operator). We model this notion to our static
types, by associating each type with a set of possible tags. 
The multiplicity arises from unions.
The meta-function $\typetag{\type}$, defined in Figure~\ref{fig:source:wf}, 
returns the possible tags that values of type $\type$ may have at runtime. 

\subparagraph*{Well-Formedness.}  In order to resolve overloads statically,
we apply certain restrictions on the form of union and
intersection types, shown by the judgment $\tcwf{\type}$ 
formalized in Figure~\ref{fig:source:wf}.
For convenience of exposition, the parts of an untagged union need to have 
distinct runtime tags, and intersection types require all 
conjuncts to have the same tag.

\begin{figure}[t!]
\judgementHead{Well-Formed Types}{\tcwf{\type}}
\begin{mathpar}
  \inferrule{}{\tcwf{\tprim}}
\and
  \inferrule
  {
    \tcwf{\type}  \\ 
    \tcwf{\typeb} \\ 
  }
  {\tcwf{\tfun{\type}{\typeb}}}
\and
  \inferrule
  {\tcwf{\type} \\ \tcwf{\typeb} \\\\ \typetag{\type} = \typetag{\typeb}}
  {\tcwf{\tand{\type}{\typeb}}}
\and
  \inferrule
  {\tcwf{\type} \\ \tcwf{\typeb} \\\\ \typetag{\type}\cap \typetag{\typeb} = \emptyset}
  {\tcwf{\tor{\type}{\typeb}}}
\end{mathpar}
\[
\begin{array}{lllll}
  \typetag{\tnumber}  & = \{\texttt{"number"}\}   & \qquad\qquad\qquad & \typetag{\tand{\type}{\type'}} & = \typetag{\type}          \\ 
  \typetag{\tbool}    & = \{\texttt{"boolean"}\}  &        & \typetag{\tor{\type}{\type'}}  & = \typetag{\type} \cup \typetag{\type'} \\
  \typetag{\tfun{\type}{\type'}}  & = \{\texttt{"function"}\}  & & & 
\end{array}
\]
\caption{Basic Type Well-Formedness} 
\label{fig:source:wf}
\end{figure}

\subsection{Target Language \texorpdfstring{(\tgtlang)}{}}\label{subsec:target}

The target language \tgtlang eliminates (value-based) overloading
and thereby provides a basic, well-typed skeleton that can be further
refined with logical predicates.  Towards this end, unions and
intersections are replaced with classical \textit{tagged unions},
\textit{products} and \tdead-casts, that encode the requirements
for basic typing.

\subparagraph*{Terms.}
Figure~\ref{fig:tgtlang} shows the terms $\texpr$ of \tgtlang, which
extend the source language with the introduction of pairs,
projections,
injections,
a case-splitting construct 
and a special constant term $\deadcast{\type}{\typeb}{\texpr}$ which denotes 
an erroneous computation.
Intuitively, a $\deadcast{\type}{\typeb}{\texpr}$ is produced in the elaboration phase whenever the
actual type $\type$ for a term $\texpr$ is incompatible with an expected 
type $\typeb$.

\begin{figure}[t!]

\noindent\textbf{Target Language: Syntax}

\begin{tabular}{>{$}r<{$} @{\quad} >{$}r<{$} >{$}r<{$} >{$}l<{$}}
  \textit{Expressions} & \texpr, \texprb & ::= & \vconst
             \spmid  \evar                                 
             \spmid  \elambda{\evar}{\texpr}{}
             \spmid  \ite{\texpr}{\texpr_1}{\texpr_2}        
             \spmid  \app{\texpr_1}{\texpr_2} \\
  &        & \spmid & \textcolor{blue}{\pair{\texpr_1}{\texpr_2}}
             \spmid   \textcolor{blue}{\proj{\projind}{\texpr}}
             \spmid   \textcolor{blue}{\inj{1}{\texpr}}
             \spmid   \textcolor{blue}{\inj{2}{\texpr}} \\
  &        & \spmid &
             \textcolor{blue}{\casematch{\texpr}{\evar_1}{\texpr_1}{\evar_2}{\texpr_2}}       
             \spmid \textcolor{blue}{\deadcast{\type}{\typeb}{\texpr}} \\
  \textit{Values} & \tval & ::=&        \vconst
             \spmid \evar
             \spmid \elambda{\evar}{\texpr}{}
             \spmid \inj{1}{\tval} 
             \spmid \inj{2}{\tval} 
             \spmid \pair{\texpr}{\texpr}
             \spmid \deadcast{\type}{\typeb}{\tval} \\[0.1in]

\textit{Ref. Types}     &
  \rtype, \rtypeb & ::= &  \reftp{\tprim}{\pred}
  \spmid \tfun{\varbinding{\evar}{\rtype}}{\rtypeb}
  \spmid \tsum{\rtype}{\rtypeb}  
  \spmid \tprod{\rtype}{\rtypeb}
\end{tabular}

\vspace*{0.2in}

\judgementHead{Target Language: Operational Semantics}{\steps{\texpr}{\texpr'}}
\begin{mathpar}
  \inferrule[\optgtevalctx]
  {\steps{\texpr}{\texpr'}}
  {\steps{\tevalctxarg{\texpr}}{\tevalctxarg{\texpr'}}}
\and
\inferrule[\optgtappa]{
  \tval \not\equiv \deadcast{\type}{\typeb}{\tval'}}
  {\steps{\app{\vconst}{\tval}}{\primapp{\vconst}{\tval}}}
\and
\inferrule[\optgtbeta]{}
{\steps{\app{(\elambda{\evar}{\texpr}{})}{\tval}}
         {\appsubst{\esubst{\tval}{\evar}}{\texpr}}}
\and
\inferrule[\optgtcondtrue]{}
  {\steps{\ite{\vtrue}{\texpr_1}{\texpr_2}}{\texpr_1}}
\and
\inferrule[\optgtcondfalse]{}
  {\steps{\ite{\vfalse}{\texpr_1}{\texpr_2}}{\texpr_2}}
\and
\inferrule[\optgtlet]{}
  {\steps{\eletin{\evar}{\tval}{\texpr}}{\appsubst{\esubst{\tval}{\evar}}{\texpr}}}
\and
\inferrule[\optgtproj]{}
  {\steps{\proj{k}{\pair{\texpr_1}{\texpr_2}}}
         {\texpr_k}}
\and
\inferrule[\optgtcase]{}
  {\steps{\casematch{\inj{k}{\tval}}{\evar_1}{\texpr_1}{\evar_2}{\texpr_2}}
         {\appsubst{\esubst{\tval}{\evar_k}}{\texpr_k}}}
\end{mathpar}
\caption{Syntax and Operational Semantics of \tgtlang}
\label{fig:tgtlang}
\label{fig:target:opsem}
\end{figure}

\subparagraph*{Operational Semantics.}
As in the source language we define evaluation contexts 
\[
\begin{array}{rl}
  \tevalctx ::= & \empevalctx \spmid \eletin{\evar}{\tevalctx}{\texpr} \spmid \ite{\tevalctx}{\texpr_1}{\texpr_2} \spmid \app{\tevalctx}{\texpr} 
      \spmid      \app{\val}{\tevalctx} \spmid \inj{k}{\tevalctx} \\
      \spmid    & \proj{\projind}{\tevalctx} \spmid 
                  \deadcast{\type}{\typeb}{\tevalctx} \spmid \casematch{\tevalctx}{\evar_1}{\texpr_1}{\evar_2}{\texpr_2} 
\end{array}
\]
and use
them to define a small-step operational semantics for the target in Figure~\ref{fig:target:opsem}.
Note how evaluation is allowed in \tdead-casts and 
$\deadcast{\type}{\typeb}{\tval}$ \emph{is} a value.

\subparagraph*{Types.}
The target language is checked against a refinement type checker.
Thus, we modify the type language to account for the new language terms
and refinements.
\emph{Basic Refinement Types} are of the form $\reftp{\tprim}{\pred}$, consisting
of the same basic types $\tprim$ as source types, and a logical 
predicate $\pred$ (over some decidable logic), which describes the properties
that values of the type must satisfy.
Here, $\vv$ is a special \textit{value variable} that describes the inhabitants
of the type, that does not appear in the program, but can appear inside
the refinement $\pred$.
Function types are of the form $\tfun{\varbinding{\evar}{\rtype}}{\rtypeb}$, 
to express the fact that the refinement predicate of the return type $\rtypeb$ 
may refer to the value of the argument $\evar$.
Sum and product types have the usual structure found in ML-like languages.

\section{Phase 1:  Trust}\label{sec:elaboration}

Terms of \srclang are elaborated to terms of \tgtlang by a 
judgment:
${\trans{\tcenv}{\expr}{\type}{}{\texpr}}$.
This is read: under the typing assumptions in $\tcenv$, term $\expr$ of the
source language is assigned a type $\type$ and elaborates to a term $\texpr$ of
the target language. This judgment follows closely Dunfield's elaboration 
judgment~\cite{Dunfield2012}, but with crucial differences that arise due to 
dynamic, value-based overloading, which we outline below.

\subparagraph*{Elaboration Ignores Refinements.}
A key aspect of the first phase is that elaboration is based solely 
on the basic types, \ie does \emph{not} take type refinements into account.
Hence, the types assigned to source terms are  
transparent with respect to refinements; or more precisely, they work just 
as placeholders for refinements that can be provided as user specifications.
These specifications
are propagated \emph{as is} during the first phase 
along with the respective basic types they are attached to. 
Due to this transparency of refinements we have decided to 
omit them entirely from our description of the elaboration phase.

\subsection{Source Language Type-checking and Elaboration}\label{subsec:elab}

Figure~\ref{fig:elaboration} shows the rules that formalize the elaboration process. 
At a high-level, following Dunfield~\cite{Dunfield2012}, unions and intersections are
translated to simpler typing constructs like sums and products (and 
the attendant injections, pattern-matches, and projections).
Unlike the above work, which focuses on the classical intersection setting 
where overloading is explicit via a ``merge'' construct \cite{forsythe}, we
are concerned with the dynamic setting where overloading is 
value-based, leading to conventional type ``errors''. 

\subparagraph*{Elaboration Modes: Strict and Flexible.} 
Thus, one of the distinguishing features of our type system is its
ability to not fail in cases where conventional static type system
would raise type incompatibility errors, but instead elaborate the
offending terms to the special error form $\deadcast{\type}{\typeb}{\texpr}$.
However, these error forms do not appear indiscriminately, but under 
certain conditions, specified by two elaboration modes:
(1)~a \textit{flexible} judgment ($\vdash_{\strictplace{\flexiblesymbol}}$) for rules
  that may yield $\deadcast{\type}{\typeb}{\texpr}$ terms, and 
(2)~a \textit{strict} judgment ($\vdash_{\strictplace{\strictsymbol}}$) for those that don't. 
Most elaboration rules come in both flavors, depending on the surrounding 
rules in a typing derivation. We write $\alpha$ to parameterize over the 
two modes.

Intuitively, we use flexible mode when checking calls to non-overloaded
functions (with a \emph{single} conjunct) and strict mode when checking
calls to overloaded ones.
In the former case, a type incompatibility truly signals a (potential)
run-time error, but in the latter case, incompatibility may indicate
the wrong choice of overload.
Consequently, the elaboration judgment also states whether the intersection 
rule has been used, or not, by annotating the hook-arrow with the label 
$\elabinter$ or $\elabnointer$, respectively.
As with strictness, we parametrize over $\elabnointer$ and $\elabinter$ 
with the variable $\elabintervar$, and use $\elabinterany$ to denote that 
the outcome is not important.

\begin{figure}[t]
\judgementHead{Elaboration Typing}{\trans{\tcenv}{\expr}{\type}{}{\texpr}}
\begin{mathpar}
  \inferrule*
  [left=\tchktoplevel]
  {\translax{\emp}{\expr}{\type}{\elabnointer}{\texpr}}
  {\trans{\emp}{\expr}{\type}{}{\texpr}}
\and
  \inferrule*
  [left=\tchkweaken]
  {\transstrict{\tcenv}{\expr}{\type}{\elabintervar}{\texpr}}
  {\translax{\tcenv}{\expr}{\type}{\elabintervar}{\texpr}}
\end{mathpar}
\begin{center}
\line(1,0){320}
\end{center}
\begin{mathpar}
  \inferrule
  [\tchkconst]
  {}
  {\transvar{\tcenv}{\vconst}{\tconst}{\elabintervar}{\vconst}}
\and
  \inferrule*
  [left=\tchkletin]
  {
    \transvar{\deftcstate}{\expr_1}{\type_1}{\elabinterany}{\texpr_1} \\
    \transvar{\tcenvext{\evar}{\type_1}}{\expr_2}{\type_2}{\elabintervar}{\texpr_2}
  }
  {\transvar{\tcenv}{\eletin{\evar}{\expr_1}{\expr_2}}{\type_2}{\elabintervar}{\eletin{\evar}{\texpr_1}{\texpr_2}}}
\and
  \inferrule*
  [left=\tchkvar]
  {\envbinding{\evar}{\type}\in \tcenv}
  {\transvar{\tcenv}{\evar}{\type}{\elabintervar}{\evar}}
\and
  \inferrule*
  [left=\tchkite]
  {
    \translax{\tcenv}{\expr}{\tbool}{\elabnointer}{\texpr} \\ 
    \foralli {i \in \{1,2\}}
    {\transvar{\tcenv}{\expr_i}{\type}{\elabintervar}{\texpr_i}}
  }
  {\transvar{\tcenv}{\ite{\expr}{\expr_1}{\expr_2}}{\type}{\elabintervar}{\ite{\texpr}{\texpr_1}{\texpr_2}}}
\\
  \inferrule*
  [left=\tchkinterintro]
  {
    \foralli {\kindex \in \{1,2\}}
        {\transvar{\tcenv}{\val}{\type_{\kindex}}{\elabintervar}{\texpr_{\kindex}}}
    \\
    \tcwf{\tand{\type_1}{\type_2}}
  }
  {
    \transvar{\tcenv}
          {\val}
          {\tand{\type_1}{\type_2}}
          {\elabintervar}{\pair{\texpr_1}{\texpr_2}}
  }
\and
\inferrule*
  [left=\tchkinterelim]
  {\transvar{\tcenv}{\expr}{\tand{\type_1}{\type_2}}{\elabinterany}{\texpr}}
  {\transvar{\tcenv}{\expr}{\type_{\kindex}}{\elabinter}{\proj{\kindex}{\texpr}}}
\and
  \inferrule*
  [left=\tchklambda]
  {
    \tcwf{\tfun{\type}{\typeb}} \\
    \transvar{\tcenvext{\evar}{\type}}{\expr}{\typeb}{\elabinterany}{\texpr}
  }
  {\transvar{\tcenv}
    {\elambda{\evar}{\expr}{}}
    {\tfun{\type}{\typeb}}{\elabnointer}
    {\elambda{\evar}{\texpr}{}}
  }
\and
  \inferrule*
  [left=\tchkapp]
  {
    \transvar{\tcenv}{\expr_1}{\tfun{\type}{\typeb}}{\elabinter/\elabnointer}{\texpr_1}
    \\\\
    \transpar{\tcenv}{{\strictsymbol/\flexiblesymbol}}{\expr_2}{\type}{\elabinterany}{\texpr_2}
  }
  {\transvar{\tcenv}
    {\app{\expr_1}{\expr_2}}{\typeb}{\elabnointer}{\app{\texpr_1}{\texpr_2}}
  }
\and
  \inferrule*
  [left=\tchkdead]
  {\translax{\tcenv}{\expr}{\type}{\elabintervar}{\texpr} \\
   \typetag{\type}\cap \typetag{\typeb} = \emptyset}
   {\translax{\tcenv}{\expr}{\typeb}{\elabintervar}{\deadcast{\type}{\typeb}{\texpr}}}
\and
  \inferrule*
  [left=\tchkup]
  {
    \translax{\tcenv}{\expr}{\type_k}{\elabintervar}{\texpr} \\
    \tcwf{\tor{\type_1}{\type_2}}
  } 
  {\translax{\tcenv}{\expr}{\tor{\type_1}{\type_2}}{\elabintervar}{\inj{k}{\texpr}}}
\and
  \inferrule*
  [left=\tchkdown]
  {
    \phantom{\transvar{\tcenv}{\expr_0}{\tor{\type_1}{\type_2}}{\elabintervar}{\texpr_0}}
    \\
    \transvar{\tcenvext{\evar_1}{\type_1}}
             {\evalctxarg{\evar_1}}
             {\typeb}
             {\elabintervar}
             {\texpr_1}
    \\\\
    \transvar{\tcenv}{\expr_0}{\tor{\type_1}{\type_2}}{\elabintervar}{\texpr_0}
    \\
    \transvar{\tcenvext{\evar_2}{\type_2}}
             {\evalctxarg{\evar_2}}
             {\typeb}
             {\elabintervar}
             {\texpr_2}
  }
  {
    \transvar{\tcenv}
             {\evalctxarg{\expr_0}}
             {\typeb}
             {\elabintervar}{\casematch{\texpr_0}{\evar_1}{\texpr_1}{\evar_2}{\texpr_2}}
  } 
\end{mathpar}
\caption{Elaboration Typing rules }
\label{fig:elaboration}
\end{figure}

\subparagraph*{Top-level Elaboration.}
Our top-level judgment is agnostic of either of the aforementioned modes.
Elaborating programs in an empty context ($\vdash$) is essentially elaborating
in the flexible sense and assumes we are not in the context of intersection
elimination (\tchktoplevel).
Furthermore, an elaboration that succeeds in strict mode also
succeeds in flexible mode (\tchkweaken), so all strict rules can be used as
flexible ones.

\subparagraph*{Standard Rules.}
Rules \tchkconst, \tchkvar are standard and preserve the structure of the source 
program.
Rule \tchkite expects the condition $\expr$ of a conditional expression
$\ite{\expr}{\expr_1}{\expr_2}$ to be of boolean type, and assigns the same type
$\type$ to each branch of the conditional.
Rule \tchkletin checks expressions of the form
$\eletin{\evar}{\expr_1}{\expr_2}$. It assigns a type $\type_1$ to expression
$\expr_1$ and checks $\expr_2$ in an environment extended with the binding of
$\type_1$ for $\evar$. 

\subparagraph*{Intersections.}
In rule \tchkinterintro the choice of the type we assign to a value $\val$
causes different elaborated terms $\tval_{\kindex}$, as different typing
requirements cause the addition of $\tdead$-casts at different places.
This rule is intended to be used primarily for abstractions, so it's limited to
accept values as input. Rule \tchkinterelim for eliminating
intersections replaces a term $\expr$ that is originally typed as an
intersection with a projection of that part of the pair that has a matching type.
By \tchkinterintro values typed at an intersection get a pair
form. 

\subparagraph*{Unions.}
Rule \tchkup for union introduction is standard.
The union elimination rule, taken from Dunfield's elaboration
scheme~\cite{Dunfield2012}, states that an expression $\expr_0$ can
be assigned a union type $\tor{\type_1}{\type_2}$ when placed at the
``hole'' of an evaluation context $\evalctx$, so long as the evaluation
context can be typed with the same type $\typeb$, when the hole is
replaced with a variable typed as $\type_1$ on the one hand and as
$\type_2$ on the other. 
While the rule is inherently non-deterministic, it suffices for a
declarative description of the elaboration process; see Dunfield's
subsequent work on untangling type-checking of intersections and
unions~\cite{Dunfield11letnormal} for an algorithmic variant via a
let-normal conversion.

\subparagraph*{Abstraction and Application.}
Rule \tchklambda assumes the arrow type $\tfun{\type}{\typeb}$ 
is given as annotation and is required to conform to the 
well-formedness constraints.
At the crux of our type system is the rule \tchkapp. 
Expression $\expr_1$ can be typed in flexible mode. 
Depending on whether intersection elimination was used for $\expr_1$ we 
toggle on the mode of checking $\expr_2$. To only allow sensible derivations, 
we disallow the use of the $\tdead$-cast insertion when choosing among the 
cases of an intersection type. 
Below, we justify this choice using an example.
If on the other hand, the type for $\expr_1$ is assigned without choosing 
among the parts of an intersection, then expression $\expr_2$ can be typed 
in flexible mode, potentially producing $\tdead$-casts.

\subparagraph*{Trusting via $\tdead$-Casts.}
The cornerstone of the ``trust'' phase lies in the presence of the \tchkdead
rule. As we mentioned earlier, this rule can only be used in flexible mode. 
The main idea here is to allow cases that are obviously wrong, as far as the
simple first phase type system is concerned; but, at the same time, include
a $\tdead$-cast annotation and defer sound type-checking for the second phase. 
The premises of this rule specify that a $\tdead$-cast annotation will only 
be used if the inferred and the expected type have different tags. 
One of the consequences of this decision is that it does not allow
$\tdead$-casts induced by a mismatch between higher-order types, as the
tags for both types would be the same (most likely \texttt{"function"}).
Thus, such mismatches are ill-typed and rejected in the first phase.
This limitation is due to the limited information that can be encoded
using the tag mechanism. A more expressive tag mechanism could eliminate 
this restriction but we omit this for simplicity of exposition.

\subparagraph*{Semantics of $\tdead$-Casts.}
To prove that elaboration preserves source level 
behaviors, our design of $\tdead$-casts preserves 
the property that the target gets stuck \emph{iff} 
the source gets stuck.
That is, source level type ``errors'' \emph{do not lead to early} 
failures (\eg at function call boundaries).
Instead, $\tdead$-casts correspond to \emph{markers} 
for all source terms that can potentially cause 
execution to get stuck. 
Hence, the target execution itself gets stuck at the 
same places as the source -- \ie when applying 
to a non-function, branching on a non-boolean 
or primitive application over the wrong base 
value, except that in the target, the stuckness 
can only occur when the value in question carries 
a $\tdead$ marker.
Consider the source program 
${(\elambda{\evar}{\evar\ 1}{})\ 0}$
which gets stuck \emph{after} the top-level application, 
when applying $1$ to $0$.
It could be elaborated to 
${(\elambda{\evar}{\evar\ 1}{})\ \deadcast{\type}{\typeb}{0}}$
(where $\type$ and $\typeb$ are respectively $\tnumber$ and $\tfun{\tnumber}{\tnumber}$)
which also has a top-level application and gets stuck at the second, inner application.

\subparagraph*{Necessity of elaboration modes.}
If we allowed 
the argument of an \emph{overloaded} call-site to be checked in \textit{flexible} context, then for the application
$\app{f}{\evar}$, where $f$ has been assigned the type 
$\bind{f}{\tand{\tfun{\tnm}{\tnm}}{\tfun{\tbl}{\tbl}}}$ and
$\bind{\evar}{\tbl}$, the following derivation would be possible:
\begin{mathpar}
  \inferrule*[left=\tchkapp]
  {
    \inferrule*[left=\tchkinterelim]
    {\vdots}
    {\translax{\dots}
             {f}
             {\tfun{\tnm}{\tnm}}
             {\elabinter}
             {\proj{1}{f}}
    }
    \and
    \inferrule*[left=\tchkdead]
    {
      \translax{\dots}{\evar}{\tbl}{\elabnointer}{\evar}
      \\\\
      \typetag{\tbl}\cap \typetag{\tnm} = \emptyset
    }
    {
      \translax{\dots}
             {x}
             {\tnm}
             {\elabnointer}
             {\deadcast{\tbl}{\tnm}{\evar}}
    }
  }
  {
    \translax{\envbinding{f}{\tand{\tfun{\tnm}{\tnm}}{\tfun{\tbl}{\tbl}}}, 
              \envbinding{\evar}{\tbl}}
             {\app{f}{\evar}}
             {\tnm}
             {\elabnointer}
             {\app{\parens{\proj{1}{f}}}{\parens{\deadcast{\tbl}{\tnm}{\evar}}}}
  }
\end{mathpar}
But, clearly, the intended derivation here is:
\begin{mathpar}
  \inferrule*[left=\tchkapp]
  {
    \inferrule*[left=\tchkinterelim]
    {\vdots}
    {\translax{\dots}{f}{\tfun{\tbl}{\tbl}}{\elabinter}{\proj{2}{f}}}
    \and
    \inferrule*[]{}
    {\transstrict{\dots}{\evar}{\tbl}{\elabnointer}{\evar}}
  }
  {
    \translax{\envbinding{f}{\tand{\tfun{\tnm}{\tnm}}{\tfun{\tbl}{\tbl}}}, 
              \envbinding{\evar}{\tbl}}
             {\app{f}{\evar}}
             {\tbl}
             {\elabnointer}
             {\app{\parens{\proj{2}{f}}}{\evar}}
  }
\end{mathpar}

\subparagraph*{Subtyping.}
This formulation has been kept simple with respect to subtyping. The only notion
of subtyping appears in the \tchkup rule, where a type $\type_1$ is widened to 
$\tor{\type_1}{\type_2}$. We could have employed a more elaborate
notion of subtyping, by introducing a subtyping relation ($\issubtype{}{}$) and
a subsumption rule for our typing elaboration. The rules for this subtyping
relation would include, among others, function subtyping:
$$
\inferrule[]
{\issubtype{\type_1'}{\type_1} \\ {\issubtype{\type_2}{\type_2'}}}
{\issubtype{\tfun{\type_1}{\type_2}}{\tfun{\type_1'}{\type_2'}}}
$$
However, supporting subtyping in higher-order constructs would only be possible
with the introduction of wrappers around functions to accommodate checks on the
arguments and results of functions. So, assuming that a cast $\ccast$ represents
a dynamic check the above rule would correspond to a cast producing relation
($\convertsign$):
$$
\inferrule[]
{\convert{\type_1'}{\type_1}{\ccast_1} \\ {\convert{\type_2}{\type_2'}{\ccast_2}}}
{\convert{\tfun{\type_1}{\type_2}}
         {\tfun{\type_1'}{\type_2'}}
         {\elambda{f}
            {\elambda{x}
        {\parens{\cexpr{\ccast_2}{(\app{f}{(\cexpr{\ccast_1}{x})})}}{}}{}}}}
$$
This formulation would just complicate the translation without giving any more
insight in the main idea of our technique, and hence we forgo it.

\subsection{Source and Target Language Consistency}\label{subsec:consistency}

In this section, we present the theorems that precisely connect the semantics 
of source programs with their elaborated targets. 
The main challenges towards establishing those are that:
(1)~the source and target do not proceed in lock-step, a single step of the one
    may be matched by several steps of the other (for example evaluating a
    projection in the target language does not correspond to any step in the
    source language), and
(2)~we must design the semantics of the $\tdead$-casts in the target to 
    ensure that $\tdead$-casts cause evaluation to get stuck iff some 
    primitive operation in the source gets stuck. 
We address these, next, with a number of lemmas and state our assumptions.

\subparagraph*{Value Monotonicity.}
This lemma fills in the mismatch that emerges when (non-value)
expressions in the source language elaborate to values in the target
language. Informally, if a source expression $\expr$ elaborates to a
target value $\tval$, then $\expr$ evaluates (after potentially
multiple steps) to a value $\val$ that is related to the target value
$\tval$ with an elaboration relation under the same type.
Furthermore, all expressions on the path to the target value $\val$
elaborate to the same value and get assigned the same type.

\begin{lemma}[Value Monotonicity]\label{lemma:valmonot}
If $\trans{\tcenv}{\expr}{\type}{}{\tval}$, 
then there exists $\val$ \st:
\begin{enumerate}[label={(\arabic*)}]
\inditem $\stepsmany{\expr}{\val}$
\inditem $\trans{\tcenv}{\val}{\type}{}{\tval}$
\inditem $\foralli{ i \;\suchthat\; \stepsmany{\expr}{\expr_i}}
                  {\trans{\tcenv}{\expr_i}{\type}{}{\tval}}$
\end{enumerate}
\end{lemma}

\begin{proof}
  The first two parts are handled similarly to
  Dunfield's~\cite{Dunfield2012} Lemma 11. 
  The last part is proved by induction on the length of the path
  $\stepsmany{\expr}{\expr_i}$.
  Details of this proof can be found in the extended version of this
  paper~\cite{rsc:appendix}.
\end{proof}

The reverse of the above lemma also comes in handy. Namely, given a value $\val$
that elaborates to an expression $\texpr$ and gets assigned the type $\type$,
there exists a value in the target language $\tval$, such that $\val$ elaborates
to $\tval$ and get assigned the \emph{same} type $\type$. 

\begin{lemma}[Reverse Value Monotonicity]\label{lemma:rev:valmonot}
If $\trans{\tcenv}{\val}{\type}{}{\texpr}$, 
then exists $\tval$ \st: 
$\stepsmany{\texpr}{\tval}$ and $\trans{\tcenv}{\val}{\type}{}{\tval}$.
\end{lemma}

\begin{proof}
Similar to proof of Lemma~\ref{lemma:valmonot}.
\end{proof}

This is an interesting
result as it establishes that different derivations may assign the same type to
a term and still elaborate it to different target terms. 
For example, one can assume derivations that consecutively apply the intersection
introduction and elimination rules. It's easy to see that the same value $\val$
can be used in the following elaborations:
\begin{align*}
  \emp \vdash \val ::  {\tand{\type_1}{\type_2}} & { \color{dark_purple}\myhook{}
\pair{\tval_1}{\tval_2}} \\
  \emp \vdash \val ::  {\tand{\type_1}{\type_2}} & { \color{dark_purple}\myhook{}
\underbrace{\pair{\proj{1}{\pair{\tval_1}{\tval_2}}}{\proj{2}{\pair{\tval_1}{\tval_2}}}}_{\texpr}}
\end{align*}
Lemma~\ref{lemma:rev:valmonot} guarantees
it will always be the case that $\stepsmany{\texpr}{\pair{\tval_1}{\tval_2}}$. 
It is up to the implementation of the type-checking algorithm to produce an efficient
target term. 

\subparagraph*{Primitive Semantics.} 
To connect the failure of the $\tdead$-casts 
with source programs getting stuck, we assume that the primitive
constants are well defined for all the values of their input domain
\emph{but not} for $\tdead$-cast values.  This lets us establish that
primitive operations $\vconst$ are invariant to elaboration. Hence, a
source primitive application gets stuck iff the elaborated argument is
a $\tdead$-cast. 
The forward version of this statement is the following assumption.

\begin{assumption}[Primitive constant application]\label{assum:const:app}
  If 
  (1)~$\trans{\emp}{\vconst}{\tfun{\type}{\typeb}}{}{\vconst}$,
  (2)~$\trans{\emp}{\val}{\type}{}{\tval}$, and
  (3)~$\tval \not \equiv \deadcast{\cdot}{\type}{\cdot}$,
then
  (i)~$\steps{\app{\vconst}{\val}}{\primapp{\vconst}{\val}}$,
  (ii)~$\steps{\app{\vconst}{\tval}}{\primapp{\vconst}{\tval}}$, and 
  (iii)~$\trans{\emp}{\primapp{\vconst}{\val}}{\typeb}{}{\primapp{\vconst}{\tval}}$.
\end{assumption}

\subparagraph*{Substitution lemma.}
As it typical in these cases, the proof of soundness requires a form 
of substitution lemma.

\begin{lemma}[Substitution]\label{lemma:subst}
If $\trans{\tcenvext{\evar}{\type}}{\expr}{\type'}{}{\texpr}$ and 
$\trans{\tcenv}{\val}{\type}{}{\tval}$ then 
$\trans{\tcenv}{\appsubst{\esubst{\val}{\evar}}{\expr}}{\type'}{}{\appsubst{\esubst{\tval}{\evar}}{\texpr}}$.
\end{lemma}

\begin{proof}
  Similar to Dunfield's substitution proof~\cite{Dunfield2012} (Lemma 12).
\end{proof}

The first consistency result is the analogue of Dunfield's
Consistency Theorem~\cite{Dunfield2012}
and states that the elaboration produces terms that are \emph{consistent} 
with the source in that each step of the target is matched by a corresponding
step of the source.
Hence, behaviors 
of the target \emph{under-approximate} the behaviors of the source. 

\begin{theorem}[Consistency]\label{theorem:consistency}
If $\trans{\emp}{\expr}{\type}{}{\texpr}$ and $\steps{\texpr}{}{\texpr'}$ then there
exists $\expr'$ such that $\stepsmany{\expr}{\expr'}$ and
$\trans{\emp}{\expr'}{\type}{}{\texpr'}$.
\end{theorem}

\begin{proof}
  The proof of this theorem is by induction on the derivation
  $\trans{\emp}{\expr}{\type}{}{\texpr}$, adapting the proof scheme
  given by Dunfield~\cite{Dunfield2012}, and using
  Lemma~\ref{lemma:valmonot}. 
  Details of this proof can be found in the extended version of this
  paper~\cite{rsc:appendix}.
\end{proof}

While this suffices to prove \emph{soundness} -- intuitively if the 
target does not ``go wrong'' then the source cannot ``go wrong'' 
either -- it is not wholly satisfactory as a trivial translation 
that converts every source program to an ill-typed target also 
satisfies the above requirement. 
So, unlike Dunfield~\cite{Dunfield2012}, we establish a completeness
result stating that if the source term steps, then the elaborated 
program will also eventually step to a corresponding (by elaboration) term. 
Theorem~\ref{theorem:rev:consistency} declares that behaviors of the elaborated target \emph{over-approximate} those of
the source, and hence, in conjunction with Theorem~\ref{theorem:consistency},
ensure that the source ``goes wrong'' iff the target does.

\begin{theorem}[Reverse Consistency]  \label{theorem:rev:consistency}
If 
$\trans{\emp}{\expr}{\type}{}{\texpr}$ 
and 
$\steps{\expr}{}{\expr'}$ 
then
there exists $\texpr'$ such that 
  $\trans{\emp}{\expr'}{\type}{}{\texpr'}$, 
  and $\stepsmanyone{\texpr}{\texpr'}$.
\end{theorem}

\begin{proof}
  Similar to the proof of Theorem~\ref{theorem:consistency}, using
  adapted versions of the lemmas used by Dunfield~\cite{Dunfield2012} and
  Lemma~\ref{lemma:rev:valmonot}.
  Again, details can be found in the accompanying report~\cite{rsc:appendix}.
\end{proof}

\section{Phase 2:  Verify}\label{sec:ref-checking}

At the end of the first phase, we have elaborated the source with
value based overloading into a classically well-typed target with
conventional typing features and $\tdead$-casts which are really
assertions that explicate the \emph{trust assumptions} made to 
type the source.
Thanks to Theorems~\ref{theorem:consistency} and \ref{theorem:rev:consistency} 
we know the semantics of the target are equivalent to the source. 
Thus, to verify the source, all that remains is to prove that the
target will not ``go wrong'', that is to prove that the $\tdead$-casts
are indeed never executed at run-time.

One advantage of our elaboration scheme is that at this point 
\emph{any} program analysis for ML-like languages (\ie supporting 
products, sums, and first class functions) can be applied to 
discharge the $\tdead$-cast~\cite{cousotcousot77}: as long as 
the target is safe, the consistency theorems guarantee that 
the source is safe. 
In our case, we choose to instantiate the second phase with 
\emph{refinement types} as they:
(1)~are especially well 
    suited to handle higher-order polymorphic functions, 
    like @minIndex@ from Figure~\ref{fig:reduce},
(2)~can easily express other correctness 
    requirements, \eg array bounds safety, thereby allowing
    us to establish not just type safety but richer correctness 
    properties, and,
(3)~are automatically inferred via the abstract interpretation 
    framework of Liquid Typing~\cite{LiquidPLDI08}.
Next, we recall how refinement typing works to show how $\tdead$-cast
checking can be carried out, and then present the end-to-end soundness
guarantees established by composing the two phases.

\subsection{Refinement Type-checking}

We present a brief overview of refinement typing as 
the target language falls under the scope of existing 
refinement type systems~\cite{Knowles10}, which can, 
after accounting for $\tdead$-casts, be reused \emph{as is} 
for the second phase.
Similarly, we limit the presentation to \emph{checking}; 
\emph{inference} follows directly from Liquid Type inference~\cite{LiquidPLDI08}.
Figure~\ref{fig:secondphase} summarizes the refinement system. 
The type-checking judgment is ${\lqcheck{\lqenv}{\texpr}{\rtype}}$, where
\textit{type environment} $\lqenv$ is a sequence of bindings 
of variables $\evar$ to refinement types $\rtype$ and 
\textit{guard predicates}, which encode control flow 
information gathered by conditional checks.
As is standard~\cite{Knowles10} each primitive constant 
$\vconst$ has a refined type $\tconst$, and a variable 
$\evar$ with type $\rtype$ is typed as $\singleton{\rtype}{\evar}$
which is $\reftp{\tprim}{\vv = \evar}$ 
if $\rtype$ is a basic type $\tprim$ and $\rtype$ otherwise.


\begin{figure*}[tb]
  \judgementHead{Refined Typechecking}{\lqcheck{\lqenv}{\texpr}{\rtype}}
\begin{mathpar}
  \inferrule*[left=\lqchksub]
  {
    \lqcheck{\lqenv}{\texpr}{\rtype_1} \\
    \lqsubtype{\lqenv}{\rtype_1}{\ifdef{\kvars}{\ltype_2}{\rtype_2}} \\
    \ifdef{\kvars}{\lqwf{\lqenv}{\ifdef{\kvars}{\ltype_2}{\rtype_2}}}{\relax}
  }
  {\lqcheck{\lqenv}{\texpr}{\ifdef{\kvars}{\ltype_2}{\rtype_2}}}
\and
  \inferrule[\lqchkconst]
  {}
  {\lqcheck{\lqenv}{\vconst}{\tconst}}
\\
  \inferrule*[left=\lqchkvar]
    {\envbinding{\evar}{\rtype} \in \lqenv }
    {\lqcheck{\lqenv}{\evar}{\singleton{\rtype}{\evar}}}
\and
  \inferrule*[left=\lqchkletin]
  {
    \lqcheck{\lqenv}{\texpr_1}{\rtype_1} \\
    \lqcheck{\lqstateext{\evar}{\rtype_1}}{\texpr_2}{\ifdef{\kvars}{\ltype_2}{\rtype_2}}   \\
    \ifdef{\kvars}{\lqwf{\lqenv}{\ifdef{\kvars}{\ltype_2}{\rtype_2}}}{\relax}
  }
  {\lqcheck{\lqenv}{\eletin{\evar}{\texpr_1}{\texpr_2}}{\ifdef{\kvars}{\ltype_2}{\rtype_2}}}
\and
  \inferrule
  [\lqchkite]
  {
    \lqcheck{\lqenv}{\texpr}{\tbool} \\
    \lqcheck{\grdext{\lqenv}{\texpr}}{\texpr_1}{\ifdef{\kvars}{\ltype}{\rtype}} \\
    \lqcheck{\grdext{\lqenv}{\neg\texpr}}{\texpr_2}{\ifdef{\kvars}{\ltype}{\rtype}} \\
    \ifdef{\kvars}{\lqwf{\lqenv}{\ifdef{\kvars}{\ltype}{\rtype}}}{\relax}
  }
  {\lqcheck{\lqenv}{\ite{\texpr}{\texpr_1}{\texpr_2}}{\ifdef{\kvars}{\ltype}{\rtype}}}
\and
  \inferrule
  [\lqchklambda]
  { 
    \lqcheck{\lqstate{\lqenvext{\evar}{\ifdef{\kvars}{\ltype_x}{\rtype_x}}}{\lqgrd}}{\texpr}{\ifdef{\kvars}{\ltype}{\rtype}} \\
    \ifdef{\kvars}{\lqwf{\lqenv}{\tfun{\ifdef{\kvars}{\ltype_x}{\rtype_x}}{\ifdef{\kvars}{\ltype}{\rtype}}}}{\relax}
  }
  {\lqcheck{\lqenv}{\elambda{\evar}{\texpr}{}}{\tfun{\ifdef{\kvars}{\ltype_x}{\rtype_x}}{\ifdef{\kvars}{\ltype}{\rtype}}}}
\and
  \inferrule
  [\lqchkapp]
  {
    \lqcheck{\lqenv}{\texpr_1}{\tfun{\rtype_x}{\rtype}}  \\
    \lqcheck{\lqenv}{\texpr_2}{\rtype_x}
  }
  {\lqcheck{\lqenv}
         {\app{\texpr_1}{\texpr_2}}
         {\appsubst{\rsubst{\texpr_2}{\evar}}{\rtype}}
  }
\and
  \inferrule
  [\lqchkpair]
  {\foralli{\kindex \in \{1,2\}}{\lqcheck{\lqenv}{\texpr_{\kindex}}{\rtype_{\kindex}}}}
  {\lqcheck{\lqenv}
         {\pair{\texpr_1}{\texpr_2}}
         {\tprod{\rtype_1}{\rtype_2}}
  }
\and
  \inferrule
  [\lqchkproj]
  {\lqcheck{\lqenv}{\texpr}{\tprod{\rtype_1}{\rtype_2}}}
  {\lqcheck{\lqenv}
           {\proj{\kindex}{\texpr}}
           {\rtype_{\kindex}}
  }
\and
  \inferrule
  [\lqchkinj]
  {\lqcheck{\lqenv}{\texpr}{\rtype_{\kindex}}}
  {\lqcheck{\lqenv}
           {\inj{\kindex}{\texpr}}
           {\tsum{\rtype_1}{\rtype_2}}
  }
\and
  \inferrule
  [\lqchkcase]
  {
    \lqcheck{\lqenv}{\texpr}{\tsum{\rtype_1}{\rtype_2}}       \\
    \lqcheck{\lqenvext{\evar_1}{\rtype_1}}{\texpr_1}{\ifdef{\kvars}{\ltype}{\rtype}}  \\
    \lqcheck{\lqenvext{\evar_2}{\rtype_2}}{\texpr_2}{\ifdef{\kvars}{\ltype}{\rtype}}  \\
    \ifdef{\kvars}{\lqwf{\lqenv}{\ifdef{\kvars}{\ltype}{\rtype}}}{\relax}
  }
  {\lqcheck{\lqenv}
    {\casematch{\texpr}{\evar_1}{\texpr_1}{\evar_2}{\texpr_2}}
    {\ifdef{\kvars}{\ltype}{\rtype}}
  }
\end{mathpar}

\vspace{\judgevspace}
\judgementHead{Refinement Subtyping}{\lqsubtype{\lqenv}{\rtype_1}{\rtype_2}}
\begin{mathpar}
  \inferrule*[left=\lqsubbase]
  {\validgrdimp{\lqenv}{\pred}{\pred'}}
  {\lqsubtype{\lqenv}{\reftp{\tprim}{\pred}}{\reftp{\tprim}{\pred'}}}
\and
  \inferrule*[left=\lqsubfun]
  {
    \lqsubtype{\lqenv}{\rtype_{\evar}'}{\rtype_{\evar}} \\
    \lqsubtype{\lqenv,\envbinding{\evar}{\rtype_{\evar}'}}{\rtype}{\rtype'}
  }
  {\lqsubtype{\lqenv}
    {\tfun{(\evar:\rtype_{\evar})}{\rtype}}
    {\tfun{(\evar:\rtype_{\evar}')}{\rtype'}}
  }
\end{mathpar}
\caption{Refined Type-checking}
\label{fig:secondphase}
\end{figure*}

\subparagraph*{Checking $\tdead$-casts.}
The refinement system verifies $\tdead$-casts by treating them as 
special function calls, \ie discharging them via the application 
rule \lqchkapp.
Formally, $\deadcast{\type}{\typeb}{\texpr}$ is treated as call to:
$$
\deadcastname{\type}{\typeb} :: \tfun{\fbot{\elabtype{\type}}}{\fbot{\elabtype{\typeb}}}
$$
The notation $\elabtype{\cdot}$ denotes the elaboration of \srclang types to \tgtlang types~\cite{Dunfield2012}: 
$$    \elabtype{\tprim}               \doteq \tprim 
\quad \elabtype{\tand{\type}{\typeb}} \doteq \tprod{\elabtype{\type}}{\elabtype{\typeb}} 
\quad \elabtype{\tor{\type}{\typeb}}  \doteq \tsum{\elabtype{\type}}{\elabtype{\typeb}}
\quad \elabtype{\tfun{\type}{\typeb}} \doteq \tfun{\elabtype{\type}}{\elabtype{\typeb}} $$
The meta-function $\fbot{\rtype} \doteq \ftx{\rtype}{\rfalse}$ where: 
\[
\begin{tabular}{>{$}l<{$} @{\quad} >{$}c<{$} @{\quad} >{$}l<{$} @{\qquad} >{$}l<{$} @{\quad} >{$}c<{$} @{\quad} >{$}l<{$} }
  \ftx{\tprim}{r}                 & \doteq & \reftp{\tprim}{r}                              & \ftx{\tsum{\rtypeb}{\rtype}}{r}  & \doteq & \tsum{\ftx{\rtypeb}{r}}{\ftx{\rtype}{r}} \\
  \ftx{\tfun{\rtypeb}{\rtype}}{r} & \doteq & \tfun{\ftx{\rtypeb}{\neg r}}{\ftx{\rtype}{r}}  & \ftx{\tprod{\rtypeb}{\rtype}}{r} & \doteq & \tprod{\ftx{\rtypeb}{r}}{\ftx{\rtype}{r}} 
\end{tabular}
\]

Returning to rule \lqchkapp for $\tdead$-casts and inverting, 
expression $\texpr$ gets assigned a refinement type $\rtype$. 
For simplicity we assume this is a base type $\tprim$. 
Due to \lqchksub we get the subtyping constraint: 
$
{\lqsubtype{\lqenv}{\reftp{\tprim}{\pred}}{\reftp{\tprim}{\rfalse}}}
$,
which generates the VC:
$
{\validgrdimp{\lqenv}{\pred}{\rfalse}}
$.
This holds if the environment combined 
with the refinement in the left-hand side is inconsistent, which 
means that the gathered flow conditions are infeasible, hence 
dead-code \cite{Knowles10}.
Thus, the refinements statically ensure that the 
specially marked $\tdead$ values are \emph{never created at run-time}. 
As only $\tdead$ terms cause execution to get stuck, the refinement 
verification phase ensures that the source is indeed type safe.

\subparagraph*{Conditional Checking.}
\lqchkite and \lqchkcase check each branch of a conditional or case
splitting statement, by enhancing the environment with a guard ($\texpr$ or
$\neg\texpr$) or the right binding ($\varbinding{\evar}{\rtype_1}$ or
$\varbinding{\evar}{\rtype_2}$), that encode the boolean test performed at the
condition, or the structural check at the pattern matching, respectively. Crucially,
this allows the use of ``tests'' inside the code to statically verify
$\tdead$-casts and other correctness properties.
The other rules are standard and are described in the refinement
type literature.

\subparagraph*{Correspondence of Elaboration and Refinement Typing.}
The following result establishes the fact that the type $\type$ assigned to a
source expression $\expr$ by elaboration and the type $\rtype$ assigned by
refinement type-checking to the elaborated expression $\texpr$ are connected
with the relation: $\elabtype{\type} = \stripref{\rtype}$, where 
$\stripref{\rtype}$ is merely a (recursive) elimination of all refinements
appearing in $\rtype$. The notation $\elabtype{\tcenv} = \stripref{\lqenv}$ means 
that for each binding
$\envbinding{\evar}{\type} \in \tcenv$ there exists 
$\envbinding{\evar}{\rtype} \in \lqenv$, such that $\elabtype{\type} =
\stripref{\rtype}$, and vice versa.

\begin{lemma}[Correspondence]
If 
$\trans{\tcenv}{\expr}{\type}{}{\texpr}$, 
$\lqcheck{\lqenv}{\texpr}{\rtype}$ and
$\elabtype{\tcenv} = \stripref{\lqenv}$,
then $\elabtype{\type} = \stripref{\rtype}$.
\end{lemma}

\begin{proof}
  By induction on pairs of derivations:
  ${\trans{\tcenv}{\expr}{\type}{}{\texpr}}$ and 
  ${\lqcheck{\lqenv}{\texpr}{\rtype}}$.
  Details of this proof can be found in the extended version of this
  paper~\cite{rsc:appendix}.
\end{proof}

The target language satisfies a progress and preservation
theorem~\cite{Knowles10}:

\begin{theorem}[Refinement Type Safety]\label{theorem:targettypesafe}
If $\infertarget{\emp}{\texpr}{\rtype}$ then either $\texpr$ is a value or
there exists $\texpr'$ such that $\steps{\texpr}{\texpr'}$ and
$\infertarget{\emp}{\texpr'}{\rtype}$.
\end{theorem}

\begin{proof}
  Given by Vazou \etal~\cite{Vazou2014} for a similar language.
\end{proof}

\subsection{Two-Phase Type Safety}\label{sec:safety}

We say that a source term $\expr$ is \emph{well two-typed} if 
there exists a source type $\type$, target term $\texpr$ 
and target (refinement) type $\rtype$ such that:
(1)~$\trans{\emp}{\expr}{\type}{}{\texpr}$, and,
(2)~$\lqcheck{\emp}{\texpr}{\rtype}$.
That is, $\expr$ is well two-typed if it elaborates 
to a refinement typed target.
The Consistency Theorems
\ref{theorem:consistency} and \ref{theorem:rev:consistency}, along with 
the Safety Theorem \ref{theorem:targettypesafe}, yield 
end-to-end soundness: well two-typed terms do not 
get stuck, and step to well two-typed terms.

\begin{theorem}[Two-Phase Soundness]\label{theorem:twophase-soundness}
  If $\expr$ is well two-typed then, either $\expr$ is a value, or 
  there exists $\expr'$ such that: 
  \begin{enumerate}[label={(\arabic*)}]
  \inditem \textbf{(Progress)} $\steps{\expr}{\expr'}$
  \inditem \textbf{(Preservation)} $\expr'$ is well two-typed.
  \end{enumerate}
\end{theorem}

\begin{proof}
  By induction on pairs of derivations:
  ${\trans{\tcenv}{\expr}{\type}{}{\texpr}}$ and 
  ${\lqcheck{\lqenv}{\texpr}{\rtype}}$.
  Details are to be found in the extended version of this
  paper~\cite{rsc:appendix}.
\end{proof}

\section{Related Work}\label{sec:related}

We focus on the highlights of prior work relevant to the key points of our
technique: static types for dynamic
languages, intersections and union types, and refinement types.
 
\subparagraph*{Types for Dynamic Functional Languages.}
\emph{Soft typing}~\cite{Car91} incorporates static analysis to
statically type dynamic languages: whenever a program cannot be proven
safe statically, it is not rejected, but instead runtime checks are
inserted.
Henglein and Rehof~\cite{Hen95} build up on this work by extending
soft typing's monomorphic typing to polymorphic coercions and
providing a translation of Scheme programs to ML.
These works foreshadow the notion of \emph{gradual typing}
\cite{Sie07} that allows the programmer to control the boundary
between static and dynamic checking depending on the trade-off between
the need for static guarantees and deployability. 
Returning to purely static enforcement, Tobin-Hochstadt \etal~\cite{Tob08,typedracket}
formalize the support for type tests as \emph{occurrence typing} and
extend it to an inter-procedural, higher-order setting by introducing
propositional \emph{latent predicates} that reflect the result of
tests in Typed Racket function signatures.

\subparagraph*{Types for Dynamic Imperative Languages.}
Thiemann \cite{Thiemann05} and Anderson \etal~\cite{Drossopoulou05}
describe early attempts towards static type systems for JavaScript,
and Furr \etal~\cite{Fur09a} present DRuby, a tool for type inference
for Ruby scripts.
However, these systems do not handle value-based
overloading (like \tsc, DRuby allows overloaded specifications for
external functions).
\emph{Flow typing}~\cite{lambdajs} and TeJaS~\cite{Lerner13} account
for tests using flow analysis, bringing occurrence typing to the
imperative \jsc setting,
but, unlike our approach, they restrict themselves to
a \emph{fixed} set of type-testing idioms (\eg @typeof@), precluding
\emph{general} value-based overloading \eg as in @reduce@ from
Figure~\ref{fig:reduce}.

\subparagraph*{Logics for Dynamic Languages.}
The intuition of expressing subtyping relations as logical implication
constraints and using SMT solvers to discharge these constraints
allows for a more extensive variety of typing idioms. 
Bierman \etal~\cite{dminor}
investigate semantic subtyping in a first order language with
refinements and type-test expressions.
In \textit{nested refinement types}~\cite{NestedPOPL12}, 
the typing relation itself is a predicate in the refinement logic and 
a feature-rich language of predicates accounts for
heavily dynamic idioms, like run-time type tests, value-indexed
dictionaries, polymorphism and higher order functions.
While program logics allow the use of arbitrary tests to establish typing,
the circular dependency between values and basic types leads to two 
significant problems in theory and practice.
First, the circular dependency complicates the \emph{metatheory} which
makes it hard to add extra (basic) typing features (\eg polymorphism,
classes) to the language.
Second, the circular dependency complicates the \emph{inference} of 
types and refinements, leading to significant annotation overheads which
make the system difficult to use in practice.
In contrast, two-phase typing allows arbitrary 
type tests while enabling the trivial composition of soundness proofs and 
inference algorithms.

\subparagraph*{Intersection and Union Types.}
Central to our elaboration phase are intersection 
and union types: Pierce~\cite{Pierce91} indicates the connection 
between unions and intersections with sums and products, that
is the basis of Dunfield's elaboration scheme~\cite{Dunfield2012} on 
which we build.
However, Dunfield studies \emph{static} source languages 
that use \emph{explicit} overloading via a merge 
operator~\cite{forsythe}.  
In contrast, we target \emph{dynamic} source languages 
with implicit value based overloading, and hence must 
account for ``ill-typed'' terms via \tdead-casts
discharged via the second phase refinement check.
Castagna \etal~\cite{Castagna92} describe a $\lambda \&$-calculus, where 
functions are overloaded by combining several different 
branches of code. The branch to be executed is determined 
at run-time by using the arguments' typing information.  
This technique resembles the code duplication that 
happens in our approach, but overload resolution 
(\ie deciding which branch is executed) is 
determined at runtime whereas we do so statically.

\subparagraph*{Refinement Types.}
DML~\cite{XiPfenning99} is an early refinement type system composing
ML's types with a decidable constraint system.
\emph{Hybrid type checking}~\cite{Knowles10} uses arbitrary
refinements over basic types. A static type system verifies
basic specifications and more complex ones are defered to dynamically
checked contracts, since the specification logic is
statically undecidable. In these cases, the source language is well
typed (ignoring refinements), and lacks intersections and unions.
Our second phase can use Liquid Types~\cite{LiquidPLDI08} to
infer refinements using predicate abstraction.

\section{Conclusions and Future Work}\label{sec:conclusion}

In this paper, we introduce two-phased typing, a novel framework 
for analyzing dynamic languages where value-based overloading is 
ubiquitous. 
The advantage of our approach over previous methods is that, unlike
purely type-based approaches~\cite{typedracket}, we are not limited to
a fixed set of tag- or type- tests, and unlike purely program
logic-based approach~\cite{NestedPOPL12}, we can decouple
reasoning about basic typing from values, thereby enabling inference.

Hence, we believe two-phased typing provides an ideal foundation for
building expressive and automatic analyses for imperative scripting 
languages like \jsc. 
However, this is just the first step; much remains to achieve this goal. 
In particular we must account for the imperative features of the 
language. We believe that decoupling makes it possible to address
this problem by applying various methods for tracking mutation and 
aliasing~\cite{Zibin2007} in the \emph{first phase},
and we intend to investigate this route in future
work to obtain a practical verifier for \tsc.

\bibliography{main}

\appendix

\renewcommand{\thesubsection}{\Alph{subsection}}

\counterwithin{equation}{enumi}
\counterwithin{enumii}{enumi}
\counterwithin{enumiii}{enumii}

\renewcommand{\labelenumi}{$\bullet$}
\renewcommand{\labelenumii}{$\triangleright$}
\renewcommand{\labelenumiii}{--}

\setcounter{theorem}{0}
\setcounter{assumption}{0}

\clearpage
\section*{Appendix}

\setlist{itemsep=1em, topsep=1em}

We now provide detailed versions of the proofs mentioned in the main part of the
paper. This part reuses the definitions of \S~3, \S~4 and \S~5, and is
structured in three main sections:
\begin{itemize}
  \item Assumptions (\ref{sec:assumptions})
  \item Lemmas (\ref{sec:lemmas})
  \item Theorems (\ref{sec:theorems})
\end{itemize}

Sections \ref{sec:assumptions} and \ref{sec:lemmas} build up to the main results:

\begin{itemize}
    \item Consistency and Reverse Consistency Theorems
      (\ref{theorem:consistency}, \ref{theorem:rev:consistency})
    \item Two-phase Safety Theorem (\ref{theorem:twophase-safety})
  \end{itemize}

For the remainder of the document we are going to use the plain version
of the elaboration relation, \ie without mode annotations:
$$
\trans{\tcenv}{\expr}{\type}{}{}{\texpr}
$$
The annotations on the judgment merely determine which rules are available at
type-checking. The majority of the proofs below involve induction over the
elaboration derivation, which is fixed once type-checking is complete, so the
annotations can be safely ignored.

In certain lemmas the reader is referred to Dunfield's techniques from his work
on the elaboration of intersection and union types~\cite{Dunfield2012}. The proofs 
there refer to a language similar but not exactly the same as ours. The main 
proof ideas, however, hold.


\section{Assumptions}\label{sec:assumptions}

\begin{assumption}[Primitive Constant Application] 
\IfAppendix{\label{assum:const:app:app}}{\label{assum:const:app}}
\IfAppendix{
If 
\begin{enumerate}[label={(\arabic*)}]
    \inditem $\trans{\emp}{\vconst}{\tfun{\type}{\typeb}}{}{\vconst}$,
  \inditem $\trans{\emp}{\val}{\type}{}{\tval}$,
  \inditem $\tval \not \equiv \deadcast{\cdot}{\type}{\cdot}$,
\end{enumerate}
then
\begin{itemize}
  \inditem $\steps{\app{\vconst}{\val}}{\primapp{\vconst}{\val}}$
  \inditem $\steps{\app{\vconst}{\tval}}{\primapp{\vconst}{\tval}}$ 
  \inditem $\trans{\emp}{\primapp{\vconst}{\val}}{\typeb}{}{\primapp{\vconst}{\tval}}$ 
\end{itemize}
}
{
  If 
  (1)~$\trans{\emp}{\vconst}{\tfun{\type}{\typeb}}{}{\vconst}$,
  (2)~$\trans{\emp}{\val}{\type}{}{\tval}$, and
  (3)~$\tval \not \equiv \deadcast{\cdot}{\type}{\cdot}$,
then
  (i)~$\steps{\app{\vconst}{\val}}{\primapp{\vconst}{\val}}$,
  (ii)~$\steps{\app{\vconst}{\tval}}{\primapp{\vconst}{\tval}}$, and 
  (iii)~$\trans{\emp}{\primapp{\vconst}{\val}}{\typeb}{}{\primapp{\vconst}{\tval}}$.
}
\end{assumption}

\begin{assumption}[Lambda Application] 
\IfAppendix{\label{assum:lambda:app:app}}{\label{assum:lambda:app}}
If 
\begin{enumerate}[label={(\arabic*)}]
    \inditem
      $\trans{\emp}{\elambda{\evar}{\expr}{}}
                   {\tfun{\type}{\typeb}}{}
                   {\elambda{\evar}{\texpr}{}}$,
  \inditem $\trans{\emp}{\val}{\type}{}{\tval}$,
  \inditem $\tval \not \equiv \deadcast{\cdot}{\type}{\cdot}$,
\end{enumerate}
then
\begin{itemize}
    \inditem $\steps{\app{\parens{\elambda{\evar}{\expr}{}}}{\val}}
                  {\appsubst{\esubst{\val}{\evar}}{\expr}}$
    \inditem $\steps{\app{\parens{\elambda{\evar}{\texpr}{}}}{\tval}}
                  {\appsubst{\esubst{\tval}{\evar}}{\texpr}}$
\end{itemize}
\end{assumption}

\begin{assumption}[Canonical Forms]\label{assumption:canonical:forms}
  \mbox{}
\begin{enumerate}[label={(\arabic*)}]

  \inditem If $\trans{\tcenv}{\elambda{\evar}{\expr}{}}{\tfun{\type}{\typeb}}{}{\tval}$ then
      \begin{itemize}
        \item $\tval \equiv \elambda{\evar}{\texpr}{}$ for some $\texpr$, or 
        \item $\tval \equiv \deadcast{\cdot}{\tfun{\type}{\typeb}}{\tval'}$ for
          some $\tval'$
      \end{itemize}

  \inditem If $\trans{\tcenv}{\vconst}{\type}{}{\tval}$ then
      \begin{itemize}
        \item $\tval \equiv \vconst$, or 
        \item $\tval \equiv \deadcast{\cdot}{\type}{\tval'}$ for
          some $\tval'$
      \end{itemize}

\end{enumerate}

\end{assumption}

\section{Auxiliary lemmas}\label{sec:lemmas}

\begin{lemma}[Multi-Step Source Evaluation Context]\label{lemma:src:evalctx}

If $\stepsmany{\expr}{\expr'}$ then $\stepsmany{\evalctxarg{\expr}}{\evalctxarg{\expr'}}$.
\end{lemma}

\begin{proof}
  Based on Lemma 7 of Dunfield's elaboration~\cite{Dunfield2012}.
\end{proof}

\begin{lemma}[Multi-Step Target Evaluation Context]\label{lemma:tgt:evalctx}
  \begin{itemize} 
    \item If $\stepsmany{\texpr}{\texpr'}$ then $\stepsmany{\tevalctxarg{\texpr}}{\tevalctxarg{\texpr'}}$
    \item If $\stepsmanyone{\texpr}{\texpr'}$ then $\stepsmanyone{\tevalctxarg{\texpr}}{\tevalctxarg{\texpr'}}$
  \end{itemize}
\end{lemma}

\begin{proof}
  Similar to proof of Lemma~\ref{lemma:src:evalctx}.
\end{proof}

\begin{lemma}[Unions/Injections]\label{lemma:unions}
If $\trans{\tcenv}{\expr}{\tor{\type_1}{\type_2}}{}{\inj{k}{\texpr}}$ then 
$\trans{\tcenv}{\expr}{\type_k}{}{\texpr}$.
\end{lemma}

\begin{proof}
  Based on Lemma 8 of Dunfield's elaboration~\cite{Dunfield2012}.
\end{proof}

\begin{lemma}[Intersections/Pairs]\label{lemma:intersections}
If
$\trans{\tcenv}{\expr}
          {\tand{\type_1}{\type_2}}{}
          {\pair{\texpr_1}{\texpr_2}}$ 
then there exist $\expr_1'$ and $\expr_2'$ such that:
\begin{enumerate}[label={(\arabic*)}]
  \inditem $\stepsmany{\expr_1}{\expr_1'}$ and $\trans{\tcenv}{\expr_1'}{\type_1}{}{\texpr_1}$
  \inditem $\stepsmany{\expr_2}{\expr_2'}$ and $\trans{\tcenv}{\expr_2'}{\type_2}{}{\texpr_2}$
\end{enumerate}
\end{lemma}

\begin{proof}
  Based on Lemma 9 of Dunfield's elaboration~\cite{Dunfield2012}.
\end{proof}

\begin{lemma}[Beta Reduction Canonical Form]   \label{lemma:beta:reduction}
  If 
  \begin{enumerate}[label={(\arabic*)}]
  \inditem $\trans{\emp}{\elambda{\evar}{\expr}{}}{\tfun{\type}{\typeb}}{}{\tval_1}$,
  \inditem $\trans{\emp}{\val_2}{\type}{}{\tval_2}$,
  \inditem $\steps{\app{\parens{\elambda{\evar}{\expr}{}}}{\val_2}}
                  {\appsubst{\esubst{\val_2}{\evar}}{\expr}}$
\end{enumerate}
  Then $\tval_1 \equiv \elambda{\evar}{\texpr}{}$ for some $\texpr$.
\end{lemma}


\begin{lemma}[Primitive Reduction Canonical Form]   \label{lemma:primitive:reduction}
  If 
  \begin{enumerate}[label={(\arabic*)}]
  \inditem $\trans{\emp}{\vconst}{\tfun{\type}{\typeb}}{}{\tval_1}$ 
  \inditem $\trans{\emp}{\val}{\type}{}{\tval_2}$,
  \inditem $\steps{\app{\vconst}{\val}}
                  {\primapp{\vconst}{\val}}$
\end{enumerate}
Then:
\begin{itemize}
  \inditem $\tval_1 \equiv \vconst$
  \inditem $\tval_2 \not\equiv \deadcast{\cdot}{\type}{\cdot}$
\end{itemize}
\end{lemma}


\begin{lemma}[Conditional Canonical Form]   \label{lemma:reduce:ite}
  If 
  \begin{enumerate}[label={(\arabic*)}]
      \inditem $\trans{\emp}{\vconst}{\tbool}{}{\tval}$,
      \inditem $\trans{\emp}{\expr_1}{\type}{}{\texpr_1}$  and
      $\trans{\emp}{\expr_2}{\type}{}{\texpr_2}$,
    \item $\steps{\ite{\vconst}{\expr_1}{\expr_2}}{\expr_k}$
\end{enumerate}
Then:
  \begin{itemize}
    \inditem $k = 1 \imp \vconst \equiv \tval \equiv \vtrue$.
    \inditem $k = 2 \imp \vconst \equiv \tval \equiv \vfalse$.

  \end{itemize}
\end{lemma}


\begin{lemma}[Value Monotonicity]
  \IfAppendix{\label{lemma:valmonot:app}}{\label{lemma:valmonot}}
If $\trans{\tcenv}{\expr}{\type}{}{\tval}$, 
then there exists $\val$ \st:
\begin{enumerate}[label={(\arabic*)}]
\inditem $\stepsmany{\expr}{\val}$
\inditem $\trans{\tcenv}{\val}{\type}{}{\tval}$
\inditem $\foralli{\;i\;.;\stepsmany{\expr}{\expr_i}}
                  {\trans{\tcenv}{\expr_i}{\type}{}{\tval}}$
\end{enumerate}
\end{lemma}

\IfAppendix{

\begin{proof}

Parts (1) and (2) of the lemma has been proved by Dunfield~\cite{Dunfield2012} for a similar
language, so here we are just going to prove part (3).



We will show this by induction on the length $i$ of the path:
$\stepsmany{\expr}{\expr_i}$. 
\begin{enumerate}
  \item For the case $i = 0$: $\expr = \expr_i$, so it trivially holds.

  \item Suppose it holds for $i = k$, \ie for
    $\stepsmany{\expr}{\expr_k}$, it holds that:
    \begin{align}
    \trans{\tcenv}{\expr_k}{\type}{}{\tval} \label{eq:monot:0}
    \end{align}
    
    We will show that it holds for
    $i = k+1$, \ie for $\expr_{k+1}$ such that: 
    \begin{align}
      \steps{\expr_k}{\expr_{k+1}}  \label{eq:monot:1}
    \end{align}

    We will do this by induction on the derivation \eqref{eq:monot:0}, but
    limit ourselves to the terms $\expr_{k}$ that
    elaborate to \textit{values}:
    \begin{enumerate}

      \item Cases \tchkconst, \tchkvar, \tchkinterintro, \tchklambda:
        For these cases, term $\expr_k$ is already a value, so doesn't step.
%
%
%

      \item Case \tchkup (assume left injection -- the case for right injection is similar):
        $$
        \inferrule[]
        {
          \trans{\tcenv}{\expr_{k}}{\type_1}{}{\tval} \\
          \tcwf{\tor{\type_1}{\type_2}}
        } 
        {\trans{\tcenv}{\expr_k}{\tor{\type_1}{\type_2}}{}{\inj{1}{\tval}}}
        $$
        By inversion:
        $$
        \trans{\tcenv}{\expr_k}{\type_1}{}{\tval}
        $$
        By i.h., using \eqref{eq:monot:1}:
        $$
        \trans{\tcenv}{\expr_{k+1}}{\type_1}{}{\tval}
        $$
        Applying rule \tchkup on the latter one:
        $$
        \trans{\tcenv}{\expr_{k+1}}{\tor{\type_1}{\type_2}}{}{\inj{1}\tval}
        $$

    \end{enumerate}
    
\end{enumerate}

\end{proof}
}
{
\begin{proof}
  The first two parts are handled similarly to
  Dunfield's~\cite{Dunfield2012} Lemma 11. 
  The last part is proved by induction on the length of the path
  $\stepsmany{\expr}{\expr_i}$.
  Full details of this proof can be found in the accompanying
  technical report.
\end{proof}
}

\begin{lemma}[Reverse Value Monotonicity]
  \IfAppendix{\label{lemma:rev:valmonot:app}}{\label{lemma:rev:valmonot}}
If $\trans{\tcenv}{\val}{\type}{}{\texpr}$, 
then exists $\tval$ s.t.: 
%
%
$\stepsmany{\texpr}{\tval}$ and $\trans{\tcenv}{\val}{\type}{}{\tval}$.
%
\end{lemma}

\IfAppendix{
\begin{proof}
Similar to proof of Lemma~\ref{lemma:valmonot:app}.
\end{proof}

}
{
\begin{proof}
Similar to proof of Lemma~\ref{lemma:valmonot}.
\end{proof}
}

\begin{lemma}[Substitution]\label{lemma:subst}
If $\trans{\tcenvext{\evar}{\type}}{\expr}{\type'}{}{\texpr}$ and 
$\trans{\tcenv}{\val}{\type}{}{\tval}$ then 
$\trans{\tcenv}{\appsubst{\esubst{\val}{\evar}}{\expr}}{\type'}{}{\appsubst{\esubst{\tval}{\evar}}{\texpr}}$.
\end{lemma}

\IfAppendix{
 \begin{proof}
  Based on Lemma 12 of Dunfield's elaboration~\cite{Dunfield2012}.
 \end{proof}
}{
 \begin{proof}
  Similar to Dunfield's substitution proof~\cite{Dunfield2012} (Lemma 12).
 \end{proof}
}

\begin{corollary}[Target Multi-step Preservation]
\label{corollary:multi:preservation}
If $\lqcheck{\emp}{\texpr}{\rtype}$
and $\stepsmany{\texpr}{\texpr'}$ 
then $\lqcheck{\emp}{\texpr'}{\rtype}$.
\end{corollary}

\begin{proof}
  Stems from Theorem~7 from main paper.
\end{proof}

\begin{corollary}[\tdead-cast Invalid]
\label{corollary:deadcast:invalid}
$\emp \not\vdash \deadcast{\type}{\typeb}{\texpr} :: \rtype$
\end{corollary}


\IfAppendix{
  \setcounter{enumi}{0}
  \setcounter{enumii}{0}
  \setcounter{equation}{0}
}{}

\begin{lemma}[Correspondence]
\IfAppendix{
\begin{align}
  \forall \tcenv, \expr, \type, \texpr, \lqenv, \rtype  \quad . \quad &
  \trans{\tcenv}{\expr}{\type}{}{\texpr}                \label{eq:elaboration:1} \\
  \wedge \quad & \lqcheck{\lqenv}{\texpr}{\rtype}       \label{eq:elaboration:2} \\
  \wedge \quad & \elabtype{\tcenv} = \stripref{\lqenv}  \label{eq:elaboration:3} \\
  \imp   \quad & \elabtype{\type} = \stripref{\rtype} \notag
\end{align}
}
{If 
%
$\trans{\tcenv}{\expr}{\type}{}{\texpr}$, 
$\lqcheck{\lqenv}{\texpr}{\rtype}$ and
$\elabtype{\tcenv} = \stripref{\lqenv}$,
then $\elabtype{\type} = \stripref{\rtype}$.
}
\end{lemma}

\IfAppendix{
\begin{proof}
%

  We prove this by induction on pairs \rulename{T-\textit{Rule}}/\rulename{R-\textit{Rule}} of derivations: 

  \begin{align*}
  \trans{\tcenv}{\expr}{\type}{}{\texpr} \\
  \lqcheck{\lqenv}{\texpr}{\rtype}
  \end{align*}

\begin{enumerate}

\item Case \tchkconst/\lqchkconst:
  \begin{align*}
    \trans{\tcenv}{\vconst}{\tconst}{}{\vconst}  \\
    \lqcheck{\lqenv}{\vconst}{\reftp{\tconst}{\vv=\constval{\vconst}}}
  \end{align*}
  
  Meta-function $\singletonname$ operates entirely on the refinement so it
  holds that:
  $$
  \stripref{\reftp{\tconst}{\vv=\constval{\vconst}}} = \stripref{\tconst}
  $$
  Also, it holds that:
  $$
  \elabtype{\tconst} = \stripref{\tconst}
  $$

\item Case \tchkvar/\lqchkvar:
  $$
  \begin{tabu} to \textwidth {X[c]X[c]}
    \mbox{\inferrule[]
    {\envbinding{\evar}{\type}\in \tcenv}
    {\trans{\tcenv}{\evar}{\type}{}{\evar}}}
    &
    \mbox{\inferrule[]
      {\envbinding{\evar}{\rtype} \in \lqenv}
    {\lqcheck{\lqenv}{\evar}{\singleton{\rtype}{\evar}}}}
  \end{tabu}
  $$
  By inversion: 
  \begin{align}
    \envbinding{\evar}{\type}\in \tcenv               \label{eq:elaboration:4} \\
    \envbinding{\evar}{\rtype}\in \lqenv               
    \label{eq:elaboration:5}
  \end{align}

  If $\evar$ is bound multiple times in $\tcenv$ and $\lqenv$, we assume the
  we have picked the correct instances from each environment. 

  By \eqref{eq:elaboration:3} we have that:
  $$
  \elabtype{\tcenvidx{\evar}} = \stripref{\lqenvidx{\evar}}
  $$
  Also meta-function $\singletonname$ operates entirely on the refinement so it
  holds that:
  \begin{align}
  \stripref{\rtype} = \stripref{\singleton{\rtype}{\evar}}  \label{eq:elaboration:6}
  \end{align}
  By 
  \eqref{eq:elaboration:4},
  \eqref{eq:elaboration:5} and
  \eqref{eq:elaboration:6}
  it holds that: 
  $$
  \elabtype{\type} = \stripref{\singleton{\rtype}{\evar}}
  $$

\item Case \tchkletin/\lqchkletin:

  From the first premise of the implication:
  $$
  \inferrule[]
  {
    \trans{\deftcstate}{\expr_1}{\type_1}{}{\texpr_1} \\
    \trans{\tcenvext{\evar}{\type_1}}{\expr_2}{\type_2}{}{\texpr_2}
  }
  {\trans{\tcenv}
         {\eletin{\evar}{\expr_1}{\expr_2}}
         {\type_2}
         {}
         {\eletin{\evar}{\texpr_1}{\texpr_2}}}
  $$
  By inversion:
  \begin{align}
    \trans{\deftcstate}{\expr_1}{\type_1}{}{\texpr_1}               \label{eq:elaboration:7}    \\
    \trans{\tcenvext{\evar}{\type_1}}{\expr_2}{\type_2}{}{\texpr_2} \label{eq:elaboration:8}
  \end{align}

  From the second premise of the implication:
  $$
  \inferrule[]
  {
    \lqcheck{\lqenv}{\texpr_1}{\rtype_1} \\
    \lqcheck{\lqstateext{\evar}{\rtype_1}}{\texpr_2}{\rtype_2}   \\
    \ifdef{\kvars}{\lqwf{\lqenv}{\rtype_2}}{\relax}
  }
  {\lqcheck{\lqenv}{\eletin{\evar}{\texpr_1}{\texpr_2}}{\rtype_2}}
  $$
  By inversion:
  \begin{align}
    \lqcheck{\lqenv}{\texpr_1}{\rtype_1}                            \label{eq:elaboration:9}    \\
    \lqcheck{\lqstateext{\evar}{\rtype_1}}{\texpr_2}{\rtype_2}      \label{eq:elaboration:10}
  \end{align}
  By i.h. on
  \eqref{eq:elaboration:3},
  \eqref{eq:elaboration:7} and
  \eqref{eq:elaboration:9}:

\begin{align} \elabtype{\type_1} = \stripref{\rtype_1} \label{eq:elaboration:11} \end{align}
  By \eqref{eq:elaboration:3} and \eqref{eq:elaboration:11}:
  \begin{align} 
    \elabtype{\tcenvext{\evar}{\type_1}} = \stripref{\lqenvext{\evar}{\rtype_1}}    \label{eq:elaboration:12}
  \end{align}
  By i.h. on 
  \eqref{eq:elaboration:8},
  \eqref{eq:elaboration:10} and
  \eqref{eq:elaboration:12}:
  $$
  \elabtype{\type_2} = \stripref{\rtype_2}
  $$

\item Case \tchkite/\lqchkite: \textit{Similar to previous case}.

\item Case \tchkinterintro/\lqchkpair: 

  From the first premise of the implication:
  $$
  \inferrule[]
  {
    \foralli {\kindex \in \{1,2\}}
    {\trans{\tcenv}{\val}{\type_{\kindex}}{}{\tval_{\kindex}}}
  }
  {
    \trans{\tcenv}
          {\val}
          {\tand{\type_1}{\type_2}}
          {}{\pair{\tval_1}{\tval_2}}
  }
  $$
  By inversion:
  \begin{align}
    \foralli {\kindex \in \{1,2\}}
    {\trans{\tcenv}{\val}{\type_{\kindex}}{}{\tval_{\kindex}}} \label{eq:elaboration:13}
  \end{align}

  From the second premise of the implication:
  $$
  \inferrule[]
  {\foralli{\kindex \in \{1,2\}}{\lqcheck{\lqenv}{\tval_{\kindex}}{\rtype_{\kindex}}}}
  {\lqcheck{\lqenv}
         {\pair{\tval_1}{\tval_2}}
         {\tprod{\rtype_1}{\rtype_2}}
  }
  $$
  By inversion:
  \begin{align}
    \foralli {\kindex \in \{1,2\}}
    {\lqcheck{\lqenv}{\tval_{\kindex}}{\rtype_{\kindex}}} \label{eq:elaboration:14}
  \end{align}
  By i.h. on
  \eqref{eq:elaboration:3},
  \eqref{eq:elaboration:13} and
  \eqref{eq:elaboration:14}:
  \begin{align}
    \foralli {\kindex \in \{1,2\}}
    {\elabtype{\type_{\kindex}} = \stripref{\rtype_{\kindex}}}
  \end{align}
  Using properties of $\elabtype{\cdot}$ and $\stripref{\cdot}$:
  $$
  \elabtype{\tand{\type_1}{\type_2}} = \tprod{\elabtype{\type_1}}{\elabtype{\type_2}} 
    = \tprod{\stripref{\rtype_1}}{\stripref{\rtype_2}} = \stripref{\tprod{\rtype_1}{\rtype_2}}
  $$

\item Case \tchkinterelim/\lqchkproj: 
  \textit{Straightforward based on earlier cases.}

\item Case \tchklambda/\lqchklambda: 
  \textit{Straightforward based on earlier cases.}

\item Case \tchkapp/\lqchkapp: 
  \textit{Straightforward based on earlier cases.}

\item Case \tchkup/\lqchkinj: 
  \textit{Straightforward based on earlier cases.}

\item Case \tchkdown/\lqchkcase:

  From the first premise of the implication:
  $$
  \inferrule[]
  {
    \phantom{\trans{\tcenv}{\expr_0}{\tor{\type_1}{\type_2}}{}{\texpr_0}}
    \\
    \trans{\tcenvext{\evar_1}{\type_1}}
             {\evalctxarg{\evar_1}}
             {\typeb}
             {}
             {\texpr_1}
    \\\\
    \trans{\tcenv}{\expr_0}{\tor{\type_1}{\type_2}}{}{\texpr_0}
    \\
    \trans{\tcenvext{\evar_2}{\type_2}}
             {\evalctxarg{\evar_2}}
             {\typeb}
             {}
             {\texpr_2}
  }
  {
    \trans{\tcenv}
             {\evalctxarg{\expr_0}}
             {\typeb}
             {}{\casematch{\texpr_0}{\evar_1}{\texpr_1}{\evar_2}{\texpr_2}}
  } 
  $$
  By inversion:
  \begin{align}
    \trans{\tcenv}{\expr_0}{\tor{\type_1}{\type_2}}{}{\texpr_0}
    \label{eq:elaboration:15}
    \\
    \trans{\tcenvext{\evar_1}{\type_1}}
             {\evalctxarg{\evar_1}}
             {\typeb}
             {}
             {\texpr_1}
    \label{eq:elaboration:16}
    \\
    \trans{\tcenvext{\evar_2}{\type_2}}
             {\evalctxarg{\evar_2}}
             {\typeb}
             {}
             {\texpr_2}
    \label{eq:elaboration:17}
  \end{align}

  From the second premise of the implication:
  $$
  \inferrule[]
  {
    \lqcheck{\lqenv}{\texpr_0}{\tsum{\rtype_1}{\rtype_2}}       \\
    \lqcheck{\lqenvext{\evar_1}{\rtype_1}}{\texpr_1}{\ifdef{\kvars}{\ltype}{\rtype}}  \\
    \lqcheck{\lqenvext{\evar_2}{\rtype_2}}{\texpr_2}{\ifdef{\kvars}{\ltype}{\rtype}}  \\
    \ifdef{\kvars}{\lqwf{\lqenv}{\ifdef{\kvars}{\ltype}{\rtype}}}{\relax}
  }
  {\lqcheck{\lqenv}
    {\casematch{\texpr_0}{\evar_1}{\texpr_1}{\evar_2}{\texpr_2}}
    {\ifdef{\kvars}{\ltype}{\rtype}}
  }
  $$
  By inversion:
  \begin{align}
    \lqcheck{\lqenv}{\texpr_0}{\tsum{\rtype_1}{\rtype_2}}      \label{eq:elaboration:18} \\
    \lqcheck{\lqenvext{\evar_1}{\rtype_1}}{\texpr_1}{\ifdef{\kvars}{\ltype}{\rtype}}  \label{eq:elaboration:19} \\
    \lqcheck{\lqenvext{\evar_2}{\rtype_2}}{\texpr_2}{\ifdef{\kvars}{\ltype}{\rtype}}  \label{eq:elaboration:20}
  \end{align}

  By i.h. on 
  \eqref{eq:elaboration:3},
  \eqref{eq:elaboration:15} and 
  \eqref{eq:elaboration:18}:
  $$
  \elabtype{\tor{\type_1}{\type_2}} = \stripref{\tsum{\rtype_1}{\rtype_2}}
  $$
  From properties of type elaboration and refinement types:
  \begin{align*}
    \elabtype{\tor{\type_1}{\type_2}} &= \tsum{\elabtype{\type_1}}{\elabtype{\type_2}} \\
    \stripref{\tsum{\rtype_1}{\rtype_2}} &= \tsum{\stripref{\rtype_1}}{\stripref{\rtype_2}}
  \end{align*}
  The right-hand side of the last two equations are tagged unions, so it is 
  possible to  match the consituent parts by structure:
  $$
  \elabtype{\type_1} = \stripref{\rtype_1}  \quad\text{and}\quad
  \elabtype{\type_2} = \stripref{\rtype_2}
  $$
  Combining the last equation with \eqref{eq:elaboration:3}:
  \begin{align} 
    \elabtype{\tcenvext{\evar}{\type_1}} = \stripref{\lqenvext{\evar}{\rtype_1}}    \label{eq:elaboration:21} \\
    \elabtype{\tcenvext{\evar}{\type_2}} = \stripref{\lqenvext{\evar}{\rtype_2}}    \label{eq:elaboration:22}
  \end{align}
  By i.h. on 
  \eqref{eq:elaboration:16}, 
  \eqref{eq:elaboration:19} and
  \eqref{eq:elaboration:21} (or
  \eqref{eq:elaboration:17}, 
  \eqref{eq:elaboration:20} and
  \eqref{eq:elaboration:22}):
  $$
  \elabtype{\typeb} = \stripref{\rtype}
  $$

\item Case \tchkdead/\lqchkapp:

  From the first premise of the implication:
  $$
  \inferrule[]
    {
      \trans{\tcenv}{\expr}{\type}{}{\texpr} \\
      \typetag{\type}\cap \typetag{\typeb} = \emptyset
    }
    {
      \trans{\tcenv}{\expr}{\typeb}{}{\deadcast{\type}{\typeb}{\texpr}}
    }
  $$

  From the second premise of the implication:
  $$
  \inferrule[]
  {
    \lqcheck{\lqenv}{\deadcastname{\type}{\typeb}}{\tfun{\fbot{\elabtype{\type}}}{\fbot{\elabtype{\typeb}}}}    \\
    \lqcheck{\lqenv}{\texpr}{\rtypeb}
  }
  {\lqcheck{\lqenv}
    {\deadcast{\type}{\typeb}{\texpr}}
         {\appsubst{\rsubst{\texpr}{\evar}}{\fbot{\elabtype{\typeb}}}}
  }
  $$
  The result type of the last derivation can also be written as:
  $$
  \appsubst{\rsubst{\texpr}{\evar}}{\fbot{\elabtype{\typeb}}}  = 
  \fbot{\elabtype{\typeb}}
  $$
  Because after the application of $\fbot{\cdot}$ all original refinement get
  erased.
  Also, after removing the refinements:
  $$
  \stripref{\fbot{\elabtype{\typeb}}} = \elabtype{\typeb}
  $$

\end{enumerate}

\end{proof}

}
{
  \begin{proof}
  By induction on pairs of derivations:
  ${\trans{\tcenv}{\expr}{\type}{}{\texpr}}$ and 
  ${\lqcheck{\lqenv}{\texpr}{\rtype}}$.
  Details can be found in the accompanying report.
  \end{proof}
}

\section{Theorems}\label{sec:theorems}

\begin{theorem}[Consistency]\label{theorem:consistency}
If $\trans{\emp}{\expr}{\type}{}{\texpr}$ and $\steps{\texpr}{}{\texpr'}$ then there
exists $\expr'$ such that $\stepsmany{\expr}{\expr'}$ and
$\trans{\emp}{\expr'}{\type}{}{\texpr'}$.
\end{theorem}

\IfAppendix{
\begin{proof}

By induction on the derivation $\trans{\emp}{\expr}{\type}{}{\texpr}$

\begin{enumerate}

\item Cases \tchkconst, \tchkvar, \tchkinterintro, \tchkinterelim
  and \tchklambda:

  The respective target expression does not step.

\item Case \tchkletin: 
  $$
  \inferrule[]
  {
    \trans{\emp}{\expr_1}{\type_1}{}{\texpr_1} \\
    \trans{\envbinding{\evar}{\type_1}}{\expr_2}{\type_2}{}{\texpr_2}
  }
  {\trans{\emp}{\eletin{\evar}{\expr_1}{\expr_2}}{\type_2}{}{\eletin{\evar}{\texpr_1}{\texpr_2}}}
  $$
  By inversion:
  \begin{align}
    \trans{\emp}{\expr_1}{\type_1}{}{\texpr_1}                              \label{eq:consistency:1} \\
    \trans{\envbinding{\evar}{\type_1}}{\expr_2}{\type_2}{}{\texpr_2}       \label{eq:consistency:2}
  \end{align}

  Cases on the form of $\steps{\texpr}{\texpr'}$:
  \begin{enumerate}
    \item 
      Subcase:
      $$
      \inferrule[]
      {\steps{\texpr_1}{\texpr_1'}}
      {\steps{\eletin{\evar}{\texpr_1}{\texpr_1}}
             {\eletin{\evar}{\texpr_1'}{\texpr_2}}
      }
      $$
      By inversion:
      \begin{align}
        \steps{\texpr_1}{\texpr_1'}                                         \label{eq:consistency:4}
      \end{align}
      By i.h. on \eqref{eq:consistency:1} and \eqref{eq:consistency:4} there exists $\expr_1'$ \st:
      \begin{align}
        \stepsmany{\expr_1}{\expr_1'}                         \notag \\
        \trans{\emp}{\expr_1'}{\type_1}{}{\texpr_1'}          \label{eq:consistency:3}
      \end{align}
      Applying rule \tchkletin on \eqref{eq:consistency:2} and \eqref{eq:consistency:3}:
      $$
      \trans{\emp}{\eletin{\evar}{\expr_1'}{\expr_2}}
                     {\type_2}
                     {}
                     {\eletin{\evar}{\texpr_1'}{\texpr_2}}
      $$

    \item Subcase:
      $$
      \inferrule[]{}
      {\steps{\eletin{\evar}{\tval_1}{\texpr_2}}
             {\appsubst{\esubst{\tval_1}{\evar}}{\texpr_2}}
      }
      $$
      Equation \eqref{eq:consistency:1} becomes:
      $$
      \trans{\emp}{\expr_1}{\type_1}{}{\tval_1} 
      $$
      By Lemma~\ref{lemma:valmonot:app} there exists $\val_1$ such that:
      \begin{align}
        \stepsmany{\expr_1}{\val_1}                         \label{eq:consistency:5} \\
        \trans{\emp}{\val_1}{\type_1}{}{\tval_1}            \label{eq:consistency:6}
      \end{align}
      By Lemma~\ref{lemma:src:evalctx} using~\eqref{eq:consistency:5} on 
      $\evalctx \equiv \eletin{\evar}{\empevalctx}{\expr_2}$ :
      \begin{align}
        \stepsmany{\eletin{\evar}{\expr_1}{\expr_2}}{\eletin{\evar}{\val_1}{\expr_2}} \notag
      \end{align}

      By Lemma~\ref{lemma:subst} 
      on \eqref{eq:consistency:2} and \eqref{eq:consistency:6}, there exists 
      $\expr' \equiv \appsubst{\esubst{\val_1}{\evar}}{\expr}$ \st:
      \begin{align}
        \steps{\eletin{\evar}{\val_1}{\expr_2}}
              {\appsubst{\esubst{\val_1}{\evar}}{\expr_2}} \notag \\
        \trans{\emp}
              {\appsubst{\esubst{\val_1}{\evar}}{\expr_2}}
              {\type_2}{}
              {\appsubst{\esubst{\tval_1}{\evar}}{\texpr_2}}    \notag
      \end{align}

  \end{enumerate}

\item Case \tchkite: 
  $$
  \inferrule[]
  {\trans{\emp}{\expr_c}{\tbool}{}{\texpr} \\ 
  \foralli {i \in \{1,2\}} {\trans{\emp}{\expr_i}{\type}{}{\texpr_i}}}
  {\trans{\emp}{\ite{\expr_c}{\expr_1}{\expr_2}}{\type}{}{\ite{\texpr_c}{\texpr_1}{\texpr_2}}}
  $$
  By inversion:
  \begin{align}
    \trans{\emp}{\expr_c}{\tbool}{}{\texpr_c}  \label{eq:consistency:7} \\
    \trans{\emp}{\expr_1}{\type}{}{\texpr_1}  \label{eq:consistency:8} \\
    \trans{\emp}{\expr_2}{\type}{}{\texpr_2}  \label{eq:consistency:9}
  \end{align}
  Cases on the form of $\steps{\texpr}{\texpr'}$:
  \begin{enumerate}
    \item Subcase:
      $$
      \inferrule[]{\steps{\texpr_c}{\texpr_c'}}
                  {\steps{\ite{\texpr_c}{\texpr_1}{\texpr_2}}
                         {\ite{\texpr_c'}{\texpr_1}{\texpr_2}}}
      $$ 
      By inversion:
      \begin{align}
        \steps{\texpr_c}{\texpr_c'}               \label{eq:consistency:10}
      \end{align}
      By i.h. using \eqref{eq:consistency:7} and \eqref{eq:consistency:10} there exists $\expr_c'$ such that 
      \begin{align}
        \stepsmany{\expr_c}{\expr_c'}     \notag \\
        \trans{\emp}{\expr_c'}{\tbool}{}{\texpr_c'} \label{eq:consistency:11}
      \end{align}
      Applying rule \tchkite on 
      \eqref{eq:consistency:11},
      \eqref{eq:consistency:8}  and
      \eqref{eq:consistency:9}  
      we get:
      \begin{align*}
        \trans{\emp}{\ite{\expr_c'}{\expr_1}{\expr_2}}{\type}{}{\ite{\texpr_c'}{\texpr_1}{\texpr_2}}
      \end{align*}

    \item Subcase: 
      $$
      \steps{\ite{\vtrue}{\texpr_1}{\texpr_2}}{\texpr_1}
      $$
      Equation~\ref{eq:consistency:7} becomes:
      $$
      \trans{\emp}{\expr_c}{\tbool}{}{\vtrue} 
      $$
      By Lemma~\ref{lemma:valmonot:app} there exists $\val_c$ such that:
      \begin{align}
        \stepsmany{\expr_c}{\val_c}                         \label{eq:consistency:12} \\
        \trans{\emp}{\val_c}{\tbool}{}{\vtrue}             \label{eq:consistency:113}
      \end{align}
      The only possible case for \eqref{eq:consistency:113} to hold is:
      $$
      \val_c \equiv \vtrue
      $$
      By applying Lemma~\ref{lemma:src:evalctx} using~\eqref{eq:consistency:12} on 
      $\evalctx \equiv \ite{\empevalctx}{\expr_1}{\expr_2}$:
      \begin{align}
        \stepsmany{\ite{\expr_c}{\expr_1}{\expr_2}}
        {\ite{\vtrue}{\expr_1}{\expr_2}}
        \notag
      \end{align}

      By \opsrccondtrue:
      $$\steps{\ite{\vtrue}{\expr_1}{\expr_2}}{\expr_1}$$ 
      So there exists $\expr' \equiv \expr_1$, such that $\stepsmany{\expr}{\expr'}$ and 
      by \eqref{eq:consistency:8} it holds that: 
      $$\trans{\emp}{\expr'}{\type}{}{\texpr_1}$$

    \item Subcase: 
      $$
      \steps{\ite{\vfalse}{\texpr_1}{\texpr_1}}{\texpr_2}
      $$
      This case is similar to the previous one.
  \end{enumerate}

\item Case \tchkapp: 
  \textit{Similar to proof given by Dunfield~\cite{Dunfield2012} in proof of Theorem 13.}

\item Case \tchkup:
  $$
  \inferrule[]
  {
    \trans{\emp}{\expr}{\type_k}{}{\texpr_0} \\
    \tcwf{\tor{\type_1}{\type_2}}
  } 
  {\trans{\emp}{\expr}{\tor{\type_1}{\type_2}}{}{\inj{k}{\texpr_0}}}
  $$
  By inversion:
  \begin{align}
    \trans{\emp}{\expr}{\type_k}{}{\texpr_0} \label{eq:consistency:100} \\
    \tcwf{\tor{\type_1}{\type_2}}   \label{eq:consistency:104}
  \end{align}
  The only possible case for $\steps{\texpr}{\texpr'}$ is:
  $$
  \inferrule[]
  {\steps{\texpr_0}{\texpr_0'}}
  {\steps{\inj{k}{\texpr_0}}{\inj{k}{\texpr_0'}}}
  $$
  By inversion:
  \begin{align}
    \steps{\texpr_0}{\texpr_0'}    \label{eq:consistency:101}
  \end{align}
  By i.h. using \eqref{eq:consistency:100} and \eqref{eq:consistency:101} there
  exists an $\expr'$ such that:
  \begin{align}
    \stepsmany{\expr}{\expr'}    \label{eq:consistency:102} \\
    \trans{\emp}{\expr'}{\type_k}{}{\texpr_0'} \label{eq:consistency:103}
  \end{align}
  By \tchkup on \eqref{eq:consistency:103} and \eqref{eq:consistency:104}:
  $$
  \trans{\emp}{\expr'}{\tor{\type_1}{\type_2}}{}{\inj{k}{\texpr_0'}} 
  $$

\item Case \tchkdown:
  $$
  \inferrule[]
  {
    \phantom{\trans{\emp}{\expr_0}{\tor{\type_1}{\type_2}}{}{\texpr_0}}
    \\
    \trans{\envbinding{\evar_1}{\type_1}}
             {\evalctxarg{\evar_1}}
             {\typeb}
             {}
             {\texpr_1}
    \\\\
    \trans{\emp}{\expr_0}{\tor{\type_1}{\type_2}}{}{\texpr_0}
    \\
    \trans{\envbinding{\evar_2}{\type_2}}
             {\evalctxarg{\evar_2}}
             {\typeb}
             {}
             {\texpr_2}
  }
  {
    \trans{\emp}
             {\evalctxarg{\expr_0}}
             {\typeb}
             {}{\casematch{\texpr_0}{\evar_1}{\texpr_1}{\evar_2}{\texpr_2}}
  } 
  $$
  By inversion:
  \begin{align}
    \trans{\emp}{\expr_0}{\tor{\type_1}{\type_2}}{}{\texpr_0}    \label{eq:consistency:13}    \\
    \trans{\envbinding{\evar_1}{\type_1}}
             {\evalctxarg{\evar_1}}
             {\typeb}
             {}
             {\texpr_1}                                          \label{eq:consistency:14}  \\
    \trans{\envbinding{\evar_2}{\type_2}}
             {\evalctxarg{\evar_2}}
             {\typeb}
             {}
             {\texpr_2}                                          \label{eq:consistency:15}
  \end{align}

  Cases on the form of $\steps{\texpr}{\texpr'}$:
  \begin{enumerate}
    \item Subcase: 
      $$
      \inferrule[]{\steps{\texpr_0}{\texpr_0'}}{
        \casematch{\texpr_0}{\evar_1}{\texpr_1}{\evar_2}{\texpr_2}
        \longrightarrow \\
        \casematch{\texpr_0'}{\evar_1}{\texpr_1}{\evar_2}{\texpr_2}
      }
      $$
      By inversion:
      \begin{align}
        \steps{\texpr_0}{\texpr_0'}   \label{eq:consistency:16}
      \end{align}
      By i.h. using \eqref{eq:consistency:13} and \eqref{eq:consistency:16} 
      there exists $\expr_0'$ such that
      \begin{align}
        \stepsmany{\expr_0}{\expr_0'} \label{eq:consistency:18} \\
        \trans{\emp}{\expr_0'}{\tor{\type_1}{\type_2}}{}{\texpr_0'}  \label{eq:consistency:17}
      \end{align}
      Applying \tchkdown on \eqref{eq:consistency:17}, 
      \eqref{eq:consistency:14} and \eqref{eq:consistency:15}:
      $$
        \trans{\emp}{\evalctxarg{\expr_0'}}{\type_1}{}
        {\casematch{\texpr_0'}{\evar_1}{\texpr_1}{\evar_2}{\texpr_2}}
      $$
      By applying Lemma~\ref{lemma:tgt:evalctx} using~\ref{eq:consistency:18}:
      $$
      \stepsmany{\evalctxarg{\expr_0}}{\evalctxarg{\expr_0'}}
      $$

    \item Subcase:
      $$
      \steps{\casematch{\inj{1}{\tval}}{\evar_1}{\texpr_1}{\evar_2}{\texpr_2}}
      {\appsubst{\esubst{\tval}{\evar_1}}{\texpr_1}}
      $$
      Equation \eqref{eq:consistency:13} becomes:
      \begin{align}
        \trans{\emp}{\expr_0}{\tor{\type_1}{\type_2}}{}{\inj{1}{\tval}} \label{eq:consistency:19}
      \end{align}
      Applying Lemma~\ref{lemma:valmonot:app} on \eqref{eq:consistency:19}, there exists $\val_0$ such that:
      \begin{align}
        \stepsmany{\expr_0}{\val_0} \label{eq:consistency:20} \\
        \trans{\emp}{\val_0}{\tor{\type_1}{\type_2}}{}{\inj{1}{\tval}} \label{eq:consistency:21}
      \end{align}
      By applying Lemma~\ref{lemma:unions} on \eqref{eq:consistency:21}:
      \begin{align}
        \trans{\emp}{\val_0}{\type_1}{}{\tval} \label{eq:consistency:22}
      \end{align}
      Applying Lemma~\ref{lemma:subst} on \eqref{eq:consistency:14} and \eqref{eq:consistency:22}:
      $$
      \trans{\emp}
               {\appsubst{\esubst{\val_0}{\evar_1}}{\evalctxarg{\evar_1}}}
               {\type_1}
               {}
               {\appsubst{\esubst{\tval}{\evar_1}}{\texpr_1}}
      $$
      Or, after the substitutions\footnote{Variable $\evar_1$ is only referenced in the ``hole"
      of the evaluation context $\evalctxarg{\evar_1}$.}:
      \begin{align}
      \trans{\emp}
               {\evalctxarg{\val_0}}
               {\type_1}
               {}
               {\appsubst{\esubst{\tval}{\evar_1}}{\texpr_1}} \label{eq:consistency:23}
      \end{align}
      Applying Lemma~\ref{lemma:src:evalctx} on \eqref{eq:consistency:20}:
      \begin{align}
        \stepsmany{\evalctxarg{\expr_0}}{\evalctxarg{\val_0}} \notag
      \end{align}

%
%

    \item Subcase:
      $$
      \steps{\casematch{\inj{2}{\tval}}{\evar_1}{\texpr_1}{\evar_2}{\texpr_2}}
            {\appsubst{\esubst{\tval}{\evar_2}}{\texpr_2}}
      $$
      is similar to the previous one.
%
%
  \end{enumerate}

\item Case \tchkdead:
  $$
  \inferrule[]
  {\trans{\emp}{\expr}{\type}{}{\texpr} \\
   \typetag{\type}\cap \typetag{\typeb} = \emptyset}
   {\trans{\emp}{\expr}{\typeb}{}{\deadcast{\type}{\typeb}{\texpr}}}
  $$
  By inversion: 
  \begin{align}
     \trans{\emp}{\expr}{\type}{}{\texpr} \label{eq:consistency:24} \\
     \typetag{\type}\cap \typetag{\typeb} = \emptyset    \label{eq:consistency:120}
  \end{align}
  The only possible step here is:
  $$
  \inferrule[]
  {\steps{\texpr}{\texpr'}}
  {\steps{\deadcast{\type}{\typeb}{\texpr}}
         {\deadcast{\type}{\typeb}{\texpr'}}}
  $$
  By inversion:
  \begin{align}
    \steps{\texpr}{\texpr'} \label{eq:consistency:110}
  \end{align}
  By i.h. using \eqref{eq:consistency:24} and \eqref{eq:consistency:110} there
  exists $\expr'$ such that:
  \begin{align}
    \stepsmany{\expr}{\expr'} \notag  \\
    \trans{\emp}{\expr'}{\type}{}{\texpr'} \label{eq:consistency:111}
  \end{align}
  By applying \tchkdead on \eqref{eq:consistency:111} and  \eqref{eq:consistency:120}:
  $$
  \trans{\emp}{\expr'}{\typeb}{}{\deadcast{\type}{\typeb}{\texpr'}}
  $$

\end{enumerate}

\end{proof}

}
{
\begin{proof}
  The proof of this theorem is by induction on the derivation
  $\trans{\emp}{\expr}{\type}{}{\texpr}$, adapting the proof scheme
  given by Dunfield~\cite{Dunfield2012}.
  It reuses adapted versions of the lemmas used by Dunfield and
  Lemma~\ref{lemma:valmonot}. 
  Full details of the proof can be found in the accompanying report.
  \pv{How is this going to be cited?}
\end{proof}
}

\begin{theorem}[Reverse Consistency]  \label{theorem:rev:consistency}
If 
$\trans{\emp}{\expr}{\type}{}{\texpr}$ 
and 
$\steps{\expr}{}{\expr'}$ 
then
there exists $\texpr'$ such that 
  $\trans{\emp}{\expr'}{\type}{}{\texpr'}$, 
  and $\stepsmanyone{\texpr}{\texpr'}$.
%
%
\end{theorem}

\IfAppendix{
\begin{proof}

By induction on the derivation $\trans{\emp}{\expr}{\type}{}{\texpr}$

\begin{enumerate}

\item Cases \tchkconst, \tchkvar, \tchkinterintro, \tchklambda: 

  The respective source expression does not step.

\item Case \tchkletin: 
  \begin{align}
  \inferrule[]
  {
    \trans{\emp}{\expr_1}{\type_1}{}{\texpr_1} \\
    \trans{\envbinding{\evar}{\type_1}}{\expr_2}{\type_2}{}{\texpr_2}
  }
  {\trans{\emp}{\eletin{\evar}{\expr_1}{\expr_2}}{\type_2}{}{\eletin{\evar}{\texpr_1}{\texpr_2}}}
  \label{eq:rev:consistency:87}
  \end{align}
  By inversion:
  \begin{align}
    \trans{\emp}{\expr_1}{\type_1}{}{\texpr_1}     \label{eq:rev:consistency:1} \\
    \trans{\envbinding{\evar}{\type_1}}{\expr_2}{\type_2}{}{\texpr_2}     \label{eq:rev:consistency:2}
  \end{align}

  Cases on the form of $\steps{\expr}{\expr'}$:
  \begin{enumerate}
    \item 
      Subcase:
      $$
      \inferrule[]
      {\steps{\expr_1}{\expr_1'}}
      {\steps{\eletin{\evar}{\expr_1}{\expr_2}}
             {\eletin{\evar}{\expr_1'}{\expr_2}}
      } 
      $$
      By inversion:
      \begin{align}
        \steps{\expr_1}{\expr_1'}                             \label{eq:rev:consistency:3}
      \end{align}
      By i.h. using \eqref{eq:rev:consistency:1} and
      \eqref{eq:rev:consistency:3}: 
      There exists $\texpr_1'$ such that:
      \begin{align}
        \trans{\emp}{\expr_1'}{\type_1}{}{\texpr_1'}          \label{eq:rev:consistency:4}  \\
        \stepsmanyone{\texpr_1}{\texpr_1'}                       \label{eq:rev:consistency:5}
      \end{align}
      Applying rule \tchkletin on \eqref{eq:rev:consistency:2} and \eqref{eq:rev:consistency:4}:
      $$
      \trans{\emp}{\eletin{\evar}{\expr_1'}{\expr_2}}
                     {\type_2}
                     {}
                     {\eletin{\evar}{\texpr_1'}{\texpr_2}}
      $$
      By Lemma~\ref{lemma:tgt:evalctx} on \eqref{eq:rev:consistency:5} we get:
      $$
      \stepsmanyone{\texpr}{\texpr'}
      $$


    \item Subcase:
      $$
      \inferrule[]{}
      {\steps{\eletin{\evar}{\val_1}{\expr_2}}
             {\appsubst{\esubst{\val_1}{\evar}}{\expr_2}}
      }
      $$
      Equation \eqref{eq:rev:consistency:1} becomes:
      \begin{align}
        \trans{\emp}{\val_1}{\type_1}{}{\texpr_1}
        \label{eq:rev:consistency:120}
      \end{align}
      By Lemma~\ref{lemma:rev:valmonot:app} on  \eqref{eq:rev:consistency:120}
      there exists $\tval_1$ such that:
      \begin{align}
        \stepsmany{\texpr_1}{\tval_1}                         \label{eq:rev:consistency:6} \\
        \trans{\emp}{\val_1}{\type_1}{}{\tval_1}              \label{eq:rev:consistency:7}
      \end{align}
      By Lemma~\ref{lemma:tgt:evalctx} using~\eqref{eq:rev:consistency:6} on 
      $\tevalctx \equiv \eletin{\evar}{\empevalctx}{\texpr_2}$:
      \begin{align}
      \stepsmany{\eletin{\evar}{\texpr_1}{\texpr_2}}{\eletin{\evar}{\tval_1}{\texpr_2}}
      \label{eq:rev:consistency:123}
      \end{align}
      By \optgtlet:
      \begin{align}
      \steps{\eletin{\evar}{\tval_1}{\texpr_2}}{\appsubst{\esubst{\tval_1}{\evar}}{\texpr_2}}
      \label{eq:rev:consistency:122}
      \end{align}
      By \eqref{eq:rev:consistency:123} and \eqref{eq:rev:consistency:122}:
      $$
      \stepsmanyone{\eletin{\evar}{\texpr_1}{\texpr_2}}{\appsubst{\esubst{\tval_1}{\evar}}{\texpr_2}}
      $$
      And by Lemma~\ref{lemma:subst} 
      on \eqref{eq:rev:consistency:2} and \eqref{eq:rev:consistency:7}:
      \begin{align}
        \trans{\emp}
              {\appsubst{\esubst{\val_1}{\evar}}{\expr_2}}
              {\type_2}{}
              {\appsubst{\esubst{\tval_1}{\evar}}{\texpr_2}}    \notag
      \end{align}

  \end{enumerate}

\item Case \tchkite: 
  \begin{align}
  \inferrule[]
  {\trans{\emp}{\expr_c}{\tbool}{}{\texpr} \\ 
  \foralli {i \in \{1,2\}} {\trans{\emp}{\expr_i}{\type}{}{\texpr_i}}}
  {\trans{\emp}{\ite{\expr_c}{\expr_1}{\expr_2}}{\type}{}{\ite{\texpr_c}{\texpr_1}{\texpr_2}}}
  \label{eq:rev:consistency:88}
  \end{align}
  By inversion:
  \begin{align}
    \trans{\emp}{\expr_c}{\tbool}{}{\texpr_c}  \label{eq:rev:consistency:8} \\
    \trans{\emp}{\expr_1}{\type}{}{\texpr_1}  \label{eq:rev:consistency:9} \\
    \trans{\emp}{\expr_2}{\type}{}{\texpr_2}  \label{eq:rev:consistency:10}
  \end{align}
  Cases on the form of $\steps{\expr}{\expr'}$:
  \begin{enumerate}
    \item Subcase:
      $$
      \inferrule[]{\steps{\expr_c}{\expr_c'}}
                  {\steps{\ite{\expr_c}{\expr_1}{\expr_2}}
                         {\ite{\expr_c'}{\expr_1}{\expr_2}}}
      $$ 
      By inversion:
      \begin{align}
        \steps{\expr_c}{\expr_c'}               \label{eq:rev:consistency:11}
      \end{align}
      By i.h. using \eqref{eq:rev:consistency:8} and
      \eqref{eq:rev:consistency:11}
      there exists $\texpr_c'$ such that:
      \begin{align}
        \trans{\emp}{\expr_c'}{\tbool}{}{\texpr_c'} \label{eq:rev:consistency:12} \\
        \stepsmanyone{\texpr_c}{\texpr_c'}     \label{eq:rev:consistency:121}
      \end{align}
      By Lemma~\ref{lemma:tgt:evalctx} using \eqref{eq:rev:consistency:121}:
      $$
      \stepsmanyone
        {\ite{\texpr_c}{\texpr_1}{\texpr_2}}
        {\ite{\texpr_c'}{\texpr_1}{\texpr_2}}
      $$
      Applying rule \tchkite on 
      \eqref{eq:rev:consistency:12},
      \eqref{eq:rev:consistency:9}  and
      \eqref{eq:rev:consistency:10}  
      we get:
      $$
      \trans{\emp}{\ite{\expr_c'}{\expr_1}{\expr_2}}{\type}{}{\ite{\texpr_c'}{\texpr_1}{\texpr_2}}
      $$



    \item Subcase: 
      \begin{align}
        \steps{\ite{\vtrue}{\expr_1}{\expr_2}}{\expr_1}
        \label{eq:rev:consistency:131}
      \end{align}
      Equation~\ref{eq:rev:consistency:8} becomes:
      \begin{align}
        \trans{\emp}{\vtrue}{\tbool}{}{\texpr_c}  \label{eq:rev:consistency:13}
      \end{align}
      By Lemma~\ref{lemma:rev:valmonot:app} there exists $\tval_c$ 
      such that:
      \begin{align}
        \stepsmany{\texpr_c}{\tval_c}     \label{eq:rev:consistency:133}  \\
        \trans{\emp}{\vtrue}{\tbool}{}{\tval_c} \label{eq:rev:consistency:130}
      \end{align}
      By Lemma~\ref{lemma:tgt:evalctx} using \eqref{eq:rev:consistency:133}:
      \begin{align}
        \stepsmany{\ite{\texpr_c}{\texpr_1}{\texpr_2}}
                  {\ite{\tval_c}{\texpr_1}{\texpr_2}}              \label{eq:rev:consistency:134}
      \end{align}
      By Lemma~\ref{lemma:reduce:ite} on 
      \eqref{eq:rev:consistency:130}, 
      \eqref{eq:rev:consistency:9}, 
      \eqref{eq:rev:consistency:10} and 
      \eqref{eq:rev:consistency:131}: 
      \begin{align}
        \tval_c \equiv \vtrue 
      \end{align}
      By \optgtcondtrue:
      \begin{align}
        \steps{\ite{\vtrue}{\texpr_1}{\texpr_2}}{\texpr_1}         \label{eq:rev:consistency:136}
      \end{align}
      By \eqref{eq:rev:consistency:134} and \eqref{eq:rev:consistency:136}:
      $$
      \stepsmanyone{\ite{\texpr_c}{\texpr_1}{\texpr_2}}{\texpr_1}  
      $$      
      Combining with \eqref{eq:rev:consistency:9} we get the wanted relation.



    \item Subcase: 
      $$
      \steps{\ite{\vfalse}{\expr_1}{\expr_1}}{\expr_2}
      $$
      \textit{Similar to the previous case.}

  \end{enumerate}

\item Case \tchkinterelim: \textit{Similar to earlier cases.}

\item Case \tchkapp:
  \begin{align}
  \inferrule[]
  {
    \trans{\emp}{\expr_1}{\tfun{\type}{\typeb}}{}{\texpr_1}  \\
    \trans{\emp}{\expr_2}{\type}{}{\texpr_2}
  }
  {\trans{\emp}{\app{\expr_1}{\expr_2}}{\typeb}{}{\app{\texpr_1}{\texpr_2}}}
  \label{eq:rev:consistency:95}
  \end{align}
  By inversion:
  \begin{align}
    \trans{\emp}{\expr_1}{\tfun{\type}{\typeb}}{}{\texpr_1}   \label{eq:rev:consistency:14}  \\
    \trans{\emp}{\expr_2}{\type}{}{\texpr_2}                  \label{eq:rev:consistency:15}
  \end{align}
  Cases on the form of $\steps{\expr}{\expr'}$:
  \begin{enumerate}
    \item Subcase:
      $$
      \inferrule[]
      {\steps{\expr_1}{\expr_1'}}
      {\steps{\app{\expr_1}{\expr_2}}{\app{\expr_1'}{\expr_2}}}
      $$
      \textit{Similar to eariler cases, e.g.
        $\steps{\eletin{\evar}{\expr_1}{\expr_2}}
      {\eletin{\evar}{\expr_1'}{\expr_2}}$}

    \item Subcase:
      $$
      \inferrule[]
      {\steps{\expr_2}{\expr_2'}}
      {\steps{\app{\val_1}{\expr_2}}{\app{\val_1}{\expr_2'}}}
      $$
      By inversion:
      \begin{align}
        \steps{\expr_2}{\expr_2'} \label{eq:rev:consistency:90}
      \end{align}
      By Lemma~\ref{lemma:rev:valmonot:app} on \eqref{eq:rev:consistency:14} there
      exists $\tval_1$ such that:
      \begin{align}
        \stepsmany{\texpr_1}{\tval_1}                         \label{eq:rev:consistency:96} \\
        \trans{\emp}{\val_1} {\type_2}{} {\tval_1}             \label{eq:rev:consistency:97}
      \end{align}
      By Lemma~\ref{lemma:tgt:evalctx} using \eqref{eq:rev:consistency:96}:
      \begin{align}
        \stepsmany{\app{\texpr_1}{\texpr_2}}
        {\app{\tval_1}{\texpr_2}} \label{eq:rev:consistency:99}
      \end{align}

      By i.h. using \eqref{eq:rev:consistency:15} and
      \eqref{eq:rev:consistency:90}
      there exists $\texpr_2'$ such that:
          \begin{align}
            \stepsmanyone{\texpr_2}{\texpr_2'}                 \label{eq:rev:consistency:91} \\
            \trans{\emp}{\expr_2'}{\type}{}{\texpr_2'}     \label{eq:rev:consistency:92}
          \end{align}
          By Lemma~\ref{lemma:tgt:evalctx} using \eqref{eq:rev:consistency:91} on the target of 
          \eqref{eq:rev:consistency:97}:
          $$
          \stepsmanyone{\app{\tval_1}{\texpr_2}}
                    {\app{\tval_1}{\texpr_2'}}
          $$
          And combining with \eqref{eq:rev:consistency:99}:
          $$
          \stepsmanyone{\app{\texpr_1}{\texpr_2}}
                          {\app{\tval_1}{\texpr_2'}}
          $$
          By rule \tchkapp using \eqref{eq:rev:consistency:14} and \eqref{eq:rev:consistency:92}:
          $$
          \trans{\emp}
          {\app{\val_1}{\expr_2'}}{\type}{}
          {\app{\tval_1}{\texpr_2'}}
          $$


    \item Subcase:
      \begin{align}
        \steps{\app{\parens{\elambda{\evar}{\expr_0}{}}}{\val_2}}
              {\appsubst{\esubst{\val_2}{\evar}}{\expr_0}}
              \label{eq:rev:consistency:140}
      \end{align}

      By Lemma~\ref{lemma:rev:valmonot:app} on  \eqref{eq:rev:consistency:14}
      there exists $\tval_1$ such that:
      \begin{align}
        \trans{\emp}{\elambda{\evar}{\expr_0}{}} 
                    {\tfun{\type}{\typeb}}{}
                    {\tval_1}           \label{eq:rev:consistency:80} \\ 
        \stepsmany{\texpr_1}{\tval_1}                                               \label{eq:rev:consistency:82} 
      \end{align}
      By applying Lemma~\ref{lemma:tgt:evalctx} on 
      $\texpr \equiv \app{\texpr_1}{\texpr_2}$ 
      given \eqref{eq:rev:consistency:19}:
      \begin{align}
        \stepsmany{\app{\texpr_1}{\texpr_2}} 
                  {\app{\tval_1}{\texpr_2}} \label{eq:rev:consistency:110}
      \end{align}
      Equation~\eqref{eq:rev:consistency:15} is:
      \begin{align}
        \trans{\emp}{\val_2}{\type}{}{\texpr_2}                  \label{eq:rev:consistency:17}
      \end{align}
      By Lemma~\ref{lemma:rev:valmonot:app}      on \eqref{eq:rev:consistency:17}, 
      there exists $\tval_2$ such that:
      \begin{align}  
        \stepsmany{\texpr_2}{\tval_2}              \label{eq:rev:consistency:19}  \\
        \trans{\emp}{\val_2}{\type}{}{\tval_2}                 \label{eq:rev:consistency:18}
      \end{align}

      By Lemma~\ref{lemma:beta:reduction}  on 
      \eqref{eq:rev:consistency:80}, 
      \eqref{eq:rev:consistency:18} and
      \eqref{eq:rev:consistency:140}, 
      there is a $\texpr_0$ such that:
      $$
      \tval_1 \equiv \elambda{\evar}{\texpr_0}{}
      $$
      So \eqref{eq:rev:consistency:14} becomes:
      \begin{align}
        \trans{\emp}{\elambda{\evar}{\expr_0}{}}{\tfun{\type}{\typeb}}{}{\elambda{\evar}{\texpr_0}{}} \label{eq:rev:consistency:81}
      \end{align}
      The only production of \eqref{eq:rev:consistency:81} is by \tchklambda:
      $$
      \inferrule[]
      {
        \tcwf{\tfun{\type}{\typeb}} \\
        \trans{\envbinding{\evar}{\type}}{\expr_0}{\typeb}{}{\texpr_0}
      }
      {\trans{\tcenv}
        {\elambda{\evar}{\expr_0}{}}
        {\tfun{\type}{\typeb}}{}
        {\elambda{\evar}{\texpr_0}{}}
      }
      $$
      By inversion:
      \begin{align}
        \trans{\envbinding{\evar}{\type}}{\expr_0}{\typeb}{}{\texpr_0}  \label{eq:rev:consistency:16}
      \end{align}

      By applying Lemma~\ref{lemma:subst} on \eqref{eq:rev:consistency:16} and 
      \eqref{eq:rev:consistency:18} we get:
      \begin{align}
        \trans{\emp}{\appsubst{\esubst{\val_2}{\evar}}{\expr_0}}{\typeb}{}
              {\appsubst{\esubst{\tval_2}{\evar}}{\texpr_0}}
              \label{eq:rev:consistency:111}
      \end{align}
      By applying Lemma~\ref{lemma:tgt:evalctx} on $\texpr \equiv 
      \app{\parens{\elambda{\evar}{\texpr_0}{}}}{\texpr_2}$ given \eqref{eq:rev:consistency:19}:
      \begin{align}
        \stepsmany{\app{\parens{\elambda{\evar}{\texpr_0}{}}}{\texpr_2}} 
                  {\app{\parens{\elambda{\evar}{\texpr_0}{}}}{\tval_2}} 
        \label{eq:rev:consistency:20}
      \end{align}
      By rule \optgtbeta: 
      \begin{align}
      \steps
      {\app{\parens{\elambda{\evar}{\texpr_0}{}}}{\tval_2}} 
      {\appsubst{\esubst{\tval_2}{\evar}}{\texpr_0}}
        \label{eq:rev:consistency:22}
      \end{align}
      By \eqref{eq:rev:consistency:110}, \eqref{eq:rev:consistency:20} and \eqref{eq:rev:consistency:22} we get:
      \begin{align}
      \stepsmanyone
      {\app{\texpr_1}{\texpr_2}} 
      {\appsubst{\esubst{\tval_2}{\evar}}{\texpr_0}}
      \label{eq:rev:consistency:112}
      \end{align}
      By \eqref{eq:rev:consistency:111} and 
       \eqref{eq:rev:consistency:112} we get the wanted relation.

    \item {Subcase:}
      \begin{align}
        \steps{\app{\vconst}{\val}}{\primapp{\vconst}{\val}}
        \label{eq:rev:consistency:116}
      \end{align}
      Equations \eqref{eq:rev:consistency:14} and \eqref{eq:rev:consistency:15} become:
      \begin{align}
        \trans{\emp}{\vconst}{\tfun{\type}{\typeb}}{}{\texpr_1}   \label{eq:rev:consistency:50}  \\
        \trans{\emp}{\val}{\type}{}{\texpr_2}                  \label{eq:rev:consistency:51}
      \end{align}

      By Lemma~\ref{lemma:rev:valmonot:app} on  \eqref{eq:rev:consistency:50}
      there exists $\tval_1$ such that:
      \begin{align}
        \trans{\emp}{\vconst} 
        {\tfun{\type}{\typeb}}{}
              {\tval_1}                     \label{eq:rev:consistency:113} \\ 
        \stepsmany{\texpr_1}{\tval_1}       \label{eq:rev:consistency:82} 
      \end{align}
      By Lemma~\ref{lemma:tgt:evalctx} on 
      $\texpr \equiv \app{\texpr_1}{\texpr_2}$
      given \eqref{eq:rev:consistency:82}:
      \begin{align}
        \stepsmany{\app{\texpr_1}{\texpr_2}} 
                  {\app{\tval_1}{\texpr_2}} \label{eq:rev:consistency:114}
      \end{align}

      By Lemma~\ref{lemma:rev:valmonot:app} on \eqref{eq:rev:consistency:51} there
      exists $\tval_2$ such that:
      \begin{align}
        \stepsmany{\texpr_2}{\tval_2}                         \label{eq:rev:consistency:54} \\
        \trans{\emp}{\val}{\type}{}{\tval_2}                 \label{eq:rev:consistency:55} 
      \end{align}
      By Lemma~\ref{lemma:tgt:evalctx} on \eqref{eq:rev:consistency:54}:
      \begin{align}
        \stepsmany{\app{\vconst}{\texpr_2}}{\app{\vconst}{\tval_2}}
        \label{eq:rev:consistency:115}
      \end{align}

      By Lemma~\ref{lemma:primitive:reduction} on 
      \eqref{eq:rev:consistency:113},
      \eqref{eq:rev:consistency:55} and 
      \eqref{eq:rev:consistency:116}:
      \begin{align}
        \tval_1 \equiv \vconst  \label{eq:rev:consistency:142}  \\
        \tval_2 \not\equiv \deadcast{\emp}{\type}{\emp}      \label{eq:rev:consistency:143}
      \end{align}
      So \eqref{eq:rev:consistency:113} becomes:
      \begin{align}
        \trans{\emp}{\vconst} 
              {\tfun{\type}{\typeb}}{}
              {\vconst}                     \label{eq:rev:consistency:114}
      \end{align}
      So we can apply \optgtappa:
      \begin{align}
        \steps{\app{\vconst}{\tval_2}}{\primapp{\vconst}{\tval_2}}
        \label{eq:rev:consistency:146}
      \end{align}
      By \eqref{eq:rev:consistency:114}, \eqref{eq:rev:consistency:115} and
      \eqref{eq:rev:consistency:146}:
      $$
      \stepsmanyone{\app{\texpr_1}{\texpr_2}}{\primapp{\vconst}{\tval_2}}
      $$
      By assumption~\ref{assum:const:app:app} using
      \eqref{eq:rev:consistency:114}, \eqref{eq:rev:consistency:143} and \eqref{eq:rev:consistency:55}:
      $$  
      \trans{\emp}{\primapp{\vconst}{\val}}{\type}{}{\primapp{\vconst}{\tval_2}}
      $$

    \end{enumerate}

\item Case \tchkup:
  $$
  \inferrule[]
  {
    \trans{\emp}{\expr}{\type_k}{}{\texpr} \\
    \tcwf{\tor{\type_1}{\type_2}}
  } 
  {\trans{\emp}{\expr}{\tor{\type_1}{\type_2}}{}{\inj{k}{\texpr}}}
  $$
  By inversion:
  \begin{align}
    \trans{\emp}{\expr}{\type_k}{}{\texpr}  \label{eq:rev:consistency:70} \\
    \tcwf{\tor{\type_1}{\type_2}}           \label{eq:rev:consistency:72} 
  \end{align}
  By i.h. using \eqref{eq:rev:consistency:70} with $\steps{\expr}{\expr'}$, 
  there exists $\texpr'$ such that:
  \begin{align}
    \trans{\emp}{\expr'}{\type_k}{}{\texpr'}  \label{eq:rev:consistency:71} \\
    \stepsmanyone{\texpr}{\texpr'} \label{eq:rev:consistency:148}
  \end{align}
  By Lemma~\ref{lemma:tgt:evalctx} using \eqref{eq:rev:consistency:148}:
  $$
  \stepsmanyone{\inj{k}{\texpr}}
            {\inj{k}{\texpr'}}
  $$
  Applying \tchkup with premises \eqref{eq:rev:consistency:71} and \eqref{eq:rev:consistency:72}:
  $$
  \trans{\emp}{\expr'}{\tor{\type_1}{\type_2}}{}{\inj{k}{\texpr'}}
  $$

\item Case \tchkdown: 
  $$
  \inferrule[]
  {
    \phantom{\trans{\emp}{\expr_0}{\tor{\type_1}{\type_2}}{}{\texpr_0}}
    \\
    \trans{\envbinding{\evar_1}{\type_1}}
             {\evalctxarg{\evar_1}}
             {\typeb}
             {}
             {\texpr_1}
    \\\\
    \trans{\emp}{\expr_0}{\tor{\type_1}{\type_2}}{}{\texpr_0}
    \\
    \trans{\envbinding{\evar_2}{\type_2}}
             {\evalctxarg{\evar_2}}
             {\typeb}
             {}
             {\texpr_2}
  }
  {
    \trans{\emp}
             {\evalctxarg{\expr_0}}
             {\typeb}
             {}{\casematch{\texpr_0}{\evar_1}{\texpr_1}{\evar_2}{\texpr_2}}
  } 
  $$
  By inversion:
  \begin{align}
    \trans{\emp}{\expr_0}{\tor{\type_1}{\type_2}}{}{\texpr_0}    \label{eq:rev:consistency:23}    \\
    \trans{\envbinding{\evar_1}{\type_1}}
             {\evalctxarg{\evar_1}}
             {\typeb}
             {}
             {\texpr_1}                                          \label{eq:rev:consistency:24}  \\
    \trans{\envbinding{\evar_2}{\type_2}}
             {\evalctxarg{\evar_2}}
             {\typeb}
             {}
             {\texpr_2}                                          \label{eq:rev:consistency:25}
  \end{align}
  Cases on the form of $\steps{\expr}{\expr'}$:
  \begin{enumerate}
    \item Subcase: 
      $$
      \inferrule[]{\steps{\expr_0}{\expr_0'}}
                  {\steps{\evalctxarg{\expr_0}}{\evalctxarg{\expr_0'}}}
      $$
      By inversion:
      \begin{align}
        \steps{\expr_0}{\expr_0'}   \label{eq:rev:consistency:26}
      \end{align}
      By i.h. using \eqref{eq:rev:consistency:23} and \eqref{eq:rev:consistency:26}:
      \begin{align}
        \trans{\emp}{\expr_0'}{\tor{\type_1}{\type_2}}{}{\texpr_0'}    \label{eq:rev:consistency:27}    \\
        \stepsmanyone{\texpr_0}{\texpr_0'}   \label{eq:rev:consistency:28}
      \end{align}
      Using rule~\tchkdown on \eqref{eq:rev:consistency:27}, \eqref{eq:rev:consistency:24} and 
      \eqref{eq:rev:consistency:25}:
      $$
      \trans{\emp}
            {\evalctxarg{\expr_0'}}
            {\typeb}
            {}{\casematch{\texpr_0'}{\evar_1}{\texpr_1}{\evar_2}{\texpr_2}}
      $$
      Also, applying Lemma~\ref{lemma:tgt:evalctx} on $\tevalctx \equiv 
            \casematch{\empevalctx}{\evar_1}{\texpr_1}{\evar_2}{\texpr_2}$ using
            \eqref{eq:rev:consistency:28}:
      \begin{align*}
        \casematch{\texpr_0}{\evar_1}{\texpr_1}{\evar_2}{\texpr_2}
        \longrightarrow^{+} \\
        \casematch{\texpr_0'}{\evar_1}{\texpr_1}{\evar_2}{\texpr_2}
      \end{align*}

    \item Subcase:
      \begin{align}
      \expr_0 \equiv \val_0 \\
        \steps{\evalctxarg{\val_0}}{\expr'}
      \end{align}
      Because $\val_0$ is a value, we can split cases for its type. 
      Without loss of generality we can assume that its type is $\type_1$ (the same exact 
      holds for $\type_2$). This is depicted on the form of $\texpr_0$ in 
      equation \eqref{eq:rev:consistency:23}, which now becomes
      (for some $\texpr_{01}$):
      \begin{align}
        \trans{\emp}{\val_0}{\tor{\type_1}{\type_2}}{}{\inj{1}{\texpr_{01}}}  
        \label{eq:rev:consistency:41}
      \end{align}
      By Lemma~\ref{lemma:unions} on \eqref{eq:rev:consistency:41}:
      \begin{align}
        \trans{\emp}{\val_0}{\type_1}{}{\texpr_{01}}  \label{eq:rev:consistency:29}
      \end{align}
      By Lemma~\ref{lemma:rev:valmonot:app} on \eqref{eq:rev:consistency:29} there exists $\tval_{01}$ such that:
      \begin{align}
        \stepsmany{\texpr_{01}}{\tval_{01}}              \label{eq:rev:consistency:30}    \\
        \trans{\emp}{\val_0}{\type_1}{}{\tval_{01}}  \label{eq:rev:consistency:31}
      \end{align}
      By Lemma~\ref{lemma:rev:valmonot:app} on \eqref{eq:rev:consistency:41}
      there exists $\tval_{01}$\footnote{This is the same that we got right
      before, due to uniqueness of normal forms.} such that:
      \begin{align}
        \stepsmany{\inj{1}{\texpr_{01}}}{\inj{1}{\tval_{01}}}   \label{eq:rev:consistency:43}  \\
        \trans{\emp}{\val_0}{\tor{\type_1}{\type_2}}{}{\inj{1}{\tval_{01}}}  \label{eq:rev:consistency:44}
      \end{align}


      Cases for the form of $\evalctxarg{\val_0}$:
      \begin{enumerate}

        \item $\evalctxarg{\val_0} \equiv \eletin{\evar}{\val_0}{\expr_1}$

          The original elaboration judgment becomes:
          $$
          \inferrule[]
          {
            \phantom{\trans{\emp}{\val_0}{\tor{\type_1}{\type_2}}{}{\texpr_0}} \\
            \trans{\envbinding{\evar_1}{\type_1}}
                  {\eletin{\evar}{\evar_1}{\expr_1}}
                  {\typeb}
                  {}
                  {\texpr_1}
            \\\\
            \trans{\emp}{\val_0}{\tor{\type_1}{\type_2}}{}{\texpr_0}
            \\
            \trans{\envbinding{\evar_2}{\type_2}}
                  {\eletin{\evar}{\evar_2}{\expr_1}}
                  {\typeb}
                  {}
                  {\texpr_2}
          }
          {
            \trans{\emp}
                  {\eletin{\evar}{\val_0}{\expr_1}}
                  {\typeb}
                  {}{\casematch{\texpr_0}{\evar_1}{\texpr_1}{\evar_2}{\texpr_2}}
          } 
          $$
          By inversion:
          \begin{align}
            \trans{\emp}{\val_0}{\tor{\type_1}{\type_2}}{}{\texpr_0}  \notag \\ 
            \trans{\envbinding{\evar_1}{\type_1}}
                  {\eletin{\evar}{\evar_1}{\expr_1}}                    
                  {\typeb}
                  {}
                  {\texpr_1}    \label{eq:rev:consistency:33}             \\
            \trans{\envbinding{\evar_2}{\type_2}}
                  {\eletin{\evar}{\evar_2}{\expr_1}}
                  {\typeb}
                  {}
                  {\texpr_2}    \notag 
          \end{align}
          The derivation for \eqref{eq:rev:consistency:33} is of the form:
          $$
          \inferrule
          {
            \trans{\envbinding{\evar_1}{\type_1}}{\evar_1}{\type_1}{}{\evar_1} \\
            \trans{\envbinding{\evar_1}{\type_1}, \envbinding{\evar}{\type_1}}
                  {\expr_1}{\typeb}{}{\texprb_1}
          }
          {\trans{\envbinding{\evar_1}{\type_1}}
                 {\eletin{\evar}{\evar_1}{\expr_1}}
                 {\typeb}
               {}{\underbrace{\eletin{\evar}{\evar_1}{\texprb_1}}_{\texpr_1}}}
          $$
          By inversion:
          \begin{align}
            \trans{\envbinding{\evar_1}{\type_1}, \envbinding{\evar}{\type_1}}
            {\expr_1}{\typeb}{}{\texprb_1} \label{eq:rev:consistency:35}
          \end{align}
          Variable $\evar_1$ does not appear in $\expr_1$, so the above is equivalent to:
          \begin{align}
            \trans{\envbinding{\evar}{\type_1}}
            {\expr_1}{\typeb}{}{\texprb_1} \label{eq:rev:consistency:36}
          \end{align}
          Applying Lemma~\ref{lemma:subst} on \eqref{eq:rev:consistency:31} and \eqref{eq:rev:consistency:36}
          we get:
          \begin{align}
            \trans{\emp}
                  {\appsubst{\esubst{\val_0}{\evar}}{\expr_1}}
                  {\typeb}{}
                  {\appsubst{\esubst{\tval_{01}}{\evar}}{\texprb_1}}
                  \label{eq:rev:consistency:37}
          \end{align}
          By \opsrclet:
          \begin{align}
            \steps
                 {\eletin{\evar}{\val_0}{\expr_1}}
                  {\appsubst{\esubst{\val_0}{\evar}}{\expr_1}}
                  \label{eq:rev:consistency:38}
          \end{align}
          Also, by Lemma~\ref{lemma:tgt:evalctx} using \eqref{eq:rev:consistency:43}:
          \begin{align*}
            \casematch{\inj{1}{\texpr_{01}}}{\evar_1}{\texpr_1}{\evar_2}{\texpr_2}
            \longrightarrow^{*} \\
          \casematch{\inj{1}{\tval_{01}}}{\evar_1}{\texpr_1}{\evar_2}{\texpr_2}
          \end{align*}
          By \optgtcase:
          \begin{align*}
            \steps{\casematch{\inj{1}{\tval_{01}}}{\evar_1}{\texpr_1}{\evar_2}{\texpr_2}}
            {\appsubst{\esubst{\tval_{01}}{\evar_1}}{\texpr_1}}
          \end{align*}
          From the equality: $\texpr_1 \equiv \eletin{\evar}{\evar_1}{\texprb_1}$, and because 
          $\evar_1$ does not appear in $\texpr_1$:
          $$
          \appsubst{\esubst{\tval_{01}}{\evar_1}}{\parens{\eletin{\evar}{\evar_1}{\texprb_1}}}
          \equiv
          \steps{\eletin{\evar}{\tval_{01}}{\texprb_1}}
          {\appsubst{\esubst{\tval_{01}}{\evar}}{\texprb_1}}
          $$
          From the last relation along with \eqref{eq:rev:consistency:38} and 
          \eqref{eq:rev:consistency:37} stems the wanted result.

        \item $\evalctxarg{\val_0} \equiv \ite{\val_0}{\expr_1}{\expr_2}$:
          \textit{Similar to case $\ite{\expr_c}{\expr_1}{\expr_2}$}

        \item $\evalctxarg{\val_0} \equiv \app{\val_0}{\expr}$:
          \textit{Similar to earlier cases.}

        \item $\evalctxarg{\val_0} \equiv \app{\parens{\elambda{\evar}{\expr_0}{}}}{\val_0}$:
          \textit{Similar to earlier cases.}

      \end{enumerate}

  \end{enumerate}

\item Case \tchkdead: 
  $$
  \inferrule[]
  {
    \trans{\emp}{\expr}{\type}{}{\texpr} \\
    \typetag{\type}\cap \typetag{\typeb} = \emptyset
  }
  {
    \trans{\emp}{\expr}{\typeb}{}{\deadcast{\type}{\typeb}{\texpr}}
  }
  $$
  By inversion:
  \begin{align}
    \trans{\emp}{\expr}{\type}{}{\texpr}  \label{eq:rev:consistency:60} \\
    \typetag{\type}\cap \typetag{\typeb} = \emptyset  \label{eq:rev:consistency:65}
  \end{align}
  There also exists $\expr'$ such that:
  \begin{align}
    \steps{\expr}{\expr'}   \label{eq:rev:consistency:61}
  \end{align}

  By i.h. on \eqref{eq:rev:consistency:60} and \eqref{eq:rev:consistency:61} there exists $\texpr'$, such that:
  \begin{align}
    \stepsmanyone{\texpr}{\texpr'}             \label{eq:rev:consistency:150} \\
    \trans{\emp}{\expr'}{\type}{}{\texpr'}  \label{eq:rev:consistency:62}
  \end{align}
  By Lemma~\ref{lemma:tgt:evalctx} on \eqref{eq:rev:consistency:150}:
  $$
  \stepsmanyone{\deadcast{\type}{\typeb}{\texpr}}
            {\deadcast{\type}{\typeb}{\texpr'}}
  $$

  Applying rule \tchkdead on \eqref{eq:rev:consistency:62} and \eqref{eq:rev:consistency:65}:
  $$
  \trans{\emp}{\expr'}{\typeb}{}{\deadcast{\type}{\typeb}{\texpr'}}
  $$
 
 \end{enumerate}

\end{proof}

}
{
\begin{proof}
  Similar to the proof of Theorem~\ref{theorem:consistency}, using
  adapted versions of the lemmas used by Dunfield~\cite{Dunfield2012} and
  Lemma~\ref{lemma:rev:valmonot}.
  Again, full details of the proof can be found in the accompanying report.
\end{proof}
}

\IfAppendix{
  \begin{theorem}[Two-Phase Safety]\label{theorem:twophase-safety}
  $\forall \; \expr, \type, \texpr, \rtype \;\;$ \st:
  \begin{enumerate}[label={(\arabic*)}]
      \inditem $\trans{\emp}{\expr}{\type}{}{\texpr}$,      \label{eq:safety:1}
      \inditem $\lqcheck{\emp}{\texpr}{\rtype}$,            \label{eq:safety:2}  
  \end{enumerate}
  Then, either $\expr$ is a value, or there exists $\expr'$ such that: 
  $\steps{\expr}{\expr'}$ and 
  $\trans{\emp}{\expr'}{\type}{}{\texpr'}$ for $\texpr'$, such that
  $\stepsmanyone{\texpr}{\texpr'}$ and $\lqcheck{\emp}{\texpr'}{\rtype}$.
  \end{theorem}
}
{

  \begin{theorem}[Two-Phase Soundness]\label{theorem:twophase-soundness}
  If $\expr$ is well two-typed then, either $\expr$ is a value, or 
  there exists $\expr'$ such that: 
  \begin{enumerate}[label={(\arabic*)}]
  \inditem \textbf{(Progress)} $\steps{\expr}{\expr'}$
  \inditem \textbf{(Preservation)} $\expr'$ is well two-typed.
  \end{enumerate}
  \end{theorem}

}

\IfAppendix{
\begin{proof}
By induction on pairs \rulename{T-\textit{Rule}}/\rulename{R-\textit{Rule}} of derivations: 
  \begin{align*}
  \trans{\tcenv}{\expr}{\type}{}{\texpr} \\
  \lqcheck{\lqenv}{\texpr}{\rtype}
  \end{align*}

\begin{enumerate}

\item Cases \tchkconst/\lqchkconst
  , \tchkvar/\lqchkvar
  , \tchkinterintro/\lqchkpair
  , \tchklambda/\lqchklambda:
  
  The term $\expr$ is a value.

\item Case \tchkletin/\lqchkletin: 

  From \ref{eq:safety:1} we have:
  $$
  \inferrule[]
  {
   \trans{\emp}{\expr_1}{\type_1}{}{\texpr_1} \\
   \trans{\envbinding{\evar}{\type_1}}{\expr_2}{\type_2}{}{\texpr_2}
  }
  {\trans{\emp}{\eletin{\evar}{\expr_1}{\expr_2}}{\type_2}{}{\eletin{\evar}{\texpr_1}{\texpr_2}}}
  $$
  By inversion: 
  \begin{align}
    \trans{\emp}{\expr_1}{\type_1}{}{\texpr_1} \label{eq:safety:4} \\
    \trans{\envbinding{\evar}{\type_1}}{\expr_2}{\type_2}{}{\texpr_2}  \label{eq:safety:5}
  \end{align}

  From \ref{eq:safety:2} we have:
  $$
  \inferrule[]
  {
    \lqcheck{\emp}{\texpr_1}{\rtype_1} \\
    \lqcheck{\envbinding{\evar}{\rtype_1}}{\texpr_2}{\ifdef{\kvars}{\ltype_2}{\rtype_2}}   \\
    \ifdef{\kvars}{\lqwf{\lqenv}{\ifdef{\kvars}{\ltype_2}{\rtype_2}}}{\relax}
  }
  {\lqcheck{\emp}{\eletin{\evar}{\texpr_1}{\texpr_2}}{\ifdef{\kvars}{\ltype_2}{\rtype_2}}}
  $$
  By inversion:
  \begin{align}
    \lqcheck{\emp}{\texpr_1}{\rtype_1}    \label{eq:safety:6} \\
    \lqcheck{\envbinding{\evar}{\rtype_1}}{\texpr_2}{\ifdef{\kvars}{\ltype_2}{\rtype_2}}  \label{eq:safety:7}
  \end{align}

  By i.h. using \eqref{eq:safety:4} and \eqref{eq:safety:6} we have two cases on the form of $\expr_1$:
  \begin{enumerate}
    \item Expression $\expr_1$ is a value: 
      $$
      \expr_1 \equiv \val_1
      $$
      By source language operational semantics:
      \begin{align}
        \expr \equiv \steps{\eletin{\evar}{\val_1}{\expr_2}}
                           {\appsubst{\esubst{\val_1}{\evar}}{\expr_2}}
                           \label{eq:safety:105}
      \end{align}

    \item There exists $\expr_1'$ such that: 
      \begin{align}
        \steps{\expr_1}{\expr_1'} \label{eq:safety:31}
      \end{align}
      Hence, by \opsrcevalctx:
      \begin{align}
        \expr \equiv \steps{\eletin{\evar}{\expr_1}{\expr_2}}
              {\eletin{\evar}{\expr_1'}{\expr_2}}  \label{eq:safety:106}
      \end{align}

  \end{enumerate}

  In either case, there exists $\expr'$ such that:
  \begin{align}
    \steps{\expr}{\expr'} \label{eq:safety:107}
  \end{align}
  By Theorem~\ref{theorem:rev:consistency} on \ref{eq:safety:1} and
  \eqref{eq:safety:107}, there exists $\texpr'$ such that:
  \begin{align}
    \stepsmanyone{\texpr}{\texpr'}   \notag \\
    \trans{\emp}{\expr}{\type_2}{}{\texpr'}
    \notag
  \end{align}
  And by Corollary~\ref{corollary:multi:preservation}:
  $$
  \lqcheck{\emp}{\texpr'}{\rtype}
  $$

\item Case \tchkite/\lqchkite:

  From \ref{eq:safety:1} we have:
  $$
  \inferrule[]
  {
    \trans{\emp}{\expr_c}{\tbool}{}{\texpr} \\ 
    \foralli {i \in \{1,2\}}
    {\trans{\emp}{\expr_i}{\type}{}{\texpr_i}}
  }
  {\trans{\emp}{\ite{\expr_c}{\expr_1}{\expr_2}}{\type}{}{\ite{\texpr_c}{\texpr_1}{\texpr_2}}}
  $$
  By inversion:
  \begin{align}
    \trans{\emp}{\expr_c}{\tbool}{}{\texpr_c}                   \label{eq:safety:8} \\
    \trans{\emp}{\expr_1}{\type}{}{\texpr_1}                  \label{eq:safety:9} \\
    \trans{\emp}{\expr_2}{\type}{}{\texpr_2}                  \notag
  \end{align}
  From \ref{eq:safety:2}:
  $$
  \inferrule[]
  {
    \lqcheck{\emp}{\texpr_c}{\tbool} \\
    \lqcheck{\texpr_c}{\texpr_1}{\ifdef{\kvars}{\ltype}{\rtype}} \\
    \lqcheck{\neg\texpr_c}{\texpr_2}{\ifdef{\kvars}{\ltype}{\rtype}} \\
    \ifdef{\kvars}{\lqwf{\emp}{\ifdef{\kvars}{\ltype}{\rtype}}}{\relax}
  }
  {\lqcheck{\emp}{\ite{\texpr_c}{\texpr_1}{\texpr_2}}{\ifdef{\kvars}{\ltype}{\rtype}}}
  $$
  By inversion:
  \begin{align}
    \lqcheck{\emp}{\texpr_c}{\tbool}                              \label{eq:safety:10} \\
    \lqcheck{\texpr_c}{\texpr_1}{\ifdef{\kvars}{\ltype}{\rtype}}  \label{eq:safety:11} \\
    \lqcheck{\neg\texpr_c}{\texpr_2}{\ifdef{\kvars}{\ltype}{\rtype}}  \label{eq:safety:35}
  \end{align}
  By i.h. using \eqref{eq:safety:8} and \eqref{eq:safety:10}  we have two case on the form of $\expr_c$:
  \begin{enumerate}

    \item Expression $\expr_c$ is a value:
      $$
      \expr_c \equiv \val_c
      $$
      By a standard canonical forms lemma $\val_c$ is either $\vtrue$ or $\vfalse$. 
      Assume the first case (the latter case is identical but involving the ``else" branch of the conditional):
      By \opsrccondtrue: 
      $$
      \steps{\ite{\vtrue}{\expr_1}{\expr_2}}{\expr_1}
      $$

    \item There exists $\expr_c'$ such that:
      \begin{align}
        \steps{\expr_c}{\expr_c'}     \label{eq:safety:33}
      \end{align}
      Hence, by \opsrcevalctx:
      $$
      \steps{\ite{\expr_c}{\expr_1}{\expr_2}}
            {\ite{\expr_c'}{\expr_1}{\expr_2}}
      $$
  \end{enumerate}
  In either case, there exists $\expr'$ such that:
  \begin{align}
    \steps{\expr}{\expr'} \label{eq:safety:110}
  \end{align}
  By Theorem~\ref{theorem:rev:consistency} on \ref{eq:safety:1} and
  \eqref{eq:safety:110}, there exists $\texpr'$ such that:
  \begin{align}
    \stepsmanyone{\texpr}{\texpr'}   \notag \\
    \trans{\emp}{\expr}{\type}{}{\texpr'}
    \notag
  \end{align}
  And by Corollary~\ref{corollary:multi:preservation}:
  $$
  \lqcheck{\emp}{\texpr'}{\rtype}
  $$

\item Case \tchkinterelim/\lqchkproj:

  Without loss of generality we're going to assume first projection 
  (the same holds for the second projection).

  From \ref{eq:safety:1}:
  $$
  \inferrule[]
  {\trans{\emp}{\expr}{\tand{\type_1}{\type_2}}{}{\texpr_0}}
  {\trans{\emp}{\expr}{\type_{1}}{}{\proj{1}{\texpr_0}}}
  $$
  By inversion:
  \begin{align}
    \trans{\emp}{\expr}{\tand{\type_1}{\type_2}}{}{\texpr_0}   \label{eq:safety:12}
  \end{align}
  From \ref{eq:safety:2}:
  $$
  \inferrule
  []
  {\lqcheck{\emp}{\texpr_0}{\tprod{\rtype_1}{\rtype_2}}}
  {\lqcheck{\emp}
           {\proj{1}{\texpr_0}}
           {\rtype_{1}}
  }
  $$
  By inversion:
  \begin{align}
    \lqcheck{\emp}{\texpr_0}{\tprod{\rtype_1}{\rtype_2}}      \label{eq:safety:13}
  \end{align}
  By i.h. using \eqref{eq:safety:12} and \eqref{eq:safety:13}   we have two case on the form of $\expr$:
  \begin{enumerate}

    \item Expression $\expr$ is a value:
      $$
      \expr \equiv \val
      $$
      So the source term does not step.

    \item There exists $\expr'$ such that:
      \begin{align}
        \steps{\expr}{\expr'} \label{eq:safety:36}
      \end{align}

      By Theorem~\ref{theorem:rev:consistency} on \ref{eq:safety:1} and
      \eqref{eq:safety:36}, there exists $\texpr'$ such that:
      \begin{align}
        \stepsmanyone{\texpr}{\texpr'}   \notag \\
        \trans{\emp}{\expr}{\type}{}{\texpr'}
        \notag
      \end{align}
      And by Corollary~\ref{corollary:multi:preservation}:
      $$
      \lqcheck{\emp}{\texpr'}{\rtype}
      $$

  \end{enumerate}

\item Case \tchkapp/\lqchkapp:

  From \ref{eq:safety:1}:
  $$
  \inferrule[]
  {
    \trans{\emp}{\expr_1}{\tfun{\type}{\typeb}}{}{\texpr_1}
    \\
    \trans{\emp}{\expr_2}{\type}{}{\texpr_2}
  }
  {\trans{\emp}
    {\app{\expr_1}{\expr_2}}{\typeb}{}{\app{\texpr_1}{\texpr_2}}
  }
  $$
  By inversion:
  \begin{align}
    \trans{\emp}{\expr_1}{\tfun{\type}{\typeb}}{}{\texpr_1}
    \label{eq:safety:14}
    \\
    \trans{\emp}{\expr_2}{\type}{}{\texpr_2}
    \label{eq:safety:15}
  \end{align}
  From \ref{eq:safety:2}:
  $$
  \inferrule[]
  {
    \lqcheck{\emp}{\texpr_1}{\tfun{\rtype_x}{\rtype}}  \\
    \lqcheck{\emp}{\texpr_2}{\rtype_x}
  }
  {\lqcheck{\emp}
         {\app{\texpr_1}{\texpr_2}}
         {\appsubst{\rsubst{\texpr_2}{\evar}}{\rtype}}
  }
  $$
  By inversion:
  \begin{align}
    \lqcheck{\emp}{\texpr_1}{\tfun{\rtype_x}{\rtype}}   \label{eq:safety:16}  \\
    \lqcheck{\emp}{\texpr_2}{\rtype_x}                  \label{eq:safety:17}
  \end{align}
  By i.h. using \eqref{eq:safety:14} and \eqref{eq:safety:16}  we have three
  cases on the form of $\expr_1$:
  \begin{enumerate}
    \item Expression $\expr_1$ is a \textit{primitive} value: 
      $$
      \expr_1 \equiv \vconst
      $$
      Elaboration \eqref{eq:safety:14} becomes:
      \begin{align}
        \trans{\emp}{\vconst}{\tfun{\type}{\typeb}}{}{\texpr_1}     \label{eq:safety:120}
      \end{align}
      By applying lemma~\ref{lemma:rev:valmonot:app} on \eqref{eq:safety:120}
      there exists $\tval_1$, such that:
      \begin{align}
        \stepsmany{\texpr_1}{\tval_1}   \label{eq:safety:1211} \\
        \trans{\emp}{\vconst}{\tfun{\type}{\typeb}}{}{\tval_1} \label{eq:safety:122}
      \end{align}

      By Corollary~\ref{corollary:multi:preservation} using \eqref{eq:safety:16}
      and \eqref{eq:safety:121}:
      \begin{align}
        \lqcheck{\emp}{\tval_1}{\tfun{\rtype_x}{\rtype}}   \label{eq:safety:127}  
      \end{align}

      By Assumption~\ref{assumption:canonical:forms} on \eqref{eq:safety:122}:
      $$
      \tval_1 \equiv \vconst
      $$
      or
      $$
      \tval_1 \equiv \deadcast{\cdot}{\type}{\tval_1'}
      $$
      The latter case combined with \eqref{eq:safety:127} contradicts
      Corollary~\ref{corollary:deadcast:invalid}, so we end up with:
      \begin{align}
        \tval_1 \equiv \vconst  \label{eq:safety:128}
      \end{align}

      By lemma~\ref{lemma:tgt:evalctx} using \eqref{eq:safety:1211}:
      \begin{align}
        \stepsmany{\app{\texpr_1}{\texpr_2}}{\app{\vconst}{\texpr_2}}
        \label{eq:safety:123} 
      \end{align}

      By i.h. using \eqref{eq:safety:15} and \eqref{eq:safety:17} we have two cases on the form of $\expr_2$:
      \begin{enumerate}
        \item   Expression $\expr_2$ is a value: 
          $$
          \expr_2 \equiv \val_2
          $$
          Elaboration \eqref{eq:safety:15} becomes:
          \begin{align}
            \trans{\emp}{\val_2}{\type}{}{\texpr_2} \label{eq:safety:40}
          \end{align}
          By lemma~\ref{lemma:rev:valmonot:app} on \eqref{eq:safety:40} there exists $\tval_2$, such that:
          \begin{align}
            \stepsmany{\texpr_2}{\tval_2}   \label{eq:safety:422} \\
            \trans{\emp}{\val_2}{\type}{}{\tval_2} \label{eq:safety:41}
          \end{align}
          By lemma~\ref{lemma:tgt:evalctx} using \eqref{eq:safety:422}:
          \begin{align}
            \stepsmany{\app{\vconst}{\texpr_2}}{\app{\vconst}{\tval_2}}   \label{eq:safety:42} 
          \end{align}

          For the sake of contradiction assume:
          \begin{align}
            \tval_2 \equiv \deadcast{\typeb}{\type}{\tval_2'}
            \label{eq:safety:129}
          \end{align}
          for some $\tval_2'$. By \eqref{eq:safety:17}:
          \begin{align}
            \lqcheck{\emp}{\tval_2}{\rtype_x}   \label{eq:safety:170}
          \end{align}
          So, by \eqref{eq:safety:170} and
          Corollary~\ref{corollary:deadcast:invalid} we have a contradiction.
          So:
          \begin{align}
            \tval_2 \not\equiv \deadcast{\typeb}{\type}{\tval_2'}
            \label{eq:safety:130}
          \end{align}
          By Assumption~\ref{assum:const:app:app} on \eqref{eq:safety:122},
          \eqref{eq:safety:128}
          \eqref{eq:safety:41} and \eqref{eq:safety:130}:
          \begin{align}
            \steps{\app{\vconst}{\val_2}}{\primapp{\vconst}{\val_2}}  
          \end{align}

        \item There exists $\expr_2'$ such that:
          \begin{align}
            \steps{\expr_2}{\expr_2'} \label{eq:safety:50}
          \end{align}
          By \opsrcevalctx:
          $$
          \steps{\app{\vconst}{\expr_2}}{\app{\vconst}{\expr_2'}}
          $$
      \end{enumerate}

    \item Expression $\expr_1$ is an abstraction:
      \begin{align}
        \expr_1 \equiv \elambda{\evar}{\expr_0}{} \label{eq:safety:51}
      \end{align}

      Elaboration \eqref{eq:safety:14} becomes:
      \begin{align}
        \trans{\emp}{\elambda{\evar}{\expr_0}{}}{\tfun{\type}{\typeb}}{}{\texpr_1}
          \label{eq:safety:140}
      \end{align}
      By applying lemma~\ref{lemma:rev:valmonot:app} on \eqref{eq:safety:120}
      there exists $\tval_1$, such that:
      \begin{align}
        \stepsmany{\texpr_1}{\tval_1}   \label{eq:safety:121} \\
        \trans{\emp}{\elambda{\evar}{\expr_0}{}}{\tfun{\type}{\typeb}}{}{\tval_1}
        \label{eq:safety:141}
      \end{align}

      By Corollary~\ref{corollary:multi:preservation} using \eqref{eq:safety:16}
      and \eqref{eq:safety:121}:
      \begin{align}
        \lqcheck{\emp}{\tval_1}{\tfun{\rtype_x}{\rtype}}   \label{eq:safety:142}  
      \end{align}
      By Assumption~\ref{assumption:canonical:forms} on \eqref{eq:safety:141}:
      $$
      \tval_1 \equiv \elambda{\evar}{\texpr_0}{}
      $$
      or
      $$
      \tval_1 \equiv \deadcast{\cdot}{\tfun{\type}{\typeb}}{\tval_1'}
      $$
      The latter case combined with \eqref{eq:safety:142} contradicts
      Corollary~\ref{corollary:deadcast:invalid}, so we end up with:
      \begin{align}
        \tval_1 \equiv \elambda{\evar}{\texpr_0}{}
        \label{eq:safety:143}
      \end{align}
      So \eqref{eq:safety:14} becomes:
      \begin{align}
        \trans{\emp}{\elambda{\evar}{\expr_0}{}}{\tfun{\type}{\typeb}}{}{\elambda{\evar}{\texpr_0}{}}
        \label{eq:safety:145}
      \end{align}

      By i.h. using \eqref{eq:safety:15} and \eqref{eq:safety:17} we have two cases on the form of $\expr_2$:
      \begin{enumerate}
        \item   Expression $\expr_2$ is a value: 
          $$
          \expr_2 \equiv \val_2
          $$
          Equation \eqref{eq:safety:15} becomes (for some $\texpr_2$):
          $$
          \trans{\emp}{\val_2}{\type}{}{\texpr_2}
          $$
          By lemma~\ref{lemma:rev:valmonot:app}, there exists $\tval_2$ such that:
          \begin{align}
            \trans{\emp}{\val_2}{\type}{}{\tval_2}  \label{eq:safety:61}
          \end{align}
        
          For the sake of contradiction assume:
          \begin{align}
            \tval_2 \equiv \deadcast{\typeb}{\type}{\tval_2'}
            \label{eq:safety:135}
          \end{align}
          for some $\tval_2'$. By \eqref{eq:safety:17}:
          \begin{align}
          \lqcheck{\emp}{\tval_2}{\rtype_x} 
            \label{eq:safety:137}
          \end{align}
          So, by \eqref{eq:safety:137} and Corollary~\ref{corollary:deadcast:invalid} we have a contradiction. So:
          \begin{align}
            \tval_2 \not\equiv \deadcast{\typeb}{\type}{\tval_2'}
            \label{eq:safety:138}
          \end{align}
          By Assumption~\ref{assum:lambda:app:app} on
          \eqref{eq:safety:145},
          \eqref{eq:safety:61} and
          \eqref{eq:safety:138}:
          \begin{align*}
          \steps{\app{\parens{\elambda{\evar}{\expr_0}{}}}{\val_2}}
                  {\appsubst{\esubst{\val_2}{\evar}}{\expr_0}}    
                \end{align*}

        \item There exists $\expr_2'$ such that:
          \begin{align}
            \steps{\expr_2}{\expr_2'} \label{eq:safety:65}
          \end{align}
          By \opsrcevalctx:
          $$
          \steps{\app{\parens{\elambda{\evar}{\expr_0}{}}}{\expr_2}}{\app{ (\elambda{\evar}{\expr_0}{})  }{\expr_2'}}
          $$
      \end{enumerate}

    \item There exists $\expr_1'$ such that:
      \begin{align}
        \steps{\expr_1}{\expr_1'} \label{eq:safety:67}
      \end{align}
      By \opsrcevalctx:
      $$
      \steps{\app{\expr_1}{\expr_2}}{\app{\expr_1'}{\expr_2}}
      $$

  \end{enumerate}

  In all cases, there exists $\expr'$ such that:
  \begin{align}
    \steps{\expr}{\expr'} \label{eq:safety:160}
  \end{align}
  By Theorem~\ref{theorem:rev:consistency} on \ref{eq:safety:1} and
  \eqref{eq:safety:160}, there exists $\texpr'$ such that:
  \begin{align}
    \stepsmanyone{\texpr}{\texpr'}   \notag \\
    \trans{\emp}{\expr}{\typeb}{}{\texpr'}
    \notag
  \end{align}
  And by Corollary~\ref{corollary:multi:preservation}:
  $$
  \lqcheck{\emp}{\texpr'}{\rtype}
  $$

\item Case \tchkup/\lqchkinj:
  \textit{Following similar methodology as before.}

\item Case \tchkdown/\lqchkcase:
  \textit{Following similar methodology as before.}

\item Case \tchkdead/\lqchkapp: 
 
  Corollary~\ref{corollary:deadcast:invalid} contradicts the second
  premise~\ref{eq:safety:2}, so the theorem does not apply here.

\end{enumerate}

\end{proof}

}
{

  \begin{proof}

  By induction on pairs of derivations:
  ${\trans{\tcenv}{\expr}{\type}{}{\texpr}}$ and 
  ${\lqcheck{\lqenv}{\texpr}{\rtype}}$.
  Details can be found in the accompanying report.
  \end{proof}

}

\end{document}